\RequirePackage{snapshot}
\documentclass[a4paper,11pt]{article}
\usepackage{booktabs} 
\usepackage{fullpage}
\usepackage{amsmath,amssymb,amsthm}
\usepackage{mathtools}
\allowdisplaybreaks

\usepackage{framed}
\usepackage{xspace}
\usepackage{hyperref}
\usepackage{color}
\usepackage{bbm}
\usepackage{bm}
\usepackage{tikz}
\usepackage[boxed]{algorithm2e}
\usepackage{dsfont}

\usepackage{float}
\floatstyle{boxed}
\newfloat{pseudocode}{thb}{pseudo}
\floatname{pseudocode}{Pseudocode}

\DeclareMathOperator*{\tr}{tr}

\date{}

\title{Average whenever you meet: \\
Opportunistic  protocols for community detection}

\author{Luca Becchetti \\
\scriptsize{Sapienza Università di Roma} \\
\scriptsize{Rome, Italy}\\ 
\footnotesize{\texttt{becchetti@dis.uniroma1.it}}
\and
Andrea Clementi \\ 
\scriptsize{Università di Roma Tor Vergata} \\
\scriptsize{Rome, Italy} \\
\footnotesize{\texttt{clementi@mat.uniroma2.it}}
\and
Pasin Manurangsi \\
\scriptsize{U.C. Berkeley} \\
\scriptsize{Berkeley, California} \\
\footnotesize{\texttt{pasin@berkeley.edu}}
\and
Emanuele Natale \\
\scriptsize{Max Planck Institute for Informatics} \\
\scriptsize{Saarbrücken, Germany} \\
\footnotesize{\texttt{enatale@mpi-inf.mpg.de}}
\and
Francesco Pasquale \\
\scriptsize{Università di Roma Tor Vergata} \\
\scriptsize{Rome, Italy} \\
\footnotesize{\texttt{pasquale@mat.uniroma2.it}}
\and
Prasad Raghavendra \\
\scriptsize{U.C. Berkeley} \\ 
\scriptsize{Berkeley, CA, United States} \\
\footnotesize{\texttt{raghavendra@berkeley.edu}}
\and
Luca Trevisan \\
\scriptsize{U.C. Berkeley} \\ 
\scriptsize{Berkeley, CA, United States} \\
\footnotesize{\texttt{luca@berkeley.edu}}
}

\newtheorem{definition}{Definition}[section]
\newtheorem{lemma}{Lemma}[section]
\newtheorem{obs}{Observation}[section]
\newtheorem{theorem}[lemma]{Theorem}
\newtheorem{cor}{Corollary}[section]
\newtheorem{fact}{Fact}[section]
\newtheorem{claim}{Claim}
\newtheorem{prop}{Proposition}
\newtheorem{remark}{Remark}

\newcommand{\bigO}{\mathcal{O}}

\newcommand{\planted}{\mathcal{G}_{n,p,q}}

\newcommand{\Prob}[2]{\mathbb{P}_{#1} \left[ #2 \right]}

\renewcommand{\Pr}[1]{\mathbb{P} \left[ #1 \right]}
\newcommand{\Expec}[2]{\mathbb{E}_{#1}\! \left[ #2 \right]}

\newcommand{\E}{\mathop{\mathbb{E}}}

\newcommand{\skproof}{\noindent\textit{Sketch of the Proof. }}

\newcommand{\beq}{\begin{equation}}
\newcommand{\eeq}{\end{equation}}

\newcommand{\ind}{\mathds{1}}

\newcommand{\Z}{\mathbb{Z}}
\newcommand{\geqs}{\geqslant}
\newcommand{\leqs}{\leqslant}
\newcommand{\cE}{\mathcal{E}}
\newcommand{\N}{\mathbb{N}}
\newcommand{\cM}{\mathcal{W}}
\def\Pr{\mathop{\mathbb P}}
\newcommand{\tend}{\text{e}}

\newcommand{\stored}{\text{s}}
\newcommand{\ttau}{\tilde{\tau}}

\newcommand{\uni}{\text{uniform}}
\newcommand{\initial}{\text{initial}}
\newcommand{\midv}{\text{ }\middle\vert\text{ }}
\newcommand{\incorrect}{\text{incorrect}}
\newcommand{\oR}{\bar{R}}

\newcommand{\aveproc}{\mbox{\sc Averaging}\xspace}
\newcommand{\labelsign}{\mbox{\sc Sign-Labeling}\xspace}
\newcommand{\labeljump}{\mbox{\sc Jump-Labeling}\xspace}

\newcommand{\CSL}{{\sc CSL}\xspace}
\newcommand{\dimsignature}{\ell}

\newcommand{\avg}{\mbox{\sc Averaging}\xspace}

\newcommand{\comment}[1]{\texttt{\textcolor{red}{[#1]}}}

\newcommand{\pasin}[1]{\comment{pasin: #1}}

\newcommand{\sgn}{\ensuremath{\mathbf{sgn}}}
\newcommand{\puki}{\ensuremath{p}} 
\newcommand{\hamm}{\ensuremath{\Delta}}

\newcommand{\lapl}{\mathcal{L}}

\newcommand{\norm}[1]{\lVert#1\rVert}

\def\bx{{\bf x}}

\def\by{{\bf y}}
\def\bv{{\bf v}}
\def\bw{{\bf w}}
\def\bz{{\bf z}}
\def\be{{\bf e}}
\def\bh{{\bf h}}
\def\bfied{{\bf f}}

\def\bchi{{\bm\chi}}
\def\bone{{\mathbf 1}}
\def\bzero{{\mathbf 0}}

\newcommand{\hsign}{\bh^{\mbox{sign}}_u\xspace}

\newcommand{\psign}{\puki^{sign}_u\xspace}

\def\R{{\mathbb R}}

\renewcommand{\le}{\leqslant}
\renewcommand{\leq}{\leqslant}
\renewcommand{\ge}{\geqslant}
\renewcommand{\geq}{\geqslant}
\renewcommand{\epsilon}{\varepsilon}

\newcommand{\avgW}{\overline{W}}

\begin{document}

\maketitle
\thispagestyle{empty}

\begin{abstract}
Consider the following asynchronous, opportunistic communication model over a
graph $G$: in each round, one edge is activated uniformly and independently at
random and (only) its two endpoints can exchange messages and perform local
computations. Under this model, we study the following random process:
\emph{The first time a vertex is an endpoint of an active edge, it chooses a
random number, say $\pm 1$ with probability $1/2$; then, in each round, the two
endpoints of the currently active edge update their values to their average.}
We show that, if $G$ exhibits a two-community structure (for example, two
expanders connected by a sparse cut), the values held by the nodes will
collectively reflect the underlying community structure over a suitable phase
of the above process, allowing efficient and effective recovery in important
cases.

In more detail, we first provide a first-moment analysis showing that, for a
large class of almost-regular clustered graphs that includes the
\emph{stochastic block model}, the expected values held by all but a negligible
fraction of the nodes eventually reflect the underlying cut signal. We prove
this property emerges after a ``mixing'' period of length $\bigO(n\log n)$. We
further provide a second-moment analysis for a more restricted class of regular
clustered graphs that includes the \emph{regular stochastic block model}. For
this case, we are able to show that most nodes can efficiently and locally
identify their community of reference over a suitable time window.  This
results in the first opportunistic protocols that approximately recover
community structure using only logarithmic (or polylogarithmic, depending on
the sparsity of the cut) work per node.  Even for the above class of regular
graphs, our second moment analysis requires new concentration bounds on the
product of certain random matrices that are technically challenging and
possibly of independent interest.


\end{abstract}




\section{Introduction}
\label{sec:intro}

Consider the following, elementary distributed process on an undirected graph
$G = (V,E)$ with $|V| = n$ nodes and $|E| = m$ edges. Each node $v$ holds a
real number $x_v$ (which we call the {\em state} of node $v$); at each time
step, one random edge $\{u,v\}$ becomes active and its endpoints $u$ and $v$
update their states to their average. 

Viewed as a protocol, the above process is consistent with asynchronous,
opportunistic communication models, such as those considered in~\cite{AAER07}
for \emph{population-protocols}; here, in every round,  one edge is activated
uniformly and independently at random and (only) its two endpoints can exchange
messages and perform local computations in that round\footnote{In an
essentially equivalent continuous-time model, each edge has a clock that ticks
at random intervals with a Poisson distribution of average 1; when the clock
ticks, then the edge endpoints  become activated. For $t$ larger than $n\log
n$, the behavior of the continuous time process for $t/n$ units of time and the
behavior of the discrete-time process for $t$ steps are roughly equivalent.}.
We further assume no global clock is available (nodes can at most count the
number of local activations) and  that the network is {\em anonymous}, i.e.,
nodes are not aware of theirs or their neighbors' identities and all nodes run
the same  process at all times.

The long-term behavior of the process outlined above is well-understood:
assuming $G$ is connected, for each initial global state $\bx\in \R^V$  the
system converges to a global state in which all nodes share a common value,
namely, the average of their initial states. A variant of an argument of Boyd
et al.~\cite{BGPS06} shows that convergence occurs in $\bigO \left(
\frac{1}{\lambda_2} n \log n \right)$ steps, where $\lambda_2$ is the second
smallest eigenvalue of the normalized Laplacian of $G$. 

Suppose now that $G$ is  \emph{well-clustered}, i.e. it exhibits a {\em
community structure} which in the simplest case consists of two equal-sized
expanders, connected by a sparse cut: This structure arises, for instance, when
the graph is sampled from the popular \emph{stochastic block
model}~\cite{MNS14} $\planted$ for $p \gg q$ and $p \geqslant \log n/n$.  If we
let the averaging process unfold on such a graph, for example starting from an
initial $\pm 1$ random global state, one might reasonably expect a faster,
transient convergence toward some local average within each community,
accompanied by a slower, global convergence toward the average taken over the
entire graph. If, as is likely the case, a gap exists between the local
averages of the two communities, the global state during the transient phase
would reflect the graph's underlying community structure.  This intuition
suggests the main questions we address in this paper:


\medskip\noindent
Is there {\em  a phase in which the global state carries information about
community structure}? If so, {\em how strong is the corresponding  ``signal''}?
Finally, {\em can nodes leverage local history to recover this information}?

\medskip

The idea of using averaging local rules to perform distributed community
detection is not new: In~\cite{BCNPT17}, Becchetti et al.  consider a
deterministic dynamics in which, at every round, each node updates its local
state to the average of its neighbors. The authors show that this results in a
fast clustering algorithm with provable accuracy on a wide class of
almost-regular graphs that includes the stochastic block model. We remark that
the  algorithm in~\cite{BCNPT17} works in a synchronous, parallel
communication model where every node exchanges data with all its neighbors in
each round. This implies considerable work and communication costs, especially
when the graph is dense. On the other hand, each step of the process is
described by the same matrix and its evolution unfolds according to the power
of this matrix applied to the initial state. In contrast, the averaging
process we consider in this paper is considerably harder to analyze than the
one in~\cite{BCNPT17}, since each step is described by a  random, possibly
different averaging matrix.
 
Differently from~\cite{BCNPT17}, our goal here is the design of simple,
lightweight protocols for fully-decentralized community detection which fit the
asynchronous, opportunistic communication model, in which a (random) link
activation represents  an opportunistic meeting that the endpoints can exploit
to  exchange one-to-one messages. More specifically, by ``lightweight'' we mean
protocols that require minimalistic assumptions as to network capabilities,
while performing their task with minimal work, storage and communication per
node (at most logarithmic or polylogarithmic in our case). In this respect, any
clustering strategies (like the one in~\cite{SZ17}) which construct (and then
work over) some static, sparse subgraph of the underlying graph are unfeasible
in the opportunistic model we consider here. This restrictive setting is
motivated by network scenarios in which individual agents need to autonomously
and locally uncover underlying, implicit communities of which they are members.
This has widespread applicability, for example in communication systems where
lightweight data can be locally shared via wireless opportunistic meetings when
agents come within close range~\cite{WWA12}.
    
We next discuss what it means to recover the ``underlying community structure''
in a distributed setting, a notion that can come in stronger or weaker flavors.
Ideally, we would like the protocol to reach a state in which, at least with
high probability, each node can use a simple rule to assign itself one of two
possible labels, so that labelling within each community is consistent and
nodes in different communities are assigned different labels. Achieving this
corresponds to {\em exact (block) reconstruction}. The next best guarantee is
{\em weak (block) reconstruction} (see Definition~\ref{def:weak}). In this
case, with high probability the above property is true for all but a small
fraction of misclassified nodes. In this paper, we introduce a third notion,
which we call {\em community-sensitive labeling} (\CSL for short): in this
case, there is a predicate that can be applied to pairs of labels so that, for
all but a small fraction of outliers, the labels of any two nodes within the
same community satisfy the predicate, whereas the converse occurs when they
belong to different communities\footnote{Note that a weak reconstruction
protocol entails a community-sensitive labeling. In this case, the predicate is
true if two labels are the same.}. In this paper, informally speaking, nodes
are labelled with binary signatures of logarithmic  length, while two labels
satisfy the predicate whenever their Hamming distance is  below a certain
threshold. This introduces a notion of similarity between nodes of the graph,
with labels behaving like profiles that reflect community
membership\footnote{Hence the phrase  \emph{community-sensitive Labeling} we
use to refer to our approach.}. Note that this  weaker notion of
community-detection allows nodes to locally tell ``friends'' in their community
from ``foes'' in the other community, which is the main application of
distributed community detection in the opportunistic setting we
consider here.

\subsection{Our results} \label{ssec:ours}

\noindent \textbf{First moment analysis.} Our first contribution is an analysis
of the expected evolution of the averaging process over a wide class of
almost-regular graphs (see Definition~\ref{def:regulargood}) that possess a
hidden and balanced partition of the nodes with the following properties: (i)
The cut separating the two communities is sparse, i.e., it contains $o(m)$
edges; (ii) the subgraphs induced by the two communities are  expanders, i.e.,
the gap $\lambda_3 - \lambda_2$ between the third and the second eigenvalues of
the normalized  Laplacian matrix $\lapl$ of the graph is constant. The above
conditions on the underlying graph are satisfied, for instance, by graphs
sampled from the  stochastic block model\footnote{See
Subsection~\ref{ss:related} for the definition of $\planted$ and for  more
details about our results for $\planted$.} $\planted$ for $q=o(p)$ and $p
\geqslant \log n/n$.

Let $L = D - A$ be the  Laplacian matrix of $G$.  The first moment analysis
considers the {\em deterministic} process  described by the  linear equation
$\bx^{(t+1)} = \avgW^t \cdot \bx^{(0)}$ ($t \geq 1$), where $\bx^{(0)} = \bx$
is the   vector with components the nodes' initial  random values and $\avgW :=
\Expec{}{W} = I - \frac{1}{2m} L$ is the expectation of the random matrix that
describes a single step of the averaging process.  While a formal proof of the
above equation can be found in Section \ref{se:expec}, our analysis reveals
that the expected values held by the nodes are correlated with the underlying
cut. This phenomenon follows from structural connections between the underlying
graph's community structure and some spectral properties of $\avgW$. This
allows us to show that, after an initial \textit{``mixing''} phase of $\Theta(n
\log n)$ rounds and for all but $o(n)$ nodes, the following properties hold:
(i) There exists a relatively large time window in which the signs of the
expected values of nodes are correlated with the community they belong to.
(ii) The expected values of nodes belonging to one of the communities increase
in each round, while those of nodes in the other community decrease.

The formal statements of the above claims can be found in
Theorem~\ref{thm:mainavg}. Here, we note that these results suggest two
different local criteria for community-sensitive labeling: (i) According to the
first one, every node uses the sign of its own state within the aforementioned
time window to set the  generic component of its binary label (we in fact use
independent copies of the averaging process to  get  binary labels of
logarithmic size - see Protocol~\labelsign{} in Section~\ref{se:label}).  (ii)
According to the second criterion, every node uses the signs of fluctuations of
its own value along consecutive rounds to set the generic component of its
binary label   (see Protocol~\labeljump{} in
Section~\ref{sec:reconstruct-analysis}) \footnote{Having a node set its label
within the correct time window is technically challenging in the asynchronous,
opportunistic communication model we consider. This issue is briefly discussed
in this and the following sections and formally addressed in Appendix
\ref{apx:toolsforcslsign}).}.

The above analysis describes the ``expected'' behaviour of the averaging
process over a large class of  well-clustered graphs, at the same time showing
that our approach might lead to efficient, opportunistic protocols for block
reconstruction. Yet, designing and analyzing protocols with provable,
high-probability guarantees, requires addressing the following questions:

\medskip


\begin{enumerate}

\item {\em Do realizations of the averaging process approximately 
follow its expected behavior with high, or even constant, 
probability?}
\item {\em If this is the case, how can nodes locally and 
asynchronously recover the cut signal, let alone guess
the ``right'' global time window?}
\end{enumerate}

\medskip
\noindent \textbf{Second moment analysis.} The first question above essentially
requires a characterization of the \emph{variance} of the process over time,
which turns out to be an extremely challenging task. The main reason for this
is that a realization of the averaging process is now described by an equation
of the form $\bx^{(t )} = W_t \cdot .....\cdot W_1\bx$, where the $W_i$'s are
sampled independently and uniformly  from some matrix distribution (see
Eq.~\eqref{randommat} in Section~\ref{se:expec}). Here, matrix $W_i$
``encodes'' both the $i$-th edge selected for activation and the averaging of
the values held by its endpoints at the end of the $(i-1)$-th step. 

Not much is known about concentration of the products of identically
distributed random matrices, but we are able to accurately characterize the
class of regular clustered graphs. We point out that many of the technical
results and tools we develop to this purpose apply to far more general settings
than the regular case and may be of independent interest. In more detail, we
are able to provide accurate upper bounds on the norm of $\bx^{(t )}$'s
projection onto the subspace spanned by the first and second eigenvector of
$\avgW$ (see the proof's outline of Theorem~\ref{thm:mainavg} and
Lemma~\ref{lm:f-bound}) for a class of regular clustered graphs that includes
the \emph{regular stochastic block model}\footnote{See
Subsection~\ref{ss:related} for more details about our results for regular
stochastic block models.}
\cite{BCNPT17,brito_recovery_2015,mossel_reconstruction_2014} - see
Definition~\ref{def:clusteredregulargood}. 

These bounds are derived separately for two different regimes, defined by the
sparseness of the cut separating the two communities. Assuming a good inner
expansion of the communities, the first concentration result concerns cuts of
size $o(m/\log^2 n)$ and it is given in Subsection \ref{ssec:secmom} while,
for  the case of cuts of size up to $\alpha m$ for any $\alpha < 1$, the
obtained concentration results are described in
Subsection~\ref{app:proofthm:main}.

These bounds alone are not sufficient to prove accuracy of the clustering
criteria, due to the asynchronous nature of the communication model we
consider, whereby every node only updates its value once every $m/d$ rounds
in expectation, where $d$ is the degree of the nodes.


\medskip
\noindent \textbf{Desynchronization.} The variance analysis outlined in the
previous paragraph ensures that, for any fixed step $t$, the actual states of a
large fraction of the nodes are ``close'' to their expectations, with high
probability. Unfortunately, the asynchrony of the model we consider does not
allow to easily apply this result (e.g., using a union bound) to prove that
most nodes will eventually label themselves within the right global window and
in a way that is consistent with the graph's community structure. Rather, we
show that there exists a large fraction of \textit{non-ephemeral} ``good
nodes'' whose states remain close to their expectations over a suitable
time-window. The technical form of the concentration bound and the relative
time window again depend on the sparsity of the cut: See
Definition~\ref{def:ephemeral} and Lemma~\ref{lem:main} for sparse cuts  and
Theorem~\ref{thm:main-concen} for dense cuts, respectively. 


\medskip
\noindent \textbf{Distributed community detection.} We exploit the second
moment analysis and the desynchronization above to devise two different
opportunistic protocols for community detection on regular clustered graphs.

\noindent
- In the case of sparse cuts (i.e. of size $o(m/\log^2 n)$), the obtained
bound on the variance of  non-ephemeral nodes (see Lemma~\ref{lem:main})
holds over a time window that essentially equals the one ``suggested'' by our
first moment analysis. This allows us to give rigorous bounds on the
performance of the opportunistic Protocol~\labelsign{} based on the sign
criterion (see Section~\ref{se:label}). This ``good'' time-window 
begins after $\bigO(n\log n )$ rounds: So, if the underlying graph has dense communities and
a sparse cut, nodes can collectively compute an accurare labeling \emph{before} the global mixing
time of the graph. For instance, if the cut is $\bigO(m/n^{\gamma})$, for
some constant $\gamma < 2$, our protocol is polynomially faster than the
global mixing time. In more detail, we prove that, given any regular clustered
graph with cut of size $o(m/\log^2 n)$, Protocol~\labelsign{} performs 
community-sensitive labeling for $n-o(n)$ nodes within global time\footnote{The
extra logarithmic factor is needed to let every node update each component of
its $\Theta(\log n)$-size binary label, independently (see
Subsection~\ref{se:label} for details).} $\bigO(n \log^2 n)$ and with 
work per node $\bigO(\log^2n)$, with high probability (see Theorem~\ref{thm:main_small}
and its Corollary~\ref{cor:csl-small} for formal and more general statements
about the performances of the protocol). Importantly enough, the 
costs of our first protocol do not depend on the cardinality of the edge set $E$.
  
\noindent
-  The bound on the variance  that allows us to adopt the sign-based criterion
   above does not hold when the cut is not sparse, i.e., whenever it is
$\omega(m/\log^2 n)$. For such dense cuts, we use a different bound on 
the variance of nodes' values  given in Theorem~\ref{thm:main-concen}, which
starts to hold after the global mixing  time of the underlying graph 
and over a time window of length $\Theta(n^2)$. In this case, the specific form of the
concentration bound  leads to adoption of the second clustering criterion
suggested by our first moment analysis, i.e., the one based on monotonicity
of the values of non-ephemeral nodes. To this aim, we consider a ``lazy''
version of the averaging process equipped with a local clustering 
criterion, whereby nodes use the signs of fluctuations of their own 
values along consecutive rounds to label themselves (see Protocol~\labeljump{} in
Section~\ref{sec:reconstruct-analysis}). A restricted but relevant version of
the algorithmic result we achieve in this setting can be stated as follows (see
Theorem~\ref{thm:main-reconstruct} and its corollaries for more general
statements): Given any regular clustered graph consisting of two expanders as
communities and a cut of size up to $\alpha m$ (for any $\alpha < 1$), the
opportunistic protocol~\labeljump{} achieves weak reconstruction for a fraction
$(1-\epsilon)n$ nodes (where $\epsilon$ is an arbitrary positive constant),
with high probability. The protocol converges within $\bigO(n \log n \, (\log^2
n + m / m_{1,2}))$ rounds (where we named $m_{1,2} = |E(V_1, V_2)|$ the size of
the cut) and every node performs $\bigO(\log^2 n + (m / m_{1,2})\log n)$ work,
with high probability. Notice that this second protocol achieves a stronger
form of community detection than the first one, but it is less 
efficient, especially when the underlying  graph is dense and exhibits a sparse cut (it requires to ``wait''
for the global mixing time of the graph). On the other hand, the two protocols have 
comparable costs in the parameter ranges they were designed 
for.\footnote{As for the fraction of outliers guaranteed by the two
protocols, please see the technical discussion after
Corollary~\ref{cor:main2}.}

\subsection{Comparison to previous work} \label{ss:related}

We earlier compared our results to those
of~\cite{BCNPT17}. The advantage of~\cite{BCNPT17} is that their 
analysis achieves concentration over a class of graphs that are
almost-regular and extends to the case of more than two communities.
Furthermore, in graphs in which the indicator of the cut is
an eigenvector of $\lambda_2$, the algorithm of~\cite{BCNPT17} achieves exact
reconstruction. On the other hand, as previously remarked, the advantage of our
work over~\cite{BCNPT17} is that, for the first time, it applies to the
asynchronous opportunistic model and the communication cost per node does not
depend on the degree of the graph, so it is much more efficient in dense graphs.

If we do not restrict to asynchronous and/or opportunistic protocols, recently,
in~\cite{SZ17}, Sun and Zanetti\footnote{Before the technical report
in~\cite{SZ17}, the same authors in~\cite{SZ16} presented a synchronous
distributed algorithm able to perform approximate reconstruction with multiple
communities. Sun and Zanetti then discovered a gap in their analysis (personal
communication), and they retracted the claims in~\cite{SZ16}.} introduced a
synchronous, averaging-based protocol that first computes a fixed random
subgraph of the underlying graph and then, working on this sparse subgraph, 
returns an efficient block-reconstruction for a wide class of almost-regular
clustered graphs including the stochastic block model. We remark that, besides
having no desynchronization issue to deal with, their second moment analysis
uses Chernoff-like concentration bounds on some random submatrix (rather than a
product of them, as in our setting) which essentially show that, under 
reasonable hypothesis, the signal of the cut can still be recovered from the corresponding
sparse subgraph the algorithm works on.
 
Further techniques for community detection and spectral clustering exist, which
are not based on averaging. In particular, Kempe and McSherry showed that the
top $k$ eigenvectors of the adjacency matrix of the underlying graph can be
computed in a  distributed fashion~\cite{kempe2004decentralized}.  These
eigenvectors can then be used to partition the graph; in our settings, since we
assume that the indicator of the cut is the second eigenvector of the graph,
applying Kempe and McSherry's algorithm with $k = 2$ immediately reveals the
underlying partition. Again, we note here that the downside of this algorithm
is that it is synchronous and quite complex. In particular, the algorithm
requires a computation of $Q_1 \bx$, for which $\lambda_2^{-1} \log n$ work per
node is a bottleneck, while our first algorithm only requires  $\lambda_3^{-1}
\log n$ work per node, a difference that can become significant for very sparse
cuts.

At a technical level, we note that our analysis establishes concentration
results for products of certain i.d.d. random matrices; concentrations of such
products have been studied in the ergodic theory
literature~\cite{crisanti_products_1993, ipsen_products_2015}, but under
assumptions that are not met in our setting, and with convergence rates that
are not suitable for our applications.


While we only focused on decentralized settings so far, we note that the
question of community detection, especially in stochastic block models, has
been extensively studied in  centralized  computational
models~\cite{abbe2014exact, coja2010graph, decelle_asymptotic_2011,
dyer1989solution, holland1983stochastic, jerrum_metropolis_1998,
mcsherry2001spectral}.  The stochastic block model offers a popular framework
for the probabilistic modelling of graphs that exhibit good clustering or
community properties. In its simplest version, the random graph $\planted$
consists of $n$ nodes and an edge probability distribution defined as follows:
The node set is partitioned into two subsets $V_1$ and $V_2$, each of size
$n/2$; edges linking nodes belonging to the same partition appear in $E$
independently at random with probability $p = p(n)$, while edges connecting
nodes from different partitions appear with probability $q = q(n) < p$. In the
centralized setting, the focus of most studies on stochastic block models is on
determining the threshold at which weak recovery becomes possible, rather than
simplicity or running time of the algorithm (as most algorithms are already
reasonably simple and efficient). After a remarkable line of
work~\cite{decelle_asymptotic_2011, mossel_reconstruction_2014,
massoulie_community_2014, mossel_proof_2013}, such a threshold has now been
precisely determined.

Calling $a = pn$ and $b = qn$, it is known~\cite{BCNPT17} that graphs sampled
from $\planted$ satisfy (w.h.p.) the approximate regularity and spectral gap
conditions required by our first moment analysis (i.e.
Theorem~\ref{thm:mainavg}) whenever $b = o(a)$ and $a = \Omega(\log n)$.
Versions of the stochastic block model in which the random graph is regular
have also been considered~\cite{mossel_reconstruction_2014,
brito_recovery_2015}. In particular Brito~et~al.~\cite{brito_recovery_2015}
show that strong reconstruction is possible in polynomial-time when $a-b >
2\sqrt{a+b-1}$. As for these regular random graphs, we remark that our
opportunistic protocol for sparse cut works whenever $a/b \geqslant\log^2 n$
(see also~\cite{B15, brito_recovery_2015}), while our protocol for dense cuts
works for any parameter $a$ and $b$ such that $a - b > 2(1 + \rho)\sqrt{a +
b}$, where $\rho$ is any positive constant. Since it is
(information-theoretically) impossible to reconstruct the graph when $a - b
\leqs \bigO(\sqrt{a + b})$~\cite{mossel_reconstruction_2014}, our result comes
within a constant factor of this threshold.

\subsection{Roadmap of the paper} \label{sec:road}
After presenting some preliminaries in Section \ref{sec:avg}, the first moment analysis 
for almost-regular graphs is given in Section \ref{se:expec}.
The   analysis of the variance of the 
averaging process in regular graphs for the case of sparse cuts and the analysis of the resulting sign-based protocol  are described in 
Section \ref{sec:second_moment}. In Section \ref{ssec:pl-analysis}, we 
address the case of dense cuts:  Similarly to the previous section, we first  give a second moment analysis and then   show how to 
apply it   to devise a suitable opportunistic protocol for this regime.

 Due to the considerable length of this paper,  most of the  technical results are given in a 
separate appendix.

\section{Preliminaries} \label{sec:avg}

We study the weighted version of the Averaging process described in the 
introduction. In each round, one edge of the graph is sampled uniformly at random
and the two endpoints of the sampled edge execute the following algorithm.

\vspace{-.3cm}

\begin{algorithm}\label{alg:aveproc}
\LinesNotNumbered
\DontPrintSemicolon
\SetAlgoVlined
\SetNoFillComment
\aveproc{$(\delta)$} \,
(for a node $u$ that is one of the two endpoints of an active edge)\;
\begin{description}
\setlength\itemsep{0mm} 
\item[Initialization:]
If it is the first time $u$ is active, then pick $\bx_u \in \{-1, +1\}$ u.a.r.
\item[Update:] Send $\bx_u$ to the other endpoint of the active edge \\ and then 
update  $\bx_u := (1-\delta) \bx_u + \delta r$,
where $r$ is the \\ value received from the other endpoint.
\end{description}
\caption{\label{algo:sparse_update} Updating rule for a node $u$
of an active edge, where $\delta \in (0,1)$
is the parameter measuring the weight given to the neighbor's value}
\end{algorithm}

\vspace{-.3cm}

For a graph $G$ with $n$ nodes and adjacency matrix $A$, let $0 = \lambda_{1} 
\leqs \cdots \leqs \lambda_{n}$ be the eigenvalues of the normalized Laplacian 
$\lapl = I - D^{-1/2}AD^{-1/2}$, where $D$ is the diagonal matrix with the degrees of the nodes.
We consider the following classes  of graphs.

\begin{definition}[Almost-regular graphs]
\label{def:regulargood} An $(n,d,\beta)$-almost-regular 
graph\footnote{This class is more general than the one
introduced in \cite{BCNPT17}, since there is no regularity constraint on
the outer node degree, i.e., on the  number of edges a node can have
towards  the other community.} $G = (V,E)$ is a connected, non-bipartite graph over 
vertex set $V$,  such that
every node has degree $d \pm \beta d$. 
\end{definition}


\begin{definition}[Clustered regular graphs]
\label{def:clusteredregulargood}
Let $n \geq 2$ be an even integer and $d$ and $b$ two positive integers such 
that $2 b < d < n$.  An $(n,d,b)$\emph{-clustered regular graph}
    $G = ((V_1,V_2),E)$ is a graph over node set $V = V_1 \cup V_2$, with 
    $|V_1|=|V_2|=n/2$ and such that: (i) Every node  has degree $d$ and (ii) Every
    node in $V_1$ has $b$ neighbors in $V_2$ and every node in $V_2$ has $b$ 
    neighbors in $V_1$.
\end{definition}

We remark that if a graph is clustered regular then we easily get that   the
indicator vector $\bchi$ of the cut $(V_1, V_2)$ is an eigenvector of $\lapl$
with eigenvalue $\frac {2b}d$; If we further assume that $\lambda_3 >
\frac{2b}d$, then $\bchi$ is an eigenvector of $\lambda_2$. We next recall
the notion of \emph{weak reconstruction} \cite{BCNPT17}.

\begin{definition}[Weak Reconstruction]
\label{def:weak}
A function $f: V \rightarrow \{\pm 1\}$ is said to be an
\emph{$\varepsilon$-weak reconstruction} of $G$ if subsets
$W_1 \subseteq V_1$ and $W_2 \subseteq V_2$ exist, each of size at least $(1 -
\varepsilon)n/2$, such that $f(W_1) \cap f(W_2) = \emptyset$. 
\end{definition}

In this paper, we introduce a weaker notion of distributed community detection.  
Namely, let  $\hamm(\mathbf{x}, \mathbf{y})$  denote the  {\em 
Hamming distance} between two binary strings $\mathbf{x}$ and $\mathbf{y}$.
\begin{definition}[Community-sensitive labeling] 
\label{def:senshash}
Let $G = (V, E)$  be a  graph, let $(V_1,V_2)$ be a partition of $V$ and let
$\gamma \in (0,1]$. For some $m \in \mathbb{N}$, a function $\bh\,:\, V_1 \cup
V_2 \rightarrow \{0,1\}^m$ is a $\gamma$-\emph{community-sensitive labeling}
for $(V_1,V_2)$ if a subset $\tilde{V} \subseteq V$ with size
$|\tilde{V}| \geqslant (1-\gamma) |V|$ and two constants $0 \leqslant c_1 < c_2
\leqslant 1$ exist, such that for all $u,v \in \tilde{V}$ it holds that
\[
\hamm(\bh_u,\bh_v)
\left\{
\begin{array}{cll}
\leqslant c_1 m & \mbox{ if } i_u = i_v  \ &  \mbox{\emph{(Case (i))},}\\[2mm]
\geqslant c_2 m & \mbox{ otherwise} \ & \mbox{\emph{(Case (ii))},}
\end{array}
\right.
\]
where $i_u = 1$ if $u \in V_1$ and $i_u = 2$ if $u \in V_2$.
\end{definition}


\section{First Moment Analysis}\label{se:expec}


\medskip In this section, we analyze the expected behaviour of Algorithm
\aveproc{$(1/2)$} on an almost-regular graph $G$ (see
Definition~\ref{def:regulargood}). The evolution of the resulting process can
be formally described by the recursion $\bx^{(t+1)} = W_t \cdot \bx^{(t)}$,
where $W_t = (W_t(i,j))$ is the random matrix that defines the updates of the
values at round $t$, i.e.,
\begin{equation}\label{randommat}
W_t(i,j) = 
\left\{
\begin{array}{cl}
0 & {\mbox{ if } i \neq j \mbox{ and } \{i,j\} \mbox{ is not sampled (at round $t$),}} \\[2mm]
1/2 & 
\begin{array}{l}
{\mbox{if $i = j$ and an edge with endpoint $i$ is sampled}} \\ 
{\mbox{or $i \neq j$ and edge $\{i,j\}$ is sampled,}}
\end{array}\\[3mm] 
1 & {\mbox{ if } i = j \mbox{ and } i \mbox{ is not an endpoint of sampled edge.}}
\end{array}
\right.
\end{equation}

\noindent
and the initial random vector $\bx^{(0)}$ is uniformly distributed in
$\{-1,1\}^n$.\footnote{Notice that, since each node chooses value $\pm 1$ with
probability $1/2$ the first time it is active, by using the principle of
deferred decisions we can assume there exists an ``initial'' random vector
$\bx^{(0)}$ uniformly distributed in $\{-1,+1\}^n$.} 

Notice that random matrices $\{ W_t \,:\, t \geqslant 0 \}$ are independent and
identically distributed and simple calculus shows that their expectation can be
expressed as (see Observation~\ref{obs:expecW} in the Appendix):
\begin{equation}\label{eq:expectedW}
\avgW := \Expec{}{W_t} = I - \frac{1}{2m} L \, , \, 
\end{equation}
where $L = D - A$ is the Laplacian matrix of $G$. Matrix $\avgW$ is thus 
symmetric and doubly-stochastic. We denote its eigenvalues as
$\bar{\lambda}_1, \dots, \bar{\lambda}_n$, with 
$
1 = \bar{\lambda}_1 \geqslant \bar{\lambda}_2 \geqslant \cdots \bar{\lambda}_n \geqslant -1\,.
$

We next provide a first moment analysis for $(n, d, 
\beta)$-almost regular graphs that exhibit a clustered structure.  Our analysis proves the
following results.

\begin{theorem}\label{thm:mainavg}
Let $G = (V,E)$ be an $(n, d, \beta)$-almost regular graph $G = (V,E)$ with a
balanced partition $V = (V_1,V_2)$ and such that: (i) The cut  $E(V_1,V_2)$ is
sparse, i.e.,  $m_{1,2} = |E(V_1,V_2)| = o(m)$; (ii) The gap $\lambda_3 -
\lambda_2 = \Omega(1)$.\footnote{In practice, this means that each of the
subgraphs induced by community $V_i$ ($i=1,2$) is an expander.} If nodes of $G$
execute Protocol~\aveproc{} then, with constant probability w.r.t. the initial
random vector $\bx^{(0)} \in \{-1,1\}^n$, after $\Theta(n \log n)$ rounds the
following holds for all but $o(n)$ nodes: (i) The expected value of a node $u$
increases or decreases depending on the community it belongs to, i.e.,
$\sgn\left( \Expec{}{\bx^{(t-1)}_u\,|\,\bx^{(0)}} -
\Expec{}{\bx^{(t)}_u\,|\,\bx^{(0)}} \right) = \sgn\left( \bchi_u \right)$; (ii)
Over a time window of length $\Omega(n \log n)$ the sign of the expected value
of a    node $u$ reflects the community $u$ belongs to, i.e.,  $ \sgn\left(
\Expec{}{\bx^{(t)}_u\,|\,\bx^{(0)}} \right) = \sgn\left( \alpha_2 \bchi_u
\right)$, for some $\alpha_2 = \alpha_2(\bx^{(0)})$.
\end{theorem}

\subsubsection*{Proof of Theorem~\ref{thm:mainavg}: An outline}
The proof  makes a black-box use of 
technical results that are rigorously given in   Appendix 
\ref{subsubse:lemmas}. We believe, these results are interesting in 
their own right, since they shed light on the evolution of the dynamics 
and its algebraic structure.

The hypotheses of Theorem~\ref{thm:mainavg} involve the
eigenvalues of the normalized Laplacian matrix 
$\lapl$ of the graph, while the expected evolution of the process is 
governed by matrix $\avgW$ and its eigenvalues $1 = \bar \lambda_1 \geqslant 
\cdots \bar \lambda_n \geqslant -1$  (see 
Lemma~\ref{lemma:expectedvalueimplicit}). However, $(n, d, 
\beta)$-almost regularity implies that the spectra of these two matrices are  
related. In particular, it is easy  to see that, under the hypotheses 
of Theorem~\ref{thm:mainavg}, we have (see Observation~\ref{obs:lambdarelations} in 
the Appendix)

\begin{equation}\label{eq:laplaciangapvswgap}
\frac{d}{2m}(1 - 2 \gamma)\left( \lambda_3 - 
\lambda_2 \right)
\leqslant \bar{\lambda}_2 - \bar{\lambda}_3 \leqslant 
\frac{d}{2m}(1 + 2\gamma)\left( \lambda_3 - 
\lambda_2 \right)
\end{equation}

In Lemma~\ref{lemma:expectedvalueimplicit}, we decompose 
$\Expec{}{\bx^{(t)}}$ into its components along the first two 
eigenvectors of $\avgW$ and into the corresponding orthogonal component
$\be^{(t)}$ (note that $\avgW$ admits an orthonormal eigenvector basis 
since it is symmetric). We further decompose the component along the 
second eigenvector of $\avgW$ into its component 
parallel to the partition indicator vector ($\alpha_2 \bar{\lambda}_2^t \bchi$)
and into the corresponding orthogonal one ($\alpha_2 \sqrt{n} \, \bar{\lambda}_2^t 
\bfied_\perp$). 
As a consequence, we can rewrite $\Expec{}{\bx^{(t)}}$ as 
\begin{equation}\label{eq:expecvalu}
\Expec{}{\bx^{(t)}_u} = \alpha_1 + \alpha_2 \bar{\lambda}_2^t \left[ \bchi_u 
+ \sqrt{n}\, \bfied_{\perp,u} \right]
+ \be^{(t)}_u, 
\end{equation}
where $\left\|\be^{(t)} \right\| \leqslant 
\bar{\lambda}_3^t \sqrt{n}$.
Hence, if $\alpha_2 \neq 0$ and $\bar{\lambda}_3 < \bar{\lambda}_2$, 
the term $\be^{(t)}_u$ becomes negligible w.r.t. the 
other two from some round $t$ onward. Moreover, for any node $u$ with $\bfied_{\perp,u} < 1 / 
\sqrt{n}$, $\sgn(\left[ \bchi_u + \sqrt{n}\, \bfied_{\perp,u} 
\right]) = \sgn(\bchi_u)$, i.e., $\bx^{(t)}_u$ \textit{identifies} the 
community $V_h$ node $u$ belongs to. 
Accordingly, we say a node $u \in [n]$ is $\varepsilon$-\emph{bad} if it does not satisfy the
above property (see 
Definition~\ref{def:badnodes}).\footnote{Consistently, a node is 
$\varepsilon$-\emph{good} otherwise.} Next, we derive an 
upper bound on the number of $\varepsilon$-bad nodes. To this 
purpose, we first prove an upper bound on the square norm of $\bfied_\perp$, as
a function of the gap $\bar{\lambda}_2 - \bar{\lambda}_3$ and the 
ratio between the size of the cut $m_{1,2}$ and the total number of edges in 
the graph (Lemma~\ref{lm:f-bound}). This easily implies an upper bound
on the number of $\varepsilon$-bad nodes (Corollary~\ref{cor:no_bad}) as a
function of the gap $\bar{\lambda}_2 - \bar{\lambda}_3$. 
From~\eqref{eq:laplaciangapvswgap} and the hypothesis $\lambda_3
- \lambda_2 = \Omega(1)$ the upper bound in
Corollary~\ref{cor:no_bad} turns out to be $\bigO(m_{1,2} / d)$, which in
turn is $o(n)$ under the hypothesis $m_{1,2} = o(m)$.

These results and \eqref{eq:expecvalu} imply the following conclusions 
for all $\varepsilon$-good nodes $u$: 

\noindent
(i) For all $t = \Omega(n \log n)$, the evolution of $\Expec{}{\bx^{(t)}_u}$
along two consecutive rounds identifies the block $u$ belongs to
(Lemma~\ref{lemma:incrdecr_crit}), namely: \[ \sgn\left(
\Expec{}{\bx^{(t-1)}_u} - \Expec{}{\bx^{(t)}_u} \right) = \sgn\left( \bchi_u
\right) \]

\noindent
(ii) If $|\alpha_2|$ is sufficiently larger than $|\alpha_1|$ and the second
and third largest eigenvalues of $\avgW$ satisfy appropriate conditions, for
all $t$ falling in a suitable time window, the sign of $\Expec{}{\bx^{(t)}_u}$
identifies the community node $u$ belongs to (Lemma~\ref{lemma:sign_crit}),
namely: \[ \sgn\left(\Expec{}{\bx^{(t)}_u} \right) = \sgn(\alpha_2 \bchi_u) \]

Moreover, Lemma~\ref{lm:initialrandom} implies that the initial
random vector $\bx$ satisfies the hypotheses of 
Lemma~\ref{lemma:incrdecr_crit} w.h.p. (i.e., $\alpha_2(\bx) \neq 0$ w.h.p.) and
those of Lemma~\ref{lemma:sign_crit} with constant
probability (i.e., $|\alpha_2(\bx)| \geqslant 2 
|\alpha_1(\bx)|/(1-\varepsilon)$ with constant probability). 

As a result, we can claim the 
following for any non-bad node $u$: if we consider the r.v.  
$\bh^{\mbox{jump}, (t)}_u = \sgn\left(\Expec{}{\bx^{(t-1)}_u - \bx^{(t)}_u \;|\; \bx^{(0)}}\right)$,
Lemma~\ref{lemma:incrdecr_crit} implies $\bh^{\mbox{jump},(t)}_u = \sgn(\alpha_2 
\bchi_u)$ w.h.p., for every $t$ such that 
\begin{equation}\label{eq:tbound_owv_monotcrit}
t \geqslant 3 \log \left(\frac{n}{1-\varepsilon} \right) 
/ \log(\bar{\lambda}_2 /\bar{\lambda}_3).
\end{equation}
Likewise, if we consider the r.v. $\bh^{\mbox{sign},(t)}_u = 
\sgn\left(\Expec{}{\bx^{(t)}_u} \right)$,
Lemma~\ref{lemma:sign_crit} implies 
\[
    \Prob{}{\bh^{\mbox{jump},(t)} = \sgn(\alpha_2 \bchi_u)} = \Omega(1),
\] 
for all $t$ such that
\begin{equation}\label{eq:tbound_owv_signcrit}
\frac{1}{\log(1/\bar{\lambda}_3)} \log (n / |\alpha_1|)
\leqslant t \leqslant
\frac{1}{\log (1/\bar{\lambda}_2)} \log\left(\frac{|\alpha_2|(1 - \varepsilon)}{2 
|\alpha_1|}\right).
\end{equation}
Finally, note that under the hypothesis $\lambda_3 - 
\lambda_2 = \Omega(1)$, the lower bounds on $t$ 
in~\eqref{eq:tbound_owv_monotcrit} and \eqref{eq:tbound_owv_signcrit} 
are both $\bigO(n \log n)$.


\section{Regular Graphs with a Sparse Cut  }
\label{sec:second_moment}

We   next  provide a second moment analysis of the $\aveproc(\delta)$  with $\delta = 1/2$
  on the class of $(n,d,b)$\emph{-clustered regular graphs} (see 
Definition~\ref{def:clusteredregulargood}) when the cut between the two communities is relatively 
sparse, i.e., for  $\lambda_2 = 2b/d =  
o(\lambda_3/\log n)$. This analysis is consistent with the 
``expected'' clustering behaviour of the dynamics explored in the 
previous section and highlights clustering properties that emerge well 
before global mixing time, as we show in Section \ref{se:label}. 
In particular, the main analysis results are discussed in Section 
\ref{ssec:secmom}, while in Section \ref{se:label}, we 
describe how the above  analysis in concentration  can be exploited to get
an opportunistic protocol for provably-good  community-sensitive labeling.

\subsection{Second moment analysis for sparse cuts}\label{ssec:secmom}

Restriction to $(n,d,b)$\emph{-clustered regular 
graphs} simplifies the analysis of the \aveproc dynamics. When $G$ is regular,
$\avgW$ defined in \eqref{eq:expectedW} can be written as



\[ 
	\avgW = \left( 1 - \frac 1n 
	\right) I + \frac 1n \, P = I - \frac 1n \, \lapl
\] 

This obviously implies that  $\avgW$ and $\lapl$ share the same 
eigenvectors, while every eigenvalue  $\lambda_i$ of $\lapl$  
corresponds to an eigenvalue $ \bar{\lambda}_i = 1- \lambda_i/n$ of $\avgW$. 
For $(n,d,b)$\emph{-clustered regular 
graphs}, these facts and our preliminary remarks in Section 
\ref{sec:avg} further imply
$\bar{\lambda}_2 = 1 - \lambda_2/n = 1 - 2b/dn$ whenever $\lambda_3 > 
\frac{2b}d$ while, very importantly,
the partition indicator vector  $\bchi$ turns out to be  the   
eigenvector of $\avgW$ corresponding to $\bar{\lambda}_2$ (see \eqref{eq:expectedW}). 
As a consequence, the orthogonal component
$\bfied_\perp$ in~\eqref{eq:expecvalu} is $\bzero$ in this case.


On the other hand, 
even in this restricted setting, 
our second moment analysis requires 
new, non-standard concentration results for the product of random
matrices that apply to far more general settings and may
be of independent interest.

For the sake of readability, we here denote by $\by^{(t)} = Q_2 \bx^{(t)}$ the component
of the state vector in the eigenspace of the second eigenvalue of 
$\avgW$, while $\bz^{(t)} = Q_{3 \cdots n} \bx^{(t)}$ denotes 
$\bx^{(t)}$'s projection onto the subspace 
orthogonal to $\bone$ and $\bchi$.
If we also set $\bx_{\|} = Q_1 \bx^{(0)}$, we can write:
\begin{equation}
    \bx^{(t)} = \bx_{\|} + \by^{(t)} + \bz^{(t)}. 
    \label{eq:decomp}
\end{equation}
Notice that, by taking  expectations in the equation above, we 
get~\eqref{eq:expecvalu} with $\Expec{}{\by^{(t)}} = \alpha_2 \bar{\lambda}_2^t
\bchi$ and $\Expec{}{\bz^{(t)}} = \be^{(t)}$.

Our analysis of the process induced by  \aveproc{$(1/2)$}  provides the following
bound.

\begin{theorem}[Second moment analysis]\label{thm:mom.bound}
    Let $G$ be  an $(n,d,b)$-clustered regular graph with $\lambda_2 
    = \frac{2b}{d} = o\left(\lambda_3 / \log n \right)$.
    Then, for every 
    $\frac{3n}{\lambda_3}\log n\le t \le \frac{n}{4\lambda_2}$ it holds that
    \begin{equation*}
        \Expec{}{\left\| \by^{(t)} + \bz^{(t)} - \by^{(0)} \right\|^2} 
        \leq  \frac{3 \lambda_2 t}{n} \,.
    \end{equation*}
\end{theorem}

We prove Theorem~\ref{thm:mom.bound} by bounding and tracking the lengths of the projections 
of $\bx^{(t)}$ onto the eigenspace of $\lambda_2$ and onto the space orthogonal
to $\bone$ and $\bchi$, i.e. $\| \by^{(t)} \|^2$ and $\|\bz^{(t)}\|^2$. 
Due to lack of space, the proof is deferred to Appendix \ref{apx:proofofhmmom.bound}.

We here want just to remark that  
 the only part using the regularity of the graph is 
the derivation of the upper bound on $\Expec{}{\|\by^{(t+1)}\|^2}$   
(see Lemma~\ref{lem:ynorm}), in particular its  second addend. This term arises from 
an expression involving the Laplacian of $G$, which is far from simple in 
general, but that very nicely simplifies in the regular case. We suspect that
increasingly weaker bounds should be achievable as the graph deviates from 
regularity.

Theorem~\ref{thm:mom.bound} gives an upper bound on the squared norm of the 
difference of the state vector at step $t$ with the state vector at step $0$. 
Corollary \ref{cor:goodt} below shows how such a \textit{global} bound can be used to derive 
\textit{pointwise} bounds on the values of the nodes.

\begin{definition}\label{def:ephemeral}
    A node $v$ is \emph{$\epsilon$-good} at time $t$ if  
    \[
        ( \bx^{(t)}_{v} - ( \bx_{\|,{v}} + \by^{(0)}_{v}) )^2 \leq 
        \frac {\epsilon^2}{n}\|\by^{(0)}\|^2,
    \]
    it is \emph{$\epsilon$-bad} otherwise. We define by $B_t$ the set of 
nodes  that are $\epsilon$-bad at time $t$:
$B_t = \{u: u\makebox{ is $\epsilon$-bad at time $t$}\}$.
\end{definition}

Observe first that, by  definition of $\epsilon$-bad node and some counting argument,  we can prove both the next inequality and 
the corollary below
  (see Appendix \ref{apx:proofofcoroll:goodt} for their proofs)

\begin{equation}\label{eq:numepsgood}
	|B_t|\le\frac n {\epsilon^2\|\by^{(0)}\|^2} \| \by^{(t)} + 
    \bz^{(t)}  - \by^{(0)} \|^2.
\end{equation}

\begin{cor}\label{cor:goodt}
Assume $3\frac{n}{\lambda_3}\log n\le t \le 3c\frac{n}{\lambda_3}\log n$ 
for any absolute constant $c \ge 1$ and $\lambda_2/\lambda_3\le 
\epsilon^4/(4c\log n)$:
\begin{equation}\label{eq:prob_good}
    \Prob{}{|B_t| > \epsilon n\,|\, \bx^{(0)} = \bx} \le \epsilon.   
\end{equation}
\end{cor}

The next lemma   gives a bound on the
number of nodes that are good over a relatively large time-window. This is the key-property
that we  will   use to analyse the asynchronous protocol~\labelsign  (see the next subsection
and  Lemma~\ref{lemma:numberunlucky}).

\begin{lemma}[Non-ephemeral good nodes]\label{lem:main}
Let $\varepsilon > 0$ be an arbitrarily small value, let $G$ be an 
$(n,d,b)$-clustered regular graph with
$\frac{\lambda_2}{\lambda_3}\le\frac{\lambda_3\epsilon^4}{c \log^2 n}$, for a 
large enough costant $c$. If we  execute \aveproc{$(1/2)$} on $G$, it holds that
 
\[ 
\Prob{}{  |B_t|  \leq 3 \epsilon\cdot n \, ,   \, \forall  \, t \, : \,  
 6 \frac{n}{\lambda_3} \log n \leq t\leq 12 \frac{n}{\lambda_3} 
\log n } \geq 1 - \epsilon \,. 
\]
 \end{lemma}

\subsubsection{Proof of Lemma \ref{lem:main}: An overview}

The main idea of the proof is to first show that with probability 
strictly larger than  $1-\epsilon$, the number of $\epsilon$-good nodes 
is at least $n\cdot (1-\epsilon/\log n)$ in every round $t\in [t_1, 
2t_1]$. Theorem~\ref{thm:mom.bound} already ensures this to be true in 
any given time step within a suitable window, but simply taking a union bound will not work, since we 
have $n \log n$ time steps and only a $1-\epsilon$ probability of 
observing the desired outcome in each of them. We will instead argue about 
the possible magnitude of the change in $ \| \by^{(t)} + \bz^{(t)} - 
\by^{(0)} \|^2$ over time, assuming this quantity is small at time 
$6 \frac{n}{\lambda_3} \log n$. 
We will then show that our argument implies that, with probability 
$1-\epsilon$, at least $n - \epsilon n$ nodes remain $\epsilon$-good 
over the entire window $[6 \frac{n}{\lambda_3} \log n, 12 \frac{n}{\lambda_3} \log n]$.

The full proof of Lemma \ref{lem:main} is given in Appendix \ref{apx:proofofmainlm}.

\subsection{ The \labelsign protocol }
\label{se:label}
Leveraging the results of Subsection~\ref{ssec:secmom}, we next propose a 
simple, lightweight opportunistic  protocol that provides community-sensitive labeling
for graphs that exhibit a relatively sparse cut.

The algorithm, denoted as~\labelsign, adds a simple
\emph{labeling rule} to  the \aveproc{$(1/2)$}\ process: Each node keeps track of the number of times it is 
activated. Upon its $T$-th activation, for a suitable $T = \Theta(\log n)$, 
the node uses the sign of its current 
value as a binary label. The above local strategy is applied to $\ell$ independent
runs of   \aveproc{$(1/2)$}, so that every node is eventually assigned
a binary signature of length $\ell$.


\begin{algorithm}
    \LinesNotNumbered
    \DontPrintSemicolon
    \SetAlgoVlined
    \SetNoFillComment
    \labelsign{(T,$\ell$)} (for a node $u$ that is one of the two endpoints of an active edge)\;
    \begin{description}
       \item[Component selection: ] Jointly\footnote{This instruction can be
                implemented using a simple strategy: $u$ and $v$ generate
                random values $j_u\in [\ell]$ and $j_v\in [\ell]$, exchange them and
                then set $j=j_u+j_v \mod\hspace{2pt} \ell$. } with the other
                endpoint choose a component\\  \hspace{100pt}$j \in [\ell]$ u.a.r.\;
        \item[Initialization and update:] {Run one step of \aveproc{ $(1/2)$  \\   \hspace{100pt}for component $j$}}.
        \item[Labeling:] If this is the $T$-th activation of component $j$: set 
            $\bh^{sign}_u(j) = \sgn(\bx_u(j))$.
    \end{description}
    \caption{\label{algo:sign_labeling} \labelsign
         algorithm for a node $u$ of an active edge.}
\end{algorithm}

\smallskip\noindent
Algorithm~\labelsign achieves community-sensitive labeling (see 
Definition~\ref{def:senshash}), as stated in the following theorem and corollary. 

\begin{theorem}[Community-sensitive labeling]
    \label{thm:main_small} \label{thm:csl-small}
    Let $\varepsilon > 0$ be an arbitrarily small value, let $G$ be an 
    $(n,d,b)$-clustered regular graph with 
    $\frac{\lambda_2}{\lambda_3}\le\frac{\lambda_3\epsilon^4}{c \log^2 n}$, for a 
    large enough constant $c$. 
    Then, protocol \labelsign$(T,\ell)$ with $T = (8/\lambda_3) \log n $ and 
    $\ell = 10 \varepsilon^{-1} \log n$ performs a $\gamma$-community-sensitive 
    labeling of $G$ according to Definition~\ref{def:senshash} with $c_1 = 
    4 \varepsilon$, $c_2 = 1/6$ and $\gamma = 6 \varepsilon$, w.h.p.
    The convergence time is $\bigO(n \ell \log n / \lambda_3)$ and the   work per node is 
    $\bigO( \ell    \log n / \lambda_3)$, w.h.p.
\end{theorem}

Notice that, according to the hypothesis of Theorem~\ref{thm:main_small}, in order to
set local parameters $T$ and $\ell$, nodes should know parameters $\varepsilon$ and
$\lambda_3$ (in addition to a polynomial upper bound on the number of the nodes). 
However, it easy to restate it in a slightly restricted form that does not require 
such assumptions on what nodes know about the underlying graph.

\begin{cor} \label{cor:csl-small}
Protocol \labelsign$(80 \log n, 600 \log n)$ performs a $(1/10)$-community-sensitive 
labeling, according to Definition~\ref{def:senshash} with $c_1 = 1/15$ 
and $c_2 = 1/6$, of any $(n,d,b)$-clustered regular graph $G$ with 
$\lambda_3 \geqslant 1/10$ and $\lambda_2 \leqslant 1/(c \log^2 n)$ for a large
enough constant $c$.
\end{cor}

\subsubsection{Proof of Theorem \ref{thm:csl-small}: An Overview} \label{sssec:prfskthm:thm:csl-small}
We here sketch the main   arguments  proving    Theorem 
\ref{thm:csl-small}: Its full  proof is deferred to Appendix \ref{apx:toolsforcslsign}.

Lemma~\ref{lem:main} essentially states that over a suitable time window of
size $\Theta(n\log n)$, for all nodes $u$ but a fraction
$\bigO\left(\varepsilon/\log n\right)$, we have $\sgn(\bx^{(t)}_u) =
\sgn(\bx_{\|,{u}} + \by^{(0)}_{u}) ))$. Recalling that $\bx_{\|}$ and
$\by^{(0)}$ respectively are $\bx^{(0)}$'s projections along $\bchi/\sqrt{n}$
and $\bone/\sqrt{n}$, this immediately implies that, with probability $1 -
\varepsilon$ and up to a fraction $\varepsilon$ of the nodes,
$\sgn(\bx^{(t)}_u) = \sgn(\bx^{(t)}_v)$, whenever $u$ and $v$ belong to the
same community and $t$ falls within the aforementioned window. As to the latter
condition, we prove that each node labels itself within the right window with
probability at least $1 - 1/n$.\footnote{ It may be worth noting that
$\sgn(\bx^{(t)}_u) = \sgn(\bx^{(t)}_v)$ for $u$ and $v$ belonging to the same
community does not imply $\sgn(\bx^{(t)}_u) \ne \sgn(\bx^{(t)}_v)$  when they
don't.} Moreover, $\sgn(\bx_{\|,{u}} + \by^{(0)}_{u}) )) = \sgn(\bchi_u)$,
whenever $\by^{(0)}_u$ exceeds $\bx_{\|,{u}}$ in modulus, which occurs with
probability $1/2 - o(1)$ from the (independent) Rademacher initialization. As a
consequence, if we run $\ell$ suitably independent copies of the process (see
Algorithm~\ref{algo:sign_labeling}), the following will happen for all but a
fraction $\bigO(\varepsilon)$ of the nodes: the signatures of two nodes
belonging to the same community will agree on $\ell - o(1)$ bits, whereas those
of two nodes belonging to different communities will disagree on $\Omega(\ell)$
bits, i.e., our algorithm returns a community-sensitive labeling of the graph.

\section{Regular Graphs with a Dense Cut}\label{ssec:pl-analysis} 
In this section, we extend our study to the lazy averaging algorithm 
\aveproc{$(\delta)$} where $\delta < 1/2$. Similar to the previous 
section, we assume that the underlying graph $G$   is an $(n, d, 
b)$-clustered regular graph and $\lambda_3 > \lambda_2 = 2b/d$. 
However, this new analyses and the clustering protocol we derive from    
will work even for large (constant) $\lambda_2$, in contrast to those 
in Section~\ref{sec:second_moment} which only works for small 
$\lambda_2 \ll 1/\log^2 n$. The structure of this section is similar to 
the previous one. Indeed, in      Subsection 
\ref{ssec:secondmoment-II}, we propose a second moment analysis of the 
\aveproc{$(\delta)$} for  the above-mentioned regime of  $\lambda_2$. 
Then, in Subsection \ref{sec:reconstruct-analysis}, we exploit the  
analysis above  to devise a   protocol that guarantees a weak 
reconstruction for the underlying graph with arbitrarily-large constant 
probability and thus, by running independent   ``copies'' of the 
protocol  (so, similarly to the previous section), we easily obtain a 
community-sensitive labeling of the graph, with high probability.

\subsection{Second moment analysis for  large $\lambda_2$} \label{ssec:secondmoment-II}

Informally speaking, we show that, for an appropriate value of 
$\delta$ and any $t$ such that $\Omega(n \log n) \leqs t \leqs 
\bigO(n^2)$, with large probability, the vector $\by^{(t)} + 
\bz^{(t)}$ is almost parallel to $\chi$, i.e., $\|\bz^{(t)}\|$ is much 
smaller than $\|\by^{(t)}\|$. A more precise statement is given  
below as Theorem~\ref{thm:main-concen}. Note that, for brevity, we 
write $\cE$ here to denote the sequence $\{(u_t, v_t)\}_{t \in \N}$ of 
the edges chosen by the protocol.

\begin{theorem} \label{thm:main-concen}
For any sufficiently large $n \in \N$, any\footnote{Here 0.8 is 
arbitrary and can be changed to any constant less than 1. However, we 
pick an absolute constant here to avoid introducing another parameter 
to our theorem.} $\delta \in (0, 0.8(\lambda_3 - \lambda_2))$ and any 
$t \in \left[\Omega\left(\frac{n}{\delta(\lambda_3 - \lambda_2)} 
\log\left(n/\delta\right)\right), 
\bigO\left(\frac{n^2}{\delta(\lambda_3 - \lambda_2)}\left(\frac{d 
(\lambda_3 - \lambda_2)}{\delta b}\right)^{2/3}\right)\right]$, we have
\begin{align*}
\Pr_{\bx^{(0)}, \cE}\left[\|\bz^{(t)}\|^2 \leqs \sqrt{\frac{\delta b}{d 
(\lambda_3 - \lambda_2)}} \|\by^{(t)}\|^2\right] \geqs 1 - 
\bigO\left(\sqrt[3]{\frac{\delta b}{d (\lambda_3 - \lambda_2)}} + 
\frac{1}{\sqrt{n}}\right).
\end{align*}
\end{theorem}

Theorem~\ref{thm:main-concen} should be compared to 
Theorem~\ref{thm:mom.bound}: both assert that $\|\by^{(t)}\|$ is much 
larger than $\|\bz^{(t)}\|$, but Theorem~\ref{thm:main-concen} works 
even when $\lambda_2$ is quite large whereas 
Theorem~\ref{thm:mom.bound} only holds for $\lambda_2 \ll 1/\log n$. 

While the parameter dependencies in Theorem~\ref{thm:main-concen} may 
look confusing at first, there are mainly two cases that are 
interesting here. First, for any error parameter $\epsilon$, we can 
pick $\delta$ depending only on $\varepsilon$ and $\lambda_3 - 
\lambda_2$ in such a way that Theorem~\ref{thm:main-concen} implies 
that, with probability $1 - \varepsilon$, $\|\bz^{(t)}\|^2$ is at most 
$\varepsilon \|\by^{(t)}\|^2$, as stated below.

\begin{cor}\label{cor.one}
For any constant $\varepsilon > 0$ and for any $\lambda_3 > \lambda_2$, 
there exists $\delta$ depending only on $\varepsilon$ and $\lambda_3 - 
\lambda_2$ such that, for any sufficiently large $n$ and for any $t \in 
[\Omega_{\varepsilon, \lambda_3 - \lambda_2}(n \log n), \bigO(n^2)]$, 
we have $$\Pr_{\bx^{(0)}, \cE}\left[\|\bz^{(t)}\|^2 \leqs \varepsilon 
\|\by^{(t)}\|^2\right] \geqs 1 - \varepsilon.$$
\end{cor}

Another interesting case is when $\delta = 1/2$ (i.e., we consider the 
basic averaging protocol). Recalling that $\lambda_2 = 2b/d$, 
observe that $\lambda_2$ appears in both the bound on $\|\bz^{(t)}\|^2$ 
and the error probability. Hence, we can derive a similar lemma as the 
one above, but with $\lambda_2$ depending on $\varepsilon$ instead of 
$\delta$:

\begin{cor}\label{cor.two}
Fix $\delta = 1/2$. For any constant $\varepsilon > 0$, 
any\footnote{0.7 here can be replaced by any constant larger than 0.5.} 
$\lambda_3 > 0.7$, any sufficiently small $\lambda_2$ depending only on 
$\varepsilon$, any sufficiently large $n$ and any $t \in 
[\Omega_{\varepsilon}(n \log n), \bigO(n^2)]$, we have 
$$\Pr_{\bx^{(0)}, \cE}\left[\|\bz^{(t)}\|^2 \leqs \varepsilon 
\|\by^{(t)}\|^2\right] \geqs 1 - \varepsilon.$$
\end{cor}

\subsubsection{Proof of Theorem~\ref{thm:main-concen}: An Overview}

Due to space constraint, the full proof of 
Theorem~\ref{thm:main-concen} is deferred to 
Appendix~\ref{app:pl-proof}. We provide a brief summary of the ideas 
behind the proof here. Compared to the proof of 
Theorem~\ref{thm:mom.bound}, the main additional technical challenge in 
the new proof is to show that $\|\by^{(t)}\|$ is large with reasonably 
high probability. In Theorem~\ref{thm:mom.bound}, this is true because 
$\lambda_2$ is so small that $\by^{(t)}$ remains almost unchanged from 
$\by^{(0)}$. However, in the setting of large $\lambda_2$, this is not 
true anymore; for constant $\lambda_2$, even $\E[\by^{(t)}]$ shrinks by 
a constant factor from $\by^{(0)}$ when $t \geqs 
\Omega_{\lambda_2}(n)$.

As a result, we need to develop a more fine-grained understanding of 
how $\|\by^{(t)}\|$, $\|\bz^{(t)}\|$ changes over time. Specifically, 
at the heart of our analysis lies the following lemma\footnote{Lemma 
\ref{lem:one-step-1} with its full statement and proof is given in 
Appendix \ref{app:pl-proof} as Lemma \ref{lem:one-step}. Recall that
$\lambda_2 = 2b/d.$} which allows us 
to understand how $\|\by^{(t)}\|$, $\|\bz^{(t)}\|$ behave, given 
$\|\by^{(t - 1)}\|$, $\|\bz^{(t - 1)}\|$:

\begin{lemma} \label{lem:one-step-1}
For any $t \in \N$,
\begin{align*}
\E [ \|\by^{(t)}\|^2 ] \leqs \left(1 - \frac{4\delta \lambda_2}{n} + 
\frac{8\delta^2\lambda_2}{n^2}\right)\| \by^{(t - 1)} \|^2 + 
\left(\frac {8\delta^2\lambda_2}{n^2}\right)\|\bz^{(t - 1)}\|^2
\end{align*}
and
\begin{align*}
\E [ \|\bz^{(t)}\|^2 ]  \leqs 
\left(\frac{4\delta^2\lambda_2}{n}\right)\|\by^{(t - 1)}\|^2 + \left( 1 
- \frac{4\delta(1 - \delta)\lambda_3}{n} \right) \|\bz^{(t - 1)}\|^2
\end{align*}
where the expectation is over the random edge selected at time $t$.
\end{lemma}

For simplicity of the overview, let us pretend that the cross terms 
were not there, i.e., that $\E [ \|\by^{(t)}\|^2 ] \leqs \left(1 - 
\frac{4\delta \lambda_2}{n} + \frac{8\delta^2\lambda_2}{n^2}\right)\| 
\by^{(t - 1)} \|^2$ and $\E [ \|\bz^{(t)}\|^2 ]  \leqs \left( 1 - 
\frac{4\delta(1 - \delta)\lambda_3}{n} \right) \|\bz^{(t - 1)}\|^2$. 
These imply that
\begin{align}
\E [ \|\by^{(t)}\|^2 ] \leqs \left(1 - \frac{4\delta \lambda_2}{n} + 
\frac{8\delta^2\lambda_2}{n^2}\right)^t\| \by^{(0)} \|^2
\label{eq:second-moment-yt}
\end{align}
and
\begin{align}
\E [ \|\bz^{(t)}\|^2 ]  \leqs \left( 1 - \frac{4\delta(1 - \delta)\lambda_3}{n} \right)^t \|\bz^{(0)}\|^2.
\label{eq:second-moment-zt}
\end{align}

Now, by Markov's inequality, (\ref{eq:second-moment-zt}) implies that, 
with 0.99 probability, $\|\bz^{(t)}\|$ is at most $\bigO\left(\left( 1 
- \frac{2\delta(1 - \delta)\lambda_3}{n} \right)^t 
\|\bz^{(0)}\|\right)$. However, it is not immediately clear how 
(\ref{eq:second-moment-yt}) can be used to lower bound $\|\by^{(t)}\|$. 
Fortunately for us, it is rather simple to see that, for a fixed 
$\by^{(0)}$, $\E[\by^{(t)}]$ can be computed exactly; in particular,
\begin{align}
\E[\by^{(t)}] = \left(1 - \frac{2 \delta \lambda_2}{n}\right)^t \by^{(0)}.
\label{eq:first-moment-yt}
\end{align}
Let $a_y(t) \in \mathbb{R}$ be such that $\by^{(t)} = a_y(t) \cdot 
(\chi / \sqrt{n})$. (\ref{eq:first-moment-yt}) can equivalently be 
stated as $\E[a_y(t)] = (1 - 2\delta \lambda_2/n)^t a_y(0)$. This, 
together with (\ref{eq:second-moment-yt}), can be used to bound the 
variance of $a_y(t)$ as follows:
\begin{align*}
\text{Var}(a_y(t)) &\leqs \left(1 - \frac{4\delta \lambda_2}{n} + \frac{8\delta^2\lambda_2}{n^2}\right)^t a_y(0)^2 - (1 - 2\delta \lambda_2/n)^{2t} a_y(0)^2 \\
&= \bigO_{\delta, \lambda_2}(t/n^2) \left(\E[a_y(t)]\right)^2.
\end{align*}
Hence, when $t \ll n^2$, Chebyshev's inequality implies that $a_y(t)$ 
concentrates around $\E[a_y(t)]$ or, equivalently, $\|\by^{(t)}\|$ 
concentrates around $\left(1 - \frac{2 \delta 
\lambda_2}{n}\right)^t\|\by_0\|$. 

Finally, observe that, since $\lambda_2 < \lambda_3$, for sufficiently 
small $\delta$, we have $2 \delta \lambda_2 < 2\delta(1 - 
\delta)\lambda_3$. Hence, when $t \gg n \log n$, $\left(1 - \frac{2 
\delta \lambda_2}{n}\right)^t$ is polynomially (say $n^{10}$ times) 
larger than $\left(1 - \frac{2 \delta(1 - \delta) 
\lambda_3}{n}\right)^t$. It is also not hard to see that, for a random 
starting vector, $\|\bz^{(0)}\| \ll n^{10} \|\by^{(0)}\|$ with high 
probability. This means that, for this range of $t$, we have $\left(1 - 
\frac{2 \delta \lambda_2}{n}\right)^t\|\by^{(0)}\| \gg \left(1 - 
\frac{2 \delta (1 - \delta) \lambda_3}{n}\right)^t\|\bz^{(0)}\|$ with 
high probability. Since $\|\by^{(t)}\|$ concentrates on the former 
quantity whereas $\|\bz^{(t)}\|$ often does not exceed a constant 
factor of the latter, we can conclude that $\|\by^{(t)}\|$ is indeed 
often much larger than $\|\bz^{(t)}\|$.

This wraps up our proof overview of Theorem~\ref{thm:main-concen}.

\subsection{The \labeljump protocol} 
\label{sec:reconstruct-analysis}


Relying on our insights from the previous section, we propose a 
lightweight protocol named \labeljump, which makes use of the  lazy 
version of the averaging process. Here $\delta \in [0, 1]$ and 
$\tau^{\stored}, \ttau^{\stored}, \tau^{\tend}, \ttau^{\tend} \in \N$ 
are parameters that will be chosen later. Intuitively, protocol 
\labeljump exploits the expected monotonicity in the behaviour of 
$\sgn(\bx_u^{(t)} - \bx^{(t-1)})$ highlighted in Section \ref{se:expec}. 
Though this property does not hold for a single realization of the averaging 
process in general, the results of Section \ref{ssec:pl-analysis} 
allow us to show that the sign of $\bx^{(\tau_u^e)} - \bx^{(\tau_u^s)}$ reflects 
$u$'s community membership for most vertices with probability $1 - 
o(1)$ (i.e., the algorithm achieves weak reconstruction) when 
$\tau_u^s$ and $\tau_u^e$ are randomly chosen within a 
suitable interval. This is the intuition behind the main result of this 
section. Due to space constraints, the full proof of Theorem 
\ref{thm:main-reconstruct} below is deferred to Appendix~\ref{app:reconstruct}.

\begin{theorem} \label{thm:main-reconstruct}
    Let $n$ be any sufficiently large even positive integer. For any $0 <
    \delta < 0.8(\lambda_3 - \lambda_2)$, there exist 
    $\tau^{\stored},
    \ttau^{\stored}, \tau^{\tend}, \ttau^{\tend} \in \N$ such that, with
    probability $1 - \bigO\left(\sqrt[8]{\frac{\delta b}{d (\lambda_3 -
    \lambda_2)}} + \sqrt[4]{\frac{1}{\log n}}\right)$, after
    $\bigO\left(\frac{n}{\delta(\lambda_3 - \lambda_2)} \log\left(n/\delta\right) +
    \frac{nd}{b\delta}\right)$ rounds of  \, \labeljump{($\delta,
    \tau^{\stored}, \ttau^{\stored}, \tau^{\tend}, \ttau^{\tend}$)}, every
    node labels its cluster and this labelling is a
    $\left(\sqrt[8]{\frac{\delta b}{d (\lambda_3 - \lambda_2)}} +
    \sqrt[4]{\frac{1}{\log n}}\right)$-weak reconstruction of $G$.
    The convergence  time of this  algorithm is $\Omega_{\delta}\left(n
\left(\log n + \frac{d}{b}\right)\right)$. 
\end{theorem}


\begin{algorithm}
    \LinesNotNumbered
    \DontPrintSemicolon
    \SetAlgoVlined
    \SetNoFillComment
    \labeljump{($\delta, \tau^{\stored}, \ttau^{\stored}, \tau^{\tend}, \ttau^{\tend}$)} 
    (for a node $u$ that is one of the two endpoints of an active edge)\;
    \begin{description}
        \item[Initialization:] The first time it is activated, $u$    chooses $\tau^{\stored}_u, \tau^{\tend}_u \in \N$
            independently\\  uniformly at random from $[\tau^{\stored}, \ttau^{\stored}]$ and $[\tau^{\tend}, \ttau^{\tend}]$ respectively.
            Moreover, let $\tau_u = 0$.
        \item[Update (and \aveproc's initialization):]
            Run one step \\ of \aveproc{$(\delta)$}.
        \item[Labeling:] 
            If $\tau_u = \tau^{\stored}_u$, then set $x^{\stored}_u = x_u$. \\
            If $\tau_u = \tau^{\tend}_u$, then label $\bh^{jump}_u = \sgn(x^{\stored}_u - x_u)$.
    \end{description}
    \caption{\label{algo:jump_labeling} \labeljump
         algorithm for a node $u$ of an active edge. Here, $\tau_u$ 
         is a local counter keeping track of the number of times $u$ 
         was an endpoint of an active edge, while $x_u$ is $u$'s current 
         value.}
\end{algorithm}
\smallskip

\begin{remark}
    The $nd/b$ dependency in the running time is necessary; imagine we start with a good state where $\bx^{(0)} = \bz^{(0)} = 0$. In this case, the values on one side of the partition are all $a_y(0)$ and the values on the
    other side are $-a_y(0)$. It is simple to see that, after $o(nd/b)$ steps of our protocol, $1 - o(1)$ fraction of the values remain the same. For these nodes, it is impossible  them to determine which cluster they are in and, hence, no good reconstruction can be achieved.
\end{remark}

Similarly to our concentration result in Section \ref{ssec:pl-analysis}, let us
demonstrate the use of Theorem~\ref{thm:main-reconstruct} to the two
interesting cases. First, let us start with the case where $\lambda_3 -
\lambda_2$ is constant. Again, in this case, for any error parameter
$\varepsilon > 0$, we can pick $\delta = \delta(\varepsilon, \lambda_2 -
\lambda_3)$ sufficiently small so that, with probability $1 - \varepsilon$, the
protocol achieves $\varepsilon$-weak reconstruction, as stated below.

\begin{cor} 
    \label{cor:reconstruct}
    For any constant $\varepsilon > 0$ and for any $\lambda_3, \lambda_2$,
    there exists $\delta$ depending only on $\varepsilon$ and $\lambda_3 -
    \lambda_2$ such that, for any sufficiently large $n$, there exists
    $\tau^{\stored}, \ttau^{\stored}, \tau^{\tend}, \ttau^{\tend} \in \N$ such
    that, with probability $1 - \varepsilon$, after $\bigO_{\varepsilon, \lambda_3
    - \lambda_2}\left(n \log n + \frac{n}{\lambda_2}\right)$ rounds of
    \labeljump{($\delta, \tau^{\stored}, \ttau^{\stored}, \tau^{\tend},
    \ttau^{\tend}$)}, every node labels its cluster and this labelling is a
    $\varepsilon$-weak reconstruction of $G$.
\end{cor}

As in Section \ref{ssec:pl-analysis}, we can consider the (non-lazy) averaging
protocol and view $\lambda_2$ instead as a parameter. On this front, we arrive
at the following reconstruction guarantee.

\begin{cor}
    \label{cor:reconstruct2}
    Fix $\delta = 1/2$. For any constant $\varepsilon > 0$, any $\lambda_3 > 0.7$,
    any sufficiently small $\lambda_2$ depending only on $\varepsilon$, any
    sufficiently large $n$, there exists $\tau^{\stored}, \ttau^{\stored},
    \tau^{\tend}, \ttau^{\tend} \in \N$ such that, with probability $1 -
    \varepsilon$, after $\bigO_{\varepsilon}\left(n \log n +
    \frac{n}{\lambda_2}\right)$ rounds of \labeljump{($\delta, \tau^{\stored},
    \ttau^{\stored}, \tau^{\tend}, \ttau^{\tend}$)}, the nodes' labelling is a
    $\varepsilon$-weak reconstruction of $G$.
\end{cor}

While the weak reconstruction in the above claims is guaranteed only with arbitrarily-large constant probability, we can boost this success probability considering  the same 
approach we used in Subsection \ref{se:label} to get community-sensitive  binary strings  of size $\ell= \Theta(\log n)$
 from the sign-based protocol. 

Indeed, we    first   run $\ell = \Theta_{\varepsilon}(\log n)$ copies of \labeljump where, 
similarly to Algorithm \ref{algo:sign_labeling}, ``running $\ell$ copies'' of \labeljump means that each node keeps $\ell$ copies of the states of \labeljump and, when an edge $\{u, v\}$ is activated, $u$ and $v$ jointly sample a random $j \in [\ell]$ and run the $j$-th copy of \labeljump. 

In the previous section, we have seen that Lemma \ref{lem:main}  and      the  repetition approach above   allowed us to get
a good community-sensitive labeling, w.h.p. (not a good weak-reconstruction).
Interestingly enough, the somewhat stronger concentration results given in this section allow us to  ``add''  a simple \emph{majority} rule on the top of the
$\ell$ components and get  a  ``good'' single-bit label, as described below.

 When all $\ell$ components of a node $u$ have been set,  
  node $u$ sets $\bh^{jump}_u = \sc{Majority}_{j \in [\ell]}(\bh^{jump}_u(i))$ where $\bh^{jump}_u(j)$ is the binary label of $u$ from the $j$-th copy of the protocol.

Observe that the weak reconstruction guarantee of \labeljump shown earlier implies that the expected number of mislabelings of each copy
 is at most $2\varepsilon n$, i.e., $\E[\{u \in V \mid |\bh^{jump}_u(i) \ne \chi_u|\}] \leqs 2\varepsilon n$. Now, since the number of mislabelings of each copy is independent, the total number of mislabelings is at most $\bigO(\varepsilon n \ell)$, w.h.p. However, if the eventual label of $u$ is incorrect, it must contributes to mislabeling across at least $\ell/2$ copies. As a result, there are at most $\bigO(\varepsilon n)$ mislabelings in the new protocol, w.h.p. 

The above    approach   in fact works for any weak reconstruction protocol (not just \labeljump) and, in our case, it easily  gives the following result.

\begin{cor}\label{cor:main2}
For any constant $\varepsilon > 0$ and $\lambda_3 > \lambda_2$, there is a protocol that yields an $\varepsilon$-weak reconstruction of $G$ , w.h.p. The convergence time
is  
$\Theta_{\varepsilon, \lambda_3 - \lambda_2}\left(n\left(\log^2 n + \frac{\log n}{\lambda_2}\right)\right)$ rounds, while the work per node is $\bigO_{\varepsilon, \lambda_3
- \lambda_2}\left(\log^2 n + \frac{\log n}{\lambda_2}\right)$.
\end{cor}


We finally remark that, for the dense-cut case we focus on  in this section (i.e. $\lambda_2 = 2b/d = \Theta(1)$), the fraction of outliers turns out to be 
a  constant we can made arbitrarily small. If we relax the  condition to   $\lambda_2 = o(1)$, then this fraction can be made $o(1)$,  accordingly.
This issue will be clarified in the full version of the paper.

\subsubsection{Proof of Theorem \ref{thm:main-reconstruct}: An Overview}
\label{sec:reconstruct-overview}

We now give an informal overview of our proof, which builds on the 
concentration results from Section \ref{ssec:pl-analysis}. Since our 
discussion here will involve both local times and global times, let us 
define the following notation to facilitate the discussion: for each 
vertex $u \in V$, let $T_u: \N \to \N$ be a function that maps the 
local time of $u$ to the global time, i.e., $T_u(\tau) \triangleq \min 
\{t \in \N \mid |\{i \leqs t \mid u \in \{u_i, v_i\}\}| \geqs \tau\}$ 
where $(\{u_i, v_i\})_{i \in \N}$ is the sequence of active edges.

Recall from the previous section that we let $a_y(t) \in \mathbb{R}$ be such that $\by^{(t)} = a_y(t) \cdot (\chi/\sqrt{n})$. Let us also assume without loss of generality that $a_y(0) \geqs 0$. Observe first that our concentration result implies the following: for any $t$ such that $\Omega(n \log n) \leqs t \leqs \bigO(n^2)$, with large probability, $\chi_u
 (\bx_u^{(t)} - \bx_{||, u})$ is roughly $\E_{\cE} a_y(t) / n$ for most vertices $u \in
 V$; let us call these vertices \emph{good for time $t$}. Imagine for a moment
 that we change the protocol in such a way that each $u$ has access to the
 global time $t$ and $u$ assigns $\bh^{jump}_u = \sgn(\bx_u^{(t^{\tend}) }-
 \bx_u^{(t^{\stored})})$ for some $t^{\stored}, t^{\tend} \in [\Omega(n \log n),
 \bigO(n^2)]$ that do not depend on $u$. If $t^{\tend} - t^{\stored}$ is large
 enough, then $\E_{\cE} a_y(t^{\stored}) \gg \E_{\cE} a_y(t^{\tend})$. This
 means that, if a vertex $u \in V$ is good at both times $t^{\stored}$ and
 $t^{\tend}$, then we have that $\chi_u (\bx_u^{(t^{\stored})} - \bx_{||, u}) \approx
 \E_{\cE} a_y(t^{\stored}) / n \gg \E_{\cE} a_y(t^{\tend}) / n \approx \chi_u
 (\bx_u^{(t^{\tend})} - \bx_{||, u})$. Note that when $\chi_u \cdot \bx_u^{(t^{\stored})} > \chi_u \cdot \bx_u^{(t^{\tend})}$, we have $\bh^{jump}_u = \chi_u$. From
 this and from almost all vertices are good at both times $t^{\stored}$ and
 $t^{\tend}$, $\bh^{jump}$ is indeed a good weak reconstruction for the graph!
 
 The problem of the modified protocol above is of course that, in our settings,
 each vertex does not know the global time $t$. Perhaps the simplest approach
 to imitate the above algorithm in this regime is to fix
 $\tau^{\stored}, \tau^{\tend} \in [\Omega(\log n), \bigO(n)]$ and, for each $u \in
 V$, proceed as in \labeljump except with $\tau_u^{\stored} =
 \tau^{\stored}$ and $\tau_u^{\tend} = \tau^{\tend}$. In other words, $u$ assigns
 $\bh^{jump}_u = \sgn(\bx_u^{(T_u(\tau^{\stored}))} - \bx_u^{(T_u(\tau^{\tend}))})$. The
 problem about this approach is that, while we know that $\E_{\cE}
 T_u(\tau^{\stored}) = 0.5 n \tau^{\stored}$ and $\E_{\cE} T_u(\tau^{\tend}) =
 0.5 n \tau^{\tend}$, the actual values of $T_u(\tau^{\stored})$ and
 $T_u(\tau^{\tend})$ differ quite a bit from their means, i.e., on average they
 will be $\Omega(n \sqrt{\log n})$ of away their mean. Since our concentration
 result only says that, at each time $t$, we expect 99\% of the vertices to be
 good, it is unclear how this can rule out the following extreme case: for many
 $u \in V$, $T_u(\tau^{\stored})$ or $T_u(\tau^{\tend})$ is a time step at
 which $u$ is bad. This case results in $\bh^{jump}$ not being a good
 weak reconstruction of $V$.

 The above issue motivates us to arrive at our eventual algorithm, in which
 $\tau^{\stored}_u$ and $\tau^{\tend}_u$ are not fixed to be the same for every
 $u$, but instead each $u$ pick these values randomly from specified
 intervals $[\tau^{\stored}, \ttau^{\stored}]$ and $[\tau^{\tend},
 \ttau^{\tend}]$. To demonstrate why this overcomes the above problem, let us
 focus on the interval $[\tau^{\stored}, \ttau^{\stored}]$. While
 $T_u(\tau^{\stored})$ and $T_u(\ttau^{\stored})$ can still differ from their
 means, the interval  $[T_u(\tau^{\stored}), T_u(\ttau^{\stored})]$ still, with
 large probability, overlaps with most of $[0.5 n \tau^{\stored}, 0.5 n
 \ttau^{\stored}]$ if $\ttau^{\stored} - \tau^{\stored}$ is sufficiently large.
 Now, if $T_u(\tau + 1) - T_u(\tau)$ are the same for all $\tau \in
 [\tau^{\stored}, \ttau^{\stored}]$, then the distribution of
 $\bx_u^{(T_u(\tau^{\stored}))}$ is very close to $\bx_u^{(t^{\stored}_u)}$ if we
 pick $t^{\stored}_u$ randomly from $[0.5 n \tau^{\stored}, 0.5 n
 \ttau^{\stored}]$. From the usual global time step argument, it is easy to see
 that the latter distribution results in most $u$ being good at time
 $t^{\stored}_u$. Of course, $T_u(\tau + 1) - T_u(\tau)$ will not be the same
 for all $\tau \in [\tau^{\stored}, \ttau^{\stored}]$, but we will be able to
 argue that, for almost all such $\tau$, $T_u(\tau + 1) - T_u(\tau)$ is not too
 small, which is sufficient for our purpose.


\bibliographystyle{alpha}
\bibliography{pbm}

\newcommand{\etalchar}[1]{$^{#1}$}
\begin{thebibliography}{DKMZ11}

\bibitem[AAER07]{AAER07}
Dana Angluin, James Aspnes, David Eisenstat, and Eric Ruppert.
\newblock The computational power of population protocols.
\newblock {\em Distributed Computing}, 20(4):279--304, 2007.

\bibitem[ABH14]{abbe2014exact}
Emmanuel Abbe, Afonso~S Bandeira, and Georgina Hall.
\newblock Exact recovery in the stochastic block model.
\newblock {\em IEEE Trans. on Information Theory}, 62(1):471--487, 2014.

\bibitem[BCN{\etalchar{+}}17]{BCNPT17}
Luca Becchetti, Andrea Clementi, Emanuele Natale, Francesco Pasquale, and Luca
  Trevisan.
\newblock Find your place: Simple distributed algorithms for community
  detection.
\newblock In {\em Proc. of the 28th Ann. ACM-SIAM Symp. on Discrete Algorithms
  (SODA'17)}, pages 940--959. SIAM, 2017.

\bibitem[BDG{\etalchar{+}}15]{brito_recovery_2015}
Gerandy Brito, Ioana Dumitriu, Shirshendu Ganguly, Christopher Hoffman, and
  Linh~V. Tran.
\newblock Recovery and rigidity in a regular stochastic block model.
\newblock In {\em Proc. of the ACM-SIAM Symposium on Discrete Algorithms
  (SODA)}, pages 371--390. ACM, 2015.

\bibitem[BGPS06]{BGPS06}
Stephen Boyd, Arpita Ghosh, Balaji Prabhakar, and Devavrat Shah.
\newblock Randomized gossip algorithms.
\newblock {\em IEEE/ACM Transactions on Networking}, 14:2508--2530, 2006.

\bibitem[Bor15]{B15}
Charles Bordenave.
\newblock {A new proof of Friedman's second eigenvalue Theorem and its
  extension to random lifts}.
\newblock {\em arXiv preprint arXiv:1502.04482}, 2015.

\bibitem[CO10]{coja2010graph}
A.~Coja-Oghlan.
\newblock Graph partitioning via adaptive spectral techniques.
\newblock {\em Combinatorics, Probability and Computing}, 19(02):227--284,
  2010.

\bibitem[CPV93]{crisanti_products_1993}
Andrea Crisanti, Giovanni Paladin, and Angelo Vulpiani.
\newblock {\em Products of {Random} {Matrices}}, volume 104 of {\em Springer
  {Series} in {Solid}-{State} {Sciences}}.
\newblock Springer Berlin Heidelberg, Berlin, Heidelberg, 1993.
\newblock DOI: 10.1007/978-3-642-84942-8.

\bibitem[DF89]{dyer1989solution}
M.~E. Dyer and A.~M. Frieze.
\newblock The solution of some random {NP}-hard problems in polynomial expected
  time.
\newblock {\em Journal of Algorithms}, 10(4):451--489, 1989.

\bibitem[DKMZ11]{decelle_asymptotic_2011}
A.~Decelle, F.~Krzakala, C.~Moore, and L.~Zdeborov{\'a}.
\newblock Asymptotic analysis of the stochastic block model for modular
  networks and its algorithmic applications.
\newblock {\em Physical Review E}, 84(6):066106, 2011.

\bibitem[DP09]{dubhashi2009concentration}
D.~P. Dubhashi and A.~Panconesi.
\newblock {\em Concentration of measure for the analysis of randomized
  algorithms}.
\newblock Cambridge University Press, 2009.

\bibitem[HLL83]{holland1983stochastic}
P.~W. Holland, K.~B. Laskey, and S.~Leinhardt.
\newblock Stochastic blockmodels: First steps.
\newblock {\em Social networks}, 5(2):109--137, 1983.

\bibitem[Ips15]{ipsen_products_2015}
J.~R. Ipsen.
\newblock Products of {Independent} {Gaussian} {Random} {Matrices}.
\newblock {\em arXiv:1510.06128 [math-ph]}, October 2015.
\newblock arXiv: 1510.06128.

\bibitem[JS98]{jerrum_metropolis_1998}
M.~Jerrum and G.~B. Sorkin.
\newblock The metropolis algorithm for graph bisection.
\newblock {\em Discrete Applied Mathematics}, 82(1), 1998.

\bibitem[KM04]{kempe2004decentralized}
D.~Kempe and F.~McSherry.
\newblock A decentralized algorithm for spectral analysis.
\newblock In {\em Proc. of the ACM Symposium on Theory of Computing (STOC)},
  pages 561--568. ACM, 2004.

\bibitem[Mas14]{massoulie_community_2014}
L.~Massoulie.
\newblock {Community Detection Thresholds and the Weak Ramanujan Property}.
\newblock In {\em Proc. of the {ACM} Symposium on Theory of Computing (STOC)},
  pages 694--703. ACM, 2014.

\bibitem[McS01]{mcsherry2001spectral}
F.~McSherry.
\newblock Spectral partitioning of random graphs.
\newblock In {\em Proc. of the IEEE Symposium on Foundations of Computer
  Science (FOCS)}, pages 529--537, 2001.

\bibitem[MNS13]{mossel_proof_2013}
E.~Mossel, J.~Neeman, and A.~Sly.
\newblock A proof of the block model threshold conjecture.
\newblock {\em arXiv preprint arXiv:1311.4115}, 2013.

\bibitem[MNS14]{mossel_reconstruction_2014}
E.~Mossel, J.~Neeman, and A.~Sly.
\newblock Reconstruction and estimation in the planted partition model.
\newblock {\em Probability Theory and Related Fields}, 162(3-4):431--461, 2014.

\bibitem[MNS16]{MNS14}
E.~Mossel, J.~Neeman, and A.~Sly.
\newblock Belief propagation, robust reconstruction and optimal recovery of
  block models.
\newblock {\em The Annals of Applied Probability}, 26(4):2211--2256, 2016.

\bibitem[SZ16]{SZ16}
H.~Sun and L.~Zanetti.
\newblock Distributed graph clustering by load balancing.
\newblock {\em CoRR}, abs/1607.04984, 2016.

\bibitem[SZ17]{SZ17}
H.~Sun and L.~Zanetti.
\newblock Distributed graph clustering and sparsification.
\newblock {\em CoRR}, abs/1711.01262, 2017.

\bibitem[WWA12]{WWA12}
M.J. Williams, R.M. Whitaker, and S.M. Allen.
\newblock Decentralised detection of periodic encounter communities in
  opportunistic networks.
\newblock {\em Ad Hoc Networks}, 10(8):1544--1556, 2012.

\end{thebibliography}

\onecolumn
\newpage
\appendix
\begin{center}
\Huge{\textbf{Appendix}}
\end{center}

\section{Tools from linear algebra}

\subsection{Projections on the main eigenspaces}
\begin{lemma}[Projection on the first two eigenvectors]
	\label{lem:toptwoprojections}
	For all $\epsilon \in (0,1)$, for a random $\bx \in \{-1,1\}^n$, with 
	probability at least $1-\bigO(\epsilon)$ we have,
	\begin{align*}
		& \Prob{}{|\bx \cdot \bone \pm \bx \cdot \bchi|  \geq \epsilon \cdot \sqrt{n}} \geq 1 - \bigO(\epsilon)
		\quad \text{and}\\
		& \Prob{}{|\bx \cdot \bone| \leq |\bx \cdot \bchi|  - \epsilon \cdot \sqrt{n} 
		\,\,|\,\, |\bx \cdot \bone \pm \bx \cdot \bchi|  \geq \epsilon \cdot \sqrt{n}} = \frac{1}{2} .
	\end{align*}
\end{lemma}
\begin{proof}
Note that $\bx \cdot (\bone + \bchi) = 2 \bx \cdot \bone_{V_1}$ and $\bx \cdot (\bone - \bchi) = 2 \bx \cdot \bone_{V_2}$.   
Using properties of the binomial distribution, it is easy to see that 
\[ 
    \Pr{( `` |\bx \cdot \bone + \bx \cdot \bchi| \geq \epsilon \sqrt{n} "
        \wedge 
        `` | \bx \cdot \bone - \bx \cdot \bchi| \geq \epsilon \sqrt{n} " )} 
        \geq  1 - \frac{4}{\sqrt{2\pi}}\epsilon .
\]
The above event implies $\left| |\bx \cdot \bone| - |\bx \cdot \bchi| \right| \geq \epsilon \sqrt{n}$. 
Since $\bx \cdot \bone$ and $\bx \cdot \bchi$ are independent sums of Rademacher random variables, 
they have the same chances of being positive or negative, thus with probability
at least $\frac{1}{2}$ we will have $|\bx \cdot \bone| \leq |\bx \cdot \bchi|$.
\end{proof}

\subsection{Properties of the spectrum of the main matrices}
We consider here Algorithm \aveproc{$(\delta)$} assuming  $\delta = 
1/2$ and recall the main notations: 

\begin{itemize}
\item $A$ is the adjacency matrix of the clustered graph $G((V_1,V_2); E)$, with $|V_h|=n/2$, 
 $m = |E|$ and  $m_{1,2}= |E(V_1,V_2)| $ is the  number of edges in the cut $(V_1,V_2)$;
\item $D$ is the diagonal matrix with the degrees of nodes;
\item $L = D - A$ is the Laplacian matrix;
\item $\mathcal{L} = D^{-1/2} L D^{-1/2}$is the  normalized Laplacian;
\item $P = D^{-1} A$ is the  transition matrix;
\item For each node $i = 1, \dots, n$, we name $d_i = a_i + b_i$ the degree of
node $i$, where $a_i$ is the number of neighbors in its own block and $b_i$
is the number of neighbors in the other block
\end{itemize}

The next facts are often used in our analysis.

\begin{obs}
Let $W = (W(i,j)) \sim \cM$ be the random matrix of one step of the averaging process,
then
\[
w_{i,j} = 
\left\{
\begin{array}{cl}
0 & \mbox{ if } i \neq j \mbox{ and } \{i,j\} \mbox{ not sampled} \\
1/2 & \mbox{ if $i = j$ and some edge incident on $i$ sampled  
or $i \neq j$ and edge $\{i,j\}$ sampled} \\
1 & \mbox{ if } i = j \mbox{ and } i \mbox{ not incident to a sampled edge. }
\end{array}
\right.
\]
\end{obs}

\begin{obs}\label{obs:expecW}
The expectation of $W$ is 
\[
\avgW := \Expec{}{W} = I - \frac{1}{2m} L \, .
\]
\end{obs}
\begin{proof}
\begin{enumerate}
\item If $i \neq j$ then $\Expec{}{W(i,j)} = \frac{1}{2m}$
\item If $i = j$ then $\Expec{}{W(i,j)} = 1 \left(1 - \frac{d_i}{m}\right) + 
\frac{1}{2} \frac{d_i}{m}
= 1 - \frac{d_i}{2m}$
\end{enumerate}
\end{proof}

As for the spectrum of the main matrices above, defined by  the averaging process, we have the following useful
properties we can derive from standard spectral algebra.

\begin{obs}
\begin{enumerate}
\item Il $\lambda$ is an egenvalue of $L$ then $1 - \lambda/(2m)$ is an 
eigenvalue of $\avgW$.
\item Vector $\bone$ is an eigenvector of $\avgW$.
\item If the underlying graph $G$ is $(n,d,b)$-regular then $\bchi$ is an 
eigenvector of $\avgW$.
\end{enumerate}
\end{obs}

\begin{obs}\label{obs:eigenvaluesrelations}
Consider a graph $G$ with adjacency matrix $A$ and diagonal degree matrix $D$.
\begin{enumerate} 
\item Let $1 = \lambda_1^P \geqslant \lambda_2^P \geqslant \cdots \geqslant
\lambda_n^P$ be the eigenvalues of the transition matrix $P = D^{-1} A$ and let
$0 = \lambda_1 \leqslant \lambda_1 \leqslant \cdots
\leqslant \lambda_1$ be the eigenvalues of the normalized Laplacian 
$\mathcal{L} = I - D^{-1/2} A D^{-1/2}$. For every $i = 1, \dots, n$ it holds
that
\[
\lambda_i = 1 - \lambda_i^{P} \, .
\]
\item Let $1 = \bar\lambda_1 \geqslant \bar\lambda_2 \geqslant 
\cdots \geqslant \bar \lambda_n$ be the eigenvalues of   
$\avgW = I - L/(2m)$ and let $0 = \lambda_1^{L} \leqslant \lambda_2^{L} 
\leqslant \cdots \leqslant \lambda_n^{L}$ be the eigenvalues of the Laplacian matrix 
$L = D - A$. For every $i = 1, \dots, n$ it holds
that
\[
\lambda_i^{L} = 2m(1 - \bar \lambda_i) \, .
\]
\item For every $i = 1, \dots, n$, the  eigenvalues of $L$ and $\mathcal{L}$ satisfy
\[
d_\text{min} \lambda_i 
\leqslant \lambda_i^{L} \leqslant
d_\text{max} \lambda_i \, ,
\]
where $d_\text{min}$ and $d_\text{max}$ are the minimum and the maximum degree
of the nodes, respectively.
\end{enumerate}
\end{obs}

As a consequence of the above relationships among the eigenvalues of matrices $L$, $\mathcal{L}$, and 
$\avgW$, we easily get the following useful bounds.

\begin{obs} \label{obs:lambdarelations}
Let $G$ be an $(n,d,\gamma)$-clustered graph   and let 
$\bar \lambda_i$ and $\lambda_i$, for $i = 1, \dots, n$, be
the eigenvalues of   $\avgW$ and   $\mathcal{L}$, 
respectively, in non-decreasing order. It holds that
\[
\frac{d}{2m}(1 - 2 \gamma)\left( \lambda_3 - 
\lambda_2 \right)
\leqslant \bar\lambda_2 - \bar\lambda_3 \leqslant 
\frac{d}{2m}(1 + 2\gamma)\left( \lambda_3 - 
\lambda_2 \right) \, .
\]
\end{obs}
\begin{proof}
We have:
\[
	\bar{\lambda}_2 - \bar{\lambda}_3 = \frac{\lambda_3^L - 
	\lambda_2^L}{2m}.
\]
To derive the lower bound, we write:
\begin{align}
	&\frac{\lambda_3^L - 
	\lambda_2^L}{2m}\ge\frac{\lambda_3d_\text{min} - 
	\lambda_2d_\text{max}}{2m}\ge\frac{\lambda_3d(1 - \gamma) - 
	\lambda_2d(1 + \gamma)}{2m} = 
	\frac{(\lambda_3 - \lambda_2)d -\gamma 
	d(\lambda_3 + \lambda_2)}{2m}\\
	&\ge\frac{(1 - 
	2\gamma)(\lambda_3 - \lambda_2)}{2m}.
\end{align}
Here, the first inequality is a direct consequence of Observation 
\ref{obs:eigenvaluesrelations}, the second follows from the definition 
of $(n,d,\gamma)$-clustered graph, while the last inequality follows 
since $\lambda_3 + \lambda_2\le 
2(\lambda_3 - \lambda_2)$, whenever
$\lambda_3 \ge 
3\lambda_2$.\footnote{Note that the latter condition 
holds, since the hypotheses of Theorem \ref{thm:mainavg} state that 
$\lambda_3 - \lambda_2 = \Omega(1)$, while 
the conditions on the cut implies that the graph's conductance is 
$o(1)$. The condition $\lambda_3 \ge 
3\lambda_2$ is thus a consequence of Cheeger's 
inequality.}
The upper bound is derived in the same way, again using $\lambda_3 \ge 
3\lambda_2$.
\end{proof}

\begin{obs}
Since the   random edges sequentially  selected by the process are mutually independent,
 the expected state of the process at time $t$ can be written as :
 
 \begin{equation} \label{eq:expectedstate}
 \Expec{}{\bx^{(t)}} \, = \ \Expec{}{W_t\cdot \ldots \cdot W_1 \cdot \bx} = \left(\Expec{}{W}\right)^{t} \bx \, = \, \avgW^t \cdot \bx \, .
 \end{equation}
\end{obs}

\section{Proofs for Section \ref{se:expec}} \label{apx:expec-proofs}

\subsection{Proof of Theorem~\ref{thm:mainavg}: Technical lemmas}\label{subsubse:lemmas}
The next lemma decomposes the state of the system at time $t$, 
explicitely identifying components parallel to $\bone$ and $\bchi$ 
respectively.
\begin{lemma}[Main decomposition]\label{lemma:expectedvalueimplicit}
Let $\bx \in \{-1,1\}^n$ be an arbitrary initial vector of values and, for 
$h = 1, 2$, let $\mu_h = \mu_h(\bx) = (2/n) \sum_{i \in V_h} \bx(i)$ be the average
of the initial values in block $h$. The expected vector of values at round $t$
conditional on the initial vector being $\bx$ can be written as
\[
\Expec{}{\bx^{(t)} \,|\, \bx^{(0)} = \bx} 
= \alpha_1 \bone + \alpha_2 \bar{\lambda}_2^t \bchi 
+ \alpha_2 \sqrt{n} \, \bar{\lambda}_2^t \bfied_\perp + \be^{(t)} \, ,
\]
where
\[
\alpha_1 = \frac{\mu_1 + \mu_2}{2} \, , 
\qquad\quad 
\alpha_2 = \frac{1}{\|\bfied - \bfied_\perp\|^2} \left( \frac{\mu_1 - \mu_2}{2}
- \frac{\langle \bx, \bfied_\perp\rangle}{\sqrt{n}} \right)
\]
and, moreover,  
\[ \left\|\be^{(t)} \right\| \leqslant 
\bar{\lambda}_3^t \sqrt{n} \, . \]
\end{lemma}
\begin{proof}
Let $1 = \bar{\lambda}_1 > \bar{\lambda}_2 \geqslant \cdots \geqslant 
\bar{\lambda}_n$ be the eigenvalues of $\avgW$ and let $\bw_1 = 
\bone / \sqrt{n}, \bw_2, \dots, \bw_n$ be a basis of orthonormal eigenvectors of
$\avgW$, so that we can write 
\[
\Expec{}{\bx^{(t)} \,|\, \bx^{(0)} = \bx} 
= \sum_{i=1}^n \bar{\lambda}_i^t \langle \bx, \bw_i \rangle  \bw_i\,.
\]
Since $\bar{\lambda}_1 = 1$ and $\bw_1 = \bone/\sqrt{n}$, we have that $\langle \bx, 
\, \bw_1 \rangle \bw_1 = (1/n) \langle \bx, \, \bone \rangle \bone$. Hence, 
$\alpha_1 = (1/n) \sum_{i \in V} \bx(i) = (\mu_1 + 
\mu_2) / 2$.

\smallskip\noindent
Since $\bw_2 = (\bfied - \bfied_\perp) / \| \bfied - \bfied_\perp \|$ we have
that
\begin{align*}
\langle \bx, \bw_2 \rangle \bw_2 
& = \left\langle \bx, \frac{\bfied - \bfied_\perp}{\| \bfied - \bfied_\perp \|}
\right\rangle \frac{\bfied - \bfied_\perp}{\| \bfied - \bfied_\perp \|}
= \frac{\langle \bx, \, \bfied - \bfied_\perp \rangle \bfied 
- \langle \bx, \, \bfied - \bfied_\perp \rangle \bfied_\perp}{\| \bfied - 
\bfied_\perp \|^2}\\
& = \frac{\langle \bx, \, \bfied - \bfied_\perp \rangle}{\| \bfied - 
\bfied_\perp \|^2\sqrt{n}}\, \bchi 
- \frac{\langle \bx, \, \bfied - \bfied_\perp \rangle}{\| \bfied - 
\bfied_\perp \|^2} \bfied_\perp 
\end{align*}
Hence,
\begin{align*}
\alpha_2 & =  \frac{\langle \bx, \, \bfied - \bfied_\perp \rangle}{\| \bfied - 
\bfied_\perp \|^2\sqrt{n}}
= \frac{1}{\| \bfied - \bfied_\perp \|^2}
\left(
\frac{\langle \bx, \, \bfied \rangle}{\sqrt{n}} 
- \frac{\langle \bx, \, \bfied_\perp \rangle}{\sqrt{n}} 
\right) \\
& = \frac{1}{\| \bfied - \bfied_\perp \|^2}
\left(
\frac{\langle \bx, \, \bchi \rangle}{n} 
- \frac{\langle \bx, \, \bfied_\perp \rangle}{\sqrt{n}} 
\right)
= \frac{1}{\| \bfied - \bfied_\perp \|^2}
\left(
\frac{\mu_1 - \mu_2}{2} 
- \frac{\langle \bx, \, \bfied_\perp \rangle}{\sqrt{n}} 
\right)
\end{align*}
Finally, the bound on $\left\|\be^{(t)} \right\|$ easily follows since, by definition, $\be\perp \bw_1, \bw_2$ and $\bar{\lambda}_3 \geq
\max\{ \bar{\lambda}_i \, | \, i \geq 4 \}$.
\end{proof}


\begin{lemma} \label{lm:f-bound}
Recall that we name $m_{1,2} = |E(V_1,V_2)|$ the size of the cut. It holds that
\[
\| \bfied_\perp \|^2 \leqslant \frac{2}{\bar{\lambda}_2 - \bar{\lambda}_3} 
\cdot \frac{m_{1,2}}{nm}\,.
\]
\end{lemma}
\begin{proof}
Observe that since $\bfied$ is orthogonal to $\bone$, we can write 
$\bfied^\intercal \avgW \bfied$ as
\begin{align}\label{eq:fiedeigendec}
\bfied^\intercal \avgW \bfied 
& = \left(\bfied_\parallel + \bfied_\perp \right)^\intercal 
\avgW \left(\bfied_\parallel + \bfied_\perp \right) 
= \bfied_\parallel^\intercal \avgW \bfied_\parallel
+ \bfied_\perp^\intercal \avgW \bfied_\perp
+ 2 \bfied_\parallel^\intercal \avgW \bfied_\perp \nonumber \\ 
& = \bfied_\parallel^\intercal \avgW \bfied_\parallel + 
\bfied_\perp^\intercal \avgW \bfied_\perp 
= \bar{\lambda}_2 \|\bfied_\parallel\|^2 + \bfied_\perp^\intercal \avgW
\bfied_\perp
\end{align}
where we used the fact that $\avgW$ is symmetric, thus 
$\bfied_\perp^\intercal \avgW \bfied_\parallel 
= \bfied_\parallel^\intercal \avgW \bfied_\perp = 0$, and the fact that
$\bfied_\parallel$ is in the second eigenspace of $\avgW$, thus
$\bfied_\parallel^\intercal \avgW \bfied_\parallel = \bar{\lambda}_2 
\|\bfied_\parallel\|^2$.
Moreover, since $\bfied_\perp$ is orthogonal to the first two eigenspaces of 
$\avgW$, the third eigenvalue of $\avgW$ is $\bar{\lambda}_3 = 
\sup_{\bx \perp \bone, \bw_2} \frac{\bx^\intercal \avgW \bx}{\|\bx\|^2} 
\geqslant \frac{\bfied_\perp^\intercal \avgW 
\bfied_\perp}{\|\bfied_\perp\|^2}$ from~\eqref{eq:fiedeigendec} it follows that
\begin{equation}\label{eq:fiedub}
\bfied^\intercal \avgW \bfied 
\leqslant \bar{\lambda}_2 \|\bfied_\parallel\|^2 + \bar{\lambda}_3 \|\bfied_\perp\|^2
\end{equation}

Observe that, since $\avgW = I - L/(2m) = I - \frac{D - A}{2m}$ we can
also write $\bfied^\intercal \avgW \bfied$ as a function of  $m_{1,2}$, indeed
\begin{align}\label{eq:fiedcut}
\bfied^\intercal \avgW \bfied 
= 1 - \frac{1}{2m} \left(\bfied^\intercal 
D \bfied  - \bfied^\intercal A \bfied \right)
= 1 - \frac{1}{2mn} \left(\bchi^\intercal 
D \bchi  - \bchi^\intercal A \bchi \right)
= 1 - 2\frac{m_{1,2}}{m n} 
\end{align}
where we used the fact that $\bchi^\intercal D \bchi = 2m$, the fact that
$\bchi^\intercal A \bchi = \sum_{i} a_i - \sum_i b_i = 2m - 4 m_{1,2}$.

From~\eqref{eq:fiedub} and \eqref{eq:fiedcut} and the fact that 
$1 = \|\bfied \|^2 = \|\bfied_\parallel\|^2 + \|\bfied_\perp\|^2$ we have
\[
1 - 2\frac{m_{1,2}}{m n} \leqslant \bar{\lambda}_2 \|\bfied_\parallel\|^2 
+ \bar{\lambda}_3 \|\bfied_\perp\|^2 = \bar{\lambda}_2 - (\bar{\lambda}_2 - \bar{\lambda}_3) 
\|\bfied_\perp \|^2 
\]
and thus
\[
\|\bfied_\perp \|^2 \leqslant \frac{2 \frac{m_{1,2}}{m n} 
- (1 -\bar{\lambda}_2)}{\bar{\lambda}_2 - \bar{\lambda}_3}
\leqslant \frac{2}{\bar{\lambda}_2 - \bar{\lambda}_3} \cdot \frac{m_{1,2}}{nm} \, .
\]
\end{proof}

\begin{definition}[Bad nodes]\label{def:badnodes}
We say a node $u \in [n]$ is $\varepsilon$-\emph{bad} if $|\bfied_{\perp,u}| \geq 
\varepsilon/\sqrt n$, for some $\varepsilon > 0$ and we call $B_\varepsilon$ the set of 
$\varepsilon$-bad nodes.
\end{definition}
Notice that the property of being a bad node only depends on the graph and on 
the protocol, not on the ``execution'' of the protocol.

From the Lemma~\ref{lm:f-bound}, an upper bound on the number of $\varepsilon$-bad 
nodes easily follows.
\begin{cor}[Number of bad nodes]\label{cor:no_bad}
The number $|B_\varepsilon|$ of $\varepsilon$-bad nodes is upper bounded by
\[
|B_\varepsilon| \le \frac{2m_{1, 2}}{\varepsilon^2(\bar{\lambda}_2 - \bar{\lambda}_3)m} \, .
\]
\end{cor}
\begin{proof}
Assume $B_\varepsilon$ vertices satisfy $(\bfied_\perp)_i > \frac{\varepsilon}{\sqrt{n}}$.
Lemma~\ref{lm:f-bound} implies:
\begin{align*}
& \frac{\varepsilon^2}{n}B_\varepsilon \le\frac{2}{\bar{\lambda}_2 - \bar{\lambda}_3} \cdot 
\frac{m_{1,2}}{nm},
\end{align*}
from which the thesis follows.
\end{proof}

\begin{lemma}[Monotonicity property]\label{lemma:incrdecr_crit}
Let $G$ be a connected graph, let $\bar{\lambda}_i$ for $i = 1, \dots, n$ be 
the eigenvalues of matrix $\avgW$, let $\varepsilon$ be such that $0 < \varepsilon
< 1$, and let $\bx \in \{-1,1\}^n$ be an arbitrary initial vector such that 
$|\alpha_2| = |\alpha_2(\bx)| > 0$. Then, for every node $u \notin B_\varepsilon$ and for 
any round $t$ such that 
$t \geqslant 3 \log \left(\frac{n}{1-\varepsilon} \right) 
/ \log(\bar{\lambda}_2 /\bar{\lambda}_3)$,
it holds that
\begin{equation} \label{eq:signblock}
\sgn\left( \Expec{}{\bx^{(t-1)}_u \,|\, \bx^{(0)} = \bx} 
- \Expec{}{\bx^{(t)}_u \,|\, \bx^{(0)} = \bx} \right) 
= \sgn\left( \alpha_2 \bchi_u \right) \ .
\end{equation}
\end{lemma}
\begin{proof}
Thanks to Lemma~\ref{lemma:expectedvalueimplicit}, the expected difference of
the value of a node in two consecutive rounds is
\[
\Expec{}{\bx^{(t-1)}_u \,|\, \bx^{(0)} = \bx} 
- \Expec{}{\bx^{(t)}_u \,|\, \bx^{(0)} = \bx} 
= \alpha_2 \bar{\lambda}_2^{(t-1)}(1 - 
\bar{\lambda}_2)[\bchi_u -
\sqrt{n} \, \bfied_{\perp,u}] + \be^{(t-1)} - \be^{(t)}
\]
From the above equation we get that, as soon as 
\[ 
|\be^{(t-1)}_u - \be^{(t)}_u| < |\alpha_2| \, \bar{\lambda}_2^{t} 
(1 - \bar{\lambda}_2) | \bchi_u \pm \varepsilon | , 
\] 
the sign of the expected difference between two consecutive values of a non 
$\varepsilon$-bad node $u$ indicates the community  node $u$ belongs to. Moreover,
since $|\bchi_u \pm \varepsilon| \geqslant 
1 - \varepsilon$ and $|\be^{(t-1)}_u - \be^{(t)}_u| \leqslant 2 \, 
\bar{\lambda}_3^t$, the above sign property turns out to be true for every 
round $t$ such that 
\begin{equation}\label{eq:lbont_incrdec_crit}
t \geqslant \log \left( \frac{2}{|\alpha_2| (1 - \bar{\lambda}_2) (1 -\varepsilon)}
\right) / \log(\bar{\lambda}_2 / \bar{\lambda}_3)
\end{equation}
The thesis thus follows from the fact that, if $|\alpha_2|> 0$ then it is at least 
$1/n$ and if the graph is connected then $1 - \bar{\lambda}_2 \geqslant 1/n$.
Hence, any  \[ t \geqslant 3 \log \left(\frac{n}{1-\varepsilon} \right) 
/ \log(\bar{\lambda}_2 /\bar{\lambda}_3) \] satisfies \eqref{eq:lbont_incrdec_crit}.
\end{proof}

\begin{lemma}[Sign property]\label{lemma:sign_crit}
Let $G$ be a connected graph, let $\bar{\lambda}_i$ for $i = 1, \dots, n$ be 
the eigenvalues of matrix $\avgW$, let $\varepsilon$ be such that $0 < \varepsilon
< 1$, and let $\bx \in \{-1,1\}^n$ be an arbitrary initial vector such that 
$\alpha_1 = \alpha_1(\bx)$ and $\alpha_2 = \alpha_2(\bx)$ satisfy $|\alpha_2| 
> 2 |\alpha_1| / (1-\varepsilon)$. Then, for every node $u \notin B_\varepsilon$ and for 
any round $t$ such that 
\[
\frac{1}{\log(1/\bar{\lambda}_3)} \log (n / |\alpha_1|)
\leqslant t \leqslant
\frac{1}{\log (1/\bar{\lambda}_2)} \log\left(\frac{|\alpha_2|(1 - \varepsilon)}{2 
|\alpha_1|}\right) 
\]
it holds that $\sgn\left( \Expec{}{\bx^{(t)}_u \,|\, \bx^{(0)} = \bx} \right) 
= \sgn\left( \alpha_2 \bchi_u \right)$.
\end{lemma}

\begin{proof}
From~\eqref{eq:expecvalu} it follows that, if $|\alpha_1 + \be^{(t)}_u| < 
|\alpha_2| \, \bar{\lambda}_2^{t} \, | \bchi_u \pm \varepsilon |$, then the sign of
the expected value of a non $\varepsilon$-bad node $u$ indicates the block node $u$ 
belongs to. Notice that, for
\begin{equation}\label{eq:lb_timewindow_expec}
t \geqslant \frac{1}{\log(1/\bar{\lambda}_3)} \log (n / |\alpha_1|)
\end{equation}
we have that
\[
|\alpha_1 + \be^{(t)}_u| 
\leqslant |\alpha_1| + |\be^{(t)}_u|
\leqslant |\alpha_1| + n \bar{\lambda}_3^{(t)}
\leqslant 2 |\alpha_1|
\]
And for
\begin{equation}\label{eq:ub_timewindow_expec}
t \leqslant \frac{1}{\log (1/\bar{\lambda}_2)} 
\log\left(\frac{|\alpha_2|(1 - \varepsilon)}{2 |\alpha_1|}\right) 
\end{equation}
we have that
\[
2 |\alpha_1| 
\leqslant |\alpha_2| \, \bar{\lambda}_2^{t} \, (1 - \varepsilon) 
\leqslant |\alpha_2| \, \bar{\lambda}_2^{t} \, |\bchi_u \pm \varepsilon|
\]
Hence, if the time-window defined by~\eqref{eq:lb_timewindow_expec} 
and~\eqref{eq:ub_timewindow_expec} is non-empty then for all $t$ in it 
\[
\frac{1}{\log(1/\bar{\lambda}_3)} \log (n / |\alpha_1|)
\, \leqslant t \leqslant \,
\frac{1}{\log (1/\bar{\lambda}_2)} 
\log\left(\frac{|\alpha_2|(1 - \varepsilon)}{2 |\alpha_1|}\right) 
\]
the sign of the expected value of a non-bad node $u$ equals the sign of 
$\alpha_2 \bchi_u$.
\end{proof}

\begin{lemma}[Projection of the initial random vector.]\label{lm:initialrandom}
Let $\bx $ be chosen uniformly at random in $\{-1,1\}^n$. Then,  two  absolute constants $\beta_1,\beta_2 >0$ exist,
such that:
\begin{itemize}
\item  with probability at least $\beta_1$, it holds that $\alpha_2 = \alpha_2(\bx) > 0$,
\item with probability at least $\beta_2$, it holds that $|\alpha_2| 
> 2 |\alpha_1| / (1-\varepsilon)$.
\end{itemize}
\end{lemma}
\skproof
Both $\alpha_1$ and $\alpha_2$ are linear combinations of a sum of $n$ independent Rademacher random variables.
Then, the two claims can be derived easily from Lemma \ref{lem:toptwoprojections}.
\qed

\section{Proofs for Section \ref{ssec:secmom}} \label{apx:smalllambda-sign}
\subsection{Proof of Theorem \ref{thm:mom.bound}} \label{apx:proofofhmmom.bound}
From the fact that   random matrix $W \sim \mathcal{W}$ defined by one step of our averaging process is  symmetric and 
idempotent ($W^\intercal W = W$) we get the following upper bound on the 
expected squared norm of $\by + \bz$ at the next step as a function of their
squared norm at the current step. For readability sake,   in the following proofs of this section 
we use $\by'$ and $\bz'$ for random variables $\by^{(t+1)}$ and $\bz^{(t+1)}$ 
conditional on the state at round $t$ being $\bx^{(t)} = \bx = \bx_{\|} + 
\by + \bz$.

\begin{lemma}\label{lem:yplusznorm} 
Let $\bx = \bx_{\|} + \by +\bz \in [-1,1]^n$ be an arbitrary vector of states. 
After one step of Algorithm~\ref{algo:sparse_update} it holds that
\[
\Expec{}{\norm{\by^{(t+1)}+\bz^{(t+1)}}^{2} \,|\, \bx^{(t)} = \bx}
\leq \left(1-\frac{\lambda_{2}}{n}\right)\norm{\by}^{2}
+ \left(1-\frac{\lambda_{3}}{n}\right)\norm{\bz}^{2}.
\]
\end{lemma}
\begin{proof} 
Since random matrix $W$ is symmetric and idempotent, it holds that
\begin{align*}
\Expec{}{\norm{\by'+\bz'}^{2}} & =\Expec{}{\norm{W\left(\by+\bz\right)}^{2}}\\
 & =\left(\by+\bz\right)^{T}\left(I-\frac{1}{n} L\right)\left(\by+\bz\right)\\
 & =\norm{\by}^{2}+\norm{\bz}^{2}-\frac{1}{n}\by^{T} L\by-\frac{1}{n}\bz^{T} L\bz\\
 & \leq\left(1-\frac{\lambda_{2}}{n}\right)\norm{\by}^{2}+\left(1-\frac{\lambda_{3}}{n}\right)\norm{\bz}^{2}.
\end{align*}
\end{proof}

\noindent
The squared norm of $\by$ at the next step can be lower bounded as a function
of the squared norms of $\by$ and $\bz$ at the current time step as follows. 
If the underlying graph is $(n,d,b)$-regular, we can get an upper bound as well.

\begin{lemma} \label{lem:ynorm}  
Let $\bx = \bx_{\|} + \by +\bz \in [-1,1]^n$ be an arbitrary vector of states. 
After one step of Algorithm~\ref{algo:sparse_update} it holds that 
\[
\Expec{}{\norm{\by^{(t+1)}}^{2} \,|\, \bx^{(t)} = \bx}
\geqslant \left(1-\frac{2\lambda_{2}}{n}\right)\norm{\by}^{2}.
\]
Moreover, if the underlying graph $G$ is an $(n, d, b)$-clustered regular graph
with $\lambda_2 = 2b/d = o(\lambda_3 / \log n)$ we also have that 
\[
    \Expec{}{\norm{\by^{(t+1)}}^{2} \,|\, \bx^{(t)} = \bx} 
    \leqslant \left(1-\frac{2\lambda_{2}}{n}\right)\norm{\by}^{2}
    + 2 \frac{\lambda_{2}}{n^{2}}\left(\norm{\by}^{2}
    + \norm{\bz}^{2}\right).
\]
\end{lemma}
\begin{proof}
Let $\{u,v\} \in E$ be the random edge sampled at step $t$, call $W_{u,v} \sim 
\mathcal{W}$ be the corresponding random matrix with   $L_{u,v}$ be such that
$W = I - \frac{1}{2} L_{u,v}$. 
As for the lower bound we have 
\begin{align*}
\Expec{}{\norm{\by'}^{2}} & =\Expec{}{\norm{Q_2W_{u,v}\left(\by+\bz\right)}^{2}}\\
 & =\Expec{}{\norm{Q_2\left(I-\frac{1}{2}L_{u,v}\right)\left(\by+\bz\right)}^{2}}\\
 & =\Expec{}{\norm{\by-\frac{1}{2}Q_2L_{u,v}\left(\by+\bz\right)}^{2}}\\
 & =\Expec{}{\norm{\by}^{2}-\by^{T}Q_2L_{u,v}\left(\by+\bz\right)+\norm{\frac{1}{2}Q_2L_{u,v}\left(\by+\bz\right)}^{2}}\\
 & \geq\norm{\by}^{2}-\by^{T}Q_2\Expec{}{L_{u,v}}\left(\by+\bz\right)\\
 & =\norm{\by}^{2}-\frac{1}{n}\by^{T}Q_2L\by\\
 & =\norm{\by}^{2}\left(1-\frac{2\lambda_{2}}{n}\right),
\end{align*}
where the last equality follows since $Q_2$ is the projector along the 
direction of $\bv_2 = \bchi / \| \bchi \|$, which in turn is $L$'s second eigenvector.

\noindent
As for the upper bound, it holds that
\begin{align}\label{eq:ynextubound}
\Expec{}{\norm{\by'}^{2}} & =\Expec{}{\norm{Q_2W_{u,v}\left(\by+\bz\right)}^{2}} 
\nonumber\\
 & =\Expec{}{\norm{Q_2\left(I-\frac{1}{2}L_{u,v}\right)\left(\by+\bz\right)}^{2}}
\nonumber \\
 & =\Expec{}{\norm{\by-\frac{1}{2}Q_2L_{u,v}\left(\by+\bz\right)}^{2}}
\nonumber\\
& =\Expec{}{\norm{\by}^{2}-\by^{T}Q_2L_{u,v}\left(\by+\bz\right)
+\norm{\frac{1}{2}Q_2L_{u,v}\left(\by+\bz\right)}^{2}}
\nonumber\\
 & =\left(1-\frac{2\lambda_{2}}{n}\right)\norm{\by}^{2}
+\Expec{}{\left(\frac{1}{2}\bv_{2}^{T}L_{u,v}\left(\by+\bz\right)\right)^{2}}
\nonumber\\
 & =\left(1-\frac{2\lambda_{2}}{n}\right)\norm{\by}^{2}
+\frac{1}{4}\norm{\by}^{2}\Expec{}{\left(\bv_{2}^{T}L_{u,v}\bv_{2}\right)^{2}}
\nonumber\\
 & \qquad+\frac{1}{2}\Expec{}{\bv_{2}^{T}L_{u,v}\by\bv_{2}^{T}L_{u,v}\bz}
+\frac{1}{4}\Expec{}{\left(\bv_{2}^{T}L_{u,v}\bz\right)^{2}}
\nonumber \\
 & =\left(1-\frac{2\lambda_{2}}{n}\right)\norm{\by}^{2}
+\frac{1}{4}\norm{\by}^{2}\Expec{}{\left(\bv_{2}\left(u\right)
-\bv_{2} \left(v\right)\right)^{4}}
\nonumber \\
& +\frac{1}{2}\norm{\by}\Expec{}{\left(\bv_{2}\left(u\right)
-\bv_2\left(v\right)\right)^{3}\left(\bz\left(u\right)
-\bz\left(v\right)\right)}
\nonumber \\
& \qquad +\frac{1}{4}\Expec{}{\left(\bv_{2}\left(u\right)
-\bv_2\left(v\right)\right)^{2}\left(\bz\left(u\right)-\bz\left(v\right)\right)^{2}}.
\end{align}
Next, note that we have
\begin{align}\label{eq:ynextuboundA}
\Expec{}{\left(\bv_{2}\left(u\right)-\bv_2\left(v\right)\right)^{4}} & = 
\frac{1}{nd}\sum_{(u, v)\in E(V_1,V_2)}\frac{4}{n^2} = \frac{b}{dn^2} = 
\frac{\lambda_2}{2n^2},
\end{align}
where we used that $\lambda_2 = 2b/d$ and the fact that $\bv_2(u) = 1/\sqrt{n}$ 
if $u$ belongs to first community and $\bv_2(u) = -1/\sqrt{n}$ when $v$ belongs to the 
second community.
We further get that
\begin{align}\label{eq:ynextuboundB}
\Expec{}{\left(\bv_{2}\left(u\right)-\bv_2\left(v\right)\right)^{3}\left(\bz\left(u\right)-\bz\left(v\right)\right)} & = 
\frac{1}{nd}\cdot\frac{8}{n\sqrt{n}}\sum_{(u, v)\in 
E(V_1,V_2)}\left(\bz\left(u\right)-\bz\left(v\right)\right)
\nonumber\\
& = \frac{1}{nd}\cdot\frac{8}{n\sqrt{n}}\left(\sum_{u\in 
V_1}b\bz(u) - \sum_{v\in V_2}b\bz(v)\right) = 0,
\end{align} 
where it is understood that if $(u, v)$ belongs to the cut, then $u\in V_1$ and 
$v\in V_2$ and where, to derive the last equality, we recall that 
$\bz\perp\mathbf{span}\{\bone, \bchi \}$.

\noindent Finally, we get that  
\begin{align}\label{eq:ynextuboundC}
\Expec{}{\left(\bv_{2}\left(u\right)-\bv_2\left(v\right)\right)^{2}
\left(\bz\left(u\right)-\bz\left(v\right)\right)^{2}} 
& = \frac{8}{dn^2}  \sum_{(u, v)\in E(V_1,V_2)}(\bz(u) - \bz(v))^2 \\
& = \frac{8}{dn^2} \|\bz\|^2 \sum_{(u, v)\in E(V_1,V_2)}\frac{(\bz(u) 
- \bz(v))^2}{\|\bz\|^2} \nonumber \\
& \le \frac{16b}{dn^2} \|\bz\|^2
= \frac{8\lambda_2}{n^2} \|\bz\|^2,
\end{align} 
where the last inequality follows by observing that $\sum_{(u, v)\in E(V_1,V_2)}\frac{(\bz(u) - 
\bz(v))^2}{\|\bz\|^2}$ is the Rayleigh quotient of the unnormalized 
Laplacian of a bipartite $b$-regular graph and the largest possible 
eigenvalue is $2b$.
The thesis follows by using~\eqref{eq:ynextuboundA}, \eqref{eq:ynextuboundB}, and
\eqref{eq:ynextuboundC} in~\eqref{eq:ynextubound}.
\end{proof}


\begin{lemma}\label{lem:znorm}
Let $\bx = \bx_{\|} + \by +\bz \in [-1,1]^n$ be an 
arbitrary vector of states. 
After one step of Algorithm~\ref{algo:sparse_update} it holds that 
\[
\Expec{}{\norm{\bz^{(t+1)}} \,|\, \bx^{(t)} = \bx}^{2}
\leq \frac{\lambda_{2}}{n}\norm{\by}^{2}
+ \left(1-\frac{\lambda_{3}}{n}\right)\norm{\bz}^{2}.
\]
\end{lemma}
\begin{proof}
From Pythagoras' Theorem and Lemmas~\ref{lem:yplusznorm} and~\ref{lem:ynorm},
we get  
\begin{align*}
\Expec{}{\norm{\bz'}^{2}} 
& = \Expec{}{\norm{\by'+\bz'}^{2}}-\Expec{}{\norm{\by'}^{2}}\\
& \leq\left(1-\frac{\lambda_{2}}{n}\right)
\norm{\by}^{2}+\left(1-\frac{\lambda_{3}}{n}\right)
\norm{\bz}^{2}-\left(1-\frac{2\lambda_{2}}{n}\right)\norm{\by}^{2}\\
& =\frac{\lambda_{2}}{n}\norm{\by}^{2}+\left(1-\frac{\lambda_{3}}{n}\right)
\norm{\bz}^{2}.
\end{align*}

\end{proof}

\noindent
Finally, by unrolling the double recursion, we get that the expected squared 
norm of $\bz$ and $\by$ at round $t$ satisfy the following inequality.

\begin{lemma}\label{lemma:unrolling}
Let $G$ be  an $(n,d,b)$-clustered regular graph with $\lambda_2 
= \frac{2b}{d} = o\left(\lambda_3 / \log n \right)$.
For every starting state $\bx^{(0)} \in \{-1, +1\}^n$ and for every $t \in \mathbb{N}$
it holds that
\[
\Expec{}{\norm{\bz^{\left(t\right)}}^{2}} 
\leq \frac{\lambda_{2}}{\lambda_{3}}\left(1-\frac{\lambda_{2}}{n}\right)^{-t} 
\Expec{}{\norm{\by^{\left(t\right)}}^{2}} 
+\left(1-\frac{\lambda_{3}}{n}\right)^{t} \Expec{}{\|\bz^{(0)}\|^2} \, .
\]
\end{lemma}
\begin{proof}
We first prove the following inequality

\begin{align}\label{eq:zunroll}
\Expec{}{\norm{\bz^{\left(t\right)}}^{2}}
\leq\frac{\lambda_{2}}{\lambda_{3}}\max_{i}\Expec{}{\norm{\by^{\left(i\right)}}^{2}}
+\left(1-\frac{\lambda_{3}}{n}\right)^{t}\Expec{}{\norm{\bz^{\left(0\right)}}^{2}}.
\end{align}
Indeed, from Lemma~\ref{lem:znorm} we get 
\begin{align*}
\Expec{}{\norm{\bz^{\left(t\right)}}^{2}} & 
\leq\frac{\lambda_{2}}{n}\Expec{}{\norm{\by^{\left(t-1\right)}}^{2}}
+ \left(1-\frac{\lambda_{3}}{n}\right)\Expec{}{\norm{\bz^{\left(t-1\right)}}^{2}}\\
& \leq \frac{\lambda_{2}}{n}\sum_{i\leq t-1}\Expec{}{\norm{\by^{\left(i\right)}}^{2}}
\left(1-\frac{\lambda_{3}}{n}\right)^{t-1-i}
+ \left(1-\frac{\lambda_{3}}{n}\right)^{t}\Expec{}{\norm{\bz^{\left(0\right)}}^{2}} \\
& \leq \frac{\lambda_{2}}{n}\max_{i \leq t-1}\Expec{}{\norm{\by^{\left(i\right)}}^{2}}
\sum_{i}\left(1-\frac{\lambda_{3}}{n}\right)^{t-1-i}
+ \left(1-\frac{\lambda_{3}}{n}\right)^{t}\Expec{}{\norm{\bz^{\left(0\right)}}^{2}} \\
& = \frac{\lambda_{2}}{n}\max_{i\leq t-1}\Expec{}{\norm{\by^{\left(i\right)}}^{2}}
\frac{1-\left(1-\frac{\lambda_{3}}{n}\right)^{t}}{\frac{\lambda_{3}}{n}}
+ \left(1-\frac{\lambda_{3}}{n}\right)^{t}\Expec{}{\norm{\bz^{\left(0\right)}}^{2}} \\
& \leq \frac{\lambda_{2}}{\lambda_{3}}\max_{i \leq t-1} 
\Expec{}{\norm{\by^{\left(i\right)}}^{2}}+\left(1-\frac{\lambda_{3}}{n}\right)^{t}
\Expec{}{\norm{\bz^{\left(0\right)}}^{2}}.
\end{align*}

\noindent
Next we observe from Lemma~\ref{lem:ynorm}, for each $i$ we have 
\[
\Expec{}{\norm{\by^{\left(t\right)}}^{2}}
\geq\left(1-\frac{2\lambda_{2}}{n}\right)^{i}\Expec{}{\norm{\by^{\left(t-i\right)}}^{2}}
\]
which means that 
\[
\Expec{}{\norm{\by^{\left(t-i\right)}}^{2}}
\leq\left(1-\frac{2\lambda_{2}}{n}\right)^{-i}
\Expec{}{\norm{\by^{\left(t\right)}}^{2}}
\leq\left(1-\frac{2\lambda_{2}}{n}\right)^{-t}\Expec{}{\norm{\by^{\left(t\right)}}^{2}}.
\]
Hence,
\begin{align}\label{eq:ymaxub}
\max_{i \leq t-1}\Expec{}{\norm{\by^{\left(i\right)}}^{2}}
\leq\left(1-\frac{\lambda_{2}}{n}\right)^{-t}
\Expec{}{\norm{\by^{\left(t\right)}}^{2}}.
\end{align}
The thesis follows by using~\eqref{eq:ymaxub} in~\eqref{eq:zunroll}.
\end{proof}


\subsubsection*{Wrapping up: Proof of Theorem~\ref{thm:mom.bound}}
\label{ssec:prmombound}
Since $\by^{(t)} - \by^{(0)}$ and $\bz^{(t)}$ are orthogonal we can write 
\begin{equation}\label{eq:distfrominitial}
\Expec{}{\left\| \by^{(t)} + \bz^{(t)} - \by^{(0)} \right\|^2} 
= \Expec{}{\left\| \by^{(t)} - \by^{(0)} \right\|^2} 
+ \Expec{}{\left\| \bz^{(t)} \right\|^2}.
\end{equation}
The proof proceeds by bounding the two terms above separately. 
As for the first term, observe that
\begin{align}\label{eq:triangley}
\Expec{}{\norm{\by^{(t)} - \by^{(0)}}^2\, |\by^{(0)}}
& = \Expec{}{\norm{\by^{(t)}}^2} + \norm{\by^{(0)}}^2 
- 2 \left\langle \Expec{}{\by^{(t)}}, \by^{(0)} \right\rangle \nonumber \\
& \leq \Expec{}{\norm{\by^{(t)}}^2} + \norm{\by^{(0)}}^2 
- 2 \left(1 - \frac{\lambda_2}{n}\right)^t \norm{\by^{(0)}}^2,
\end{align}
where the last inequality follows from the lower bound in
Lemma~\ref{lem:ynorm}. By taking the expectation of both sides of the previous
inequality with respect to the random choice of the initial state, we
immediately have that \eqref{eq:triangley} holds if $\norm{\by^{(0)}}^2$ is
replaced by $\Expec{}{\norm{\by^{(0)}}^2}$. Moreover, the upper bound in 
Lemma~\ref{lem:ynorm} yields
\[
    \Expec{}{\left\|\by'\right\|^2} 
    \leqslant \left( 1 - \frac{2\lambda_2}{n} \right) \E[\|\by\|^2] 
    + 2 \frac{\lambda_2}{n^2} \left(\E[\norm{\by}^2] + \E[\norm{\bz}^2] \right).
\]
Recall the upper bound on $\Expec{}{\norm{\bz^{(t-1)}}^2}$ given by 
Lemma~\ref{lemma:unrolling}, namely
\begin{equation}
    \Expec{}{\norm{\bz^{(t-1)}}^2}\le 
    \frac{\lambda_{2}}{\lambda_{3}}\left(1-\frac{\lambda_{2}}{n}\right)^{-\left(t-1\right)}
    \Expec{}{\norm{\by^{(t-1)}}^{2}}+\left(1-\frac{\lambda_{3}}{n}\right)^{t-1} \Expec{}{\|\bz^{(0)}\|^2}.
    \label{eq:EZagain}
\end{equation}

\noindent
We thus get, for $t = \bigO(n/\lambda_2)$,
\begin{equation}
\Expec{}{\norm{\by^{(t)}}^{2}} 
\leqslant f(n)\Expec{}{\norm{\by^{(t-1)}}^{2}} + 
\frac{2\lambda_2}{n^2}\left(1 - \frac{\lambda_3}{n}\right)^{t-1}\Expec{}{\|\bz^{(0)}\|^2},
\end{equation}
where 
\begin{equation}
f(n) = \left(1 - \frac{2\lambda_2}{n}\right) + 
\frac{2\lambda_2}{n^2} + 
\frac{2\lambda_2^2}{\lambda_3n^2}\left(1 - 
\frac{\lambda_2}{n}\right)^{-(t-1)} \le \left(1 - 
\frac{\lambda_2}{n}\right) \leqslant 1, \label{eq:fn}
\end{equation}
for the values of $t$ under consideration.
We can unfold the recursion above to obtain
\begin{align}
    \label{eq:ubyt}
    \Expec{}{\norm{\by^{(t)}}^{2}} 
    & \leqslant \left(1 - \frac{\lambda_2}{n} \right)^t \Expec{}{\norm{\by^{(0)}}^2} + 
    \frac{2\lambda_2}{n^2}\Expec{}{\|\bz^{(0)}\|^2}
    \sum_{j=0}^{t-1}\left(1 - \frac{\lambda_3}{n}\right)^{j} \nonumber \\
    & \leqslant \left(1 - \frac{\lambda_2}{n} \right)^t \Expec{}{\norm{\by^{(0)}}^2} 
    + \frac{2 \lambda_2}{n \lambda_3}\Expec{}{\|\bz^{(0)}\|^2}.
\end{align}

As for the second term in~\eqref{eq:distfrominitial}, from \eqref{eq:EZagain} 
we have that 
\begin{align}\label{eq:ubzetat}
    \Expec{}{\norm{\bz^{\left(t\right)}}^{2}} 
    & \leqslant \frac{\lambda_{2}}{\lambda_{3}} 
    \left(1-\frac{\lambda_{2}}{n}\right)^{-t}
    \Expec{}{\norm{\by^{\left(t\right)}}^{2}} 
    + \left(1-\frac{\lambda_{3}}{n}\right)^{t} \Expec{}{\|\bz^{(0)}\|^2} \nonumber \\ 
    & \leqslant \frac{\lambda_2}{\lambda_3}
    \Expec{}{\norm{\by^{(0)}}^2} + \frac{2}{n}\left(1-\frac{\lambda_{2}}{n}\right)^{-t} 
\left(\frac{\lambda_2}{\lambda_3}\right)^2 \Expec{}{\|\bz^{(0)}\|^2}
    + \left(1-\frac{\lambda_{3}}{n}\right)^{t} \Expec{}{\|\bz^{(0)}\|^2} \nonumber \\
& \leqslant \frac{\lambda_2}{\lambda_3} \Expec{}{\|\by^{(0)}\|^2}
    + \left(\frac{4}{n} \left(\frac{\lambda_2}{\lambda_3}\right)^2 + 
    \frac{1}{n^3}\right) \Expec{}{\|\bz^{(0)}\|^2},
\end{align}
where in the last inequality we used the fact that 
$\left(1 - \lambda_2 / n \right)^{-t}\le 2 $ for $t \le n / (4\lambda_2)$
and the fact that $\left(1 - \lambda_3/ n \right)^t 
\leqslant 1/n^3$ for 
$t \geqslant (3n / \lambda_3) \log n$.

Using~\eqref{eq:ubyt} in~\eqref{eq:triangley} and then~\eqref{eq:triangley} 
and~\eqref{eq:ubzetat} in~\eqref{eq:distfrominitial} we get
\begin{align*}
    & \Expec{}{\left\| \by^{(t)} + \bz^{(t)} - \by^{(0)} \right\|^2}\\
    & \leq \left[ 1 - \left(1-\frac{\lambda_2}{n} \right)^t + \frac{\lambda_2}{\lambda_3} \right] 
    \Expec{}{\|\by^{(0)}\|^2}
    + \left[\frac{2 \lambda_2}{n \lambda_3} 
    + \frac{4}{n}\left(\frac{\lambda_2}{\lambda_3} \right)^2 + \frac{1}{n^3} \right]
    \Expec{}{\|\bz^{(0)}\|^2} \\
    & \leqslant 2 \frac{\lambda_2 t}{n} \Expec{}{\|\by^{(0)}\|^2} 
    + \frac{1}{n^3} \Expec{}{\|\bz^{(0)}\|^2}.
\end{align*}
The thesis then follows from the fact that $\Expec{}{\|\by^{(0)}\|^2} = 1$ and
$\Expec{}{\|\bz^{(0)}\|^2} \leqslant 1/n$.
\qed

\subsection{Proofs of Corollary \ref{cor:goodt} and of  Equation \eqref{eq:numepsgood}} \label{apx:proofofcoroll:goodt}
 
\noindent
- As for Corollary \ref{cor:goodt}, we first note that  $t$ meets the conditions 
of Theorem~\ref{thm:mom.bound}. Moreover, $t\le c\frac{n}{\lambda_3}\log 
n$, so we immediately have:
\begin{align*}
	&\Expec{}{B_t}\le \frac n {\epsilon^2\|\by^{(0)}\|^2}\Expec{}{\| \by^{(t)} + 
    \bz^{(t)}  - \by^{(0)} \|^2}\le \frac{4\lambda_2 t}{\epsilon^2}\le 
    4c\frac{\lambda_2}{\lambda_3}n\log n,
\end{align*}
which is at most $\epsilon^2 n$, whenever 
$\frac{\lambda_2}{\lambda_3}\le \frac{\epsilon^4}{4c\log n}$.
Hence, the second claim follows directly from Markov's inequality.

\noindent
- As for Equation \eqref{eq:numepsgood}, note that the definition of $\epsilon$-bad node implies:
\begin{equation*}
    |B_t|\le\frac n {\epsilon^2\|\by^{(0)}\|^2} \| \bx^{(t)} - \bx_{\|} - \by^{(0)} \|^2 
    \stackrel{(a)}{=} \frac n {\epsilon^2\|\by^{(0)}\|^2} \| \by^{(t)} + 
    \bz^{(t)}  - \by^{(0)} \|^2,
\end{equation*}
where in $(a)$ we used \eqref{eq:decomp}. This easily implies \eqref{eq:numepsgood}.

\subsection{Proof of Lemma \ref{lem:main}} \label{apx:proofofmainlm}
In order to prove Lemma \ref{lem:main}, we need some preliminary results.

\begin{claim}\label{claim.wy}
Let $W_1,\ldots, W_t$ be denote a sequence of matrices describing $t$ 
steps of the \avg protocol that includes $c$ cross edges and $t-c$ internal edges. Then
\[ 
    \| W_t \cdots W_1 \bchi - \bchi \|^2 \leq 4 c.
\]
Furthermore, if $W_1,\ldots, W_t$ are chosen randomly according to the 
\avg protocol we have,
\begin{align*}
    \Expec{}{\| W_t \cdots W_1 \bchi - \bchi \|^2} & \leq 4\frac {tb}{d}
    \text{ and }\\
    \Prob{}{ \| W_t \cdots W_1 \bchi - \bchi \|^2 \geq 8 t \frac bd } &\leq e^{- \Omega( bt/d)} .
\end{align*}
\end{claim}

\begin{proof} 
    We have
    \begin{align*}
        \| W_t \cdots W_1 \bchi - \bchi \|^2 
        &= \| W_t \cdots W_1 \bchi \|^2 - 2 \bchi^T W_t \cdots W_1 \bchi + \| \bchi \|^2 \\
        &\leq 2 n -  2 \bchi^T W_t \cdots W_1 \bchi.
    \end{align*}
    To complete the proof, observe that for every vector $\bw$ such that $\| \bw \|_\infty \leq 1$, we have
    $ \bchi^T W_{u,v} \bw = \bchi^T \bw $
    if $(u,v)$ is an internal edge and
    \begin{align*}
        \bchi^T W_{u,v} \bw 
            &= \bchi^T \left ( \bw  +  \frac 12 (w_u - w_v) \bone_v + \frac 12 (w_v - w_u) \bone_u \right) \\
            &\geq \bchi^T \bw  - 2.
    \end{align*}
    By induction we thus get
    \[ 
        \bchi^T W_t \cdots W_1 \bchi  \geq n - 2 c ,
    \]
    which implies the first part of the lemma. The furthermore part follows by noting that the average of $c$ is $bt/d$, and that
    Chernoff bounds imply that $c$ is concentrated around its average.
\end{proof}

\begin{claim}
\label{lem:one} 
    Let $t_1 = 6\frac{n}{\lambda_3} \log n$ and assume
	$\frac{\lambda_2}{\lambda_3}\le\frac{\epsilon^4}{96\log n}$. Then 
	\begin{equation}\label{eq:always_good}
		\Prob{}{\forall t\in\{t_1, \ldots , 2t_1\}: 
	|B_t|\le\frac{48\lambda_2t}{\epsilon^3}} \ge 1 - \epsilon - \frac{\log n}{n^2}.
	\end{equation}
\end{claim}
\begin{proof} 
Recall that $\by^{(t)} + \bz^{(t)} = W_{t} \ldots 
	W_{t_1}(\by^{(t_1)} + \bz^{(t_1)})$. Then, we deterministically have:
\begin{align} 
	&\norm{ \by^{(t)} + \bz^{(t)} - \by^{(0)} }^2 = \norm{ W_{t} \ldots 
	W_{t_1}(\by^{(t_1)} + \bz^{(t_1)}) - \by^{(0)} }^2\nonumber \\
	&\le 3 \norm{ W_{t} \ldots W_{t_1}\by^{(t_1)} - \by^{(t_1)} }^2 + 3 
	\norm{ \by^{(t_1)} - \by^{(0)} }^2 + 3 \norm{W_{t} \ldots W_{t_1}\bz^{(t_1)} }^2\nonumber\\
	&\le 3 \norm{ W_{t} \ldots W_{t_1}\by^{(t_1)} - \by^{(t_1)} }^2 + 3 
	\norm{ \by^{(t_1)} - \by^{(0)} }^2 + 3 \norm{\bz^{(t_1)} }^2\nonumber\\
	& = 3\|W_{t} \ldots W_{t_1}\by^{(t_1)} - \by^{(t_1)}\|^2 + 3\norm{ \by^{(t_1)} + \bz^{(t_1)}- \by^{(0)} }^2,
	\label{eq:one}
\end{align}
where the first inequality follows from obvious manipulations, the 
second follows since $W_{t} \ldots 	W_{t_1}$ has norm one, while the 
last equality follows from Pythagoras' Theorem. 

The proof of Claim 
\ref{lem:one} next proceeds in the following steps:

\medskip
\noindent{\em 1. Upper bound  to $\norm{ \by^{(t_1)} + \bz^{(t_1)}- \by^{(0)} }^2$.}

\noindent Theorem~\ref{thm:mom.bound} and Markov's inequality immediately imply:
\begin{equation}\label{eq:first}
	\Prob{}{\norm{ \by^{(t_1)} + \bz^{(t_1)}- \by^{(0)} }^2 \le 
	\frac{4\lambda_2 t_1}{\epsilon n}\|\by^{(0)}\|^2}\ge 1 - \epsilon.
\end{equation}
Moreover, since $\norm{ \by^{(t_1)} + \bz^{(t_1)}- \by^{(0)} }^2 = \norm{ 
\by^{(t_1)} - \by^{(0)}}^2 + \norm{\bz^{(t_1)} }^2$, we have the 
following useful derivations: 
\begin{align}
	&\norm{ \by^{(t_1)} + \bz^{(t_1)}- \by^{(0)} }^2 \le 
	\frac{4\lambda_2 t_1}{\epsilon n}\|\by^{(0)}\|^2\Longrightarrow \norm{ 
	\by^{(t_1)} - \by^{(0)}}^2\le \frac{4\lambda_2 t_1}{\epsilon 
	n}\|\by^{(0)}\|^2\nonumber\\
	&\Longrightarrow \norm{ \by^{(t_1)}}^2\le\left (1 + \frac{4\lambda_2 
	t_1}{\epsilon n}\right)\norm{y^{(0)}}^2\label{eq:norm_y1}
\end{align}

\medskip
\noindent{\em 2. Upper bound to $\|W_{t} \ldots W_{t_1}\by^{(t_1)} - 
\by^{(t_1)}\|^2$, for $t \in \{t_1,\ldots , 2t_1\}$.} 

\noindent We have:
\begin{align*} 
	&\|W_{t} \ldots W_{t_1}\by^{(t_1)} - \by^{(t_1)}\|^2 =
	\|W_{t} \ldots W_1\bchi - \bchi\|^2\frac{\|\by^{(t_1)}\|^2}{n}
\end{align*}

Next, for any steps $\tau_1, 
\tau_2\ge t_1$, denote by $C_{[\tau_1, \tau_2]}$ the number of cross 
edges sampled over the steps $\tau_1, \ldots , \tau_2$. Claim~\ref{claim.wy}
then implies:
\begin{equation}\label{eq:u_steps}
	\|W_{t} \ldots W_{t_1}\by^{(t_1)} - \by^{(t_1)}\|^2 \le  
	6C_{[t_1, t]}\frac{\|\by^{(t_1)}\|^2}{n}
\end{equation}
The deterministic upper bound stated by \eqref{eq:u_steps} immediately 
implies:
\begin{align*} 
	&\Prob{}{\|W_{t} \ldots W_{t_1}\by^{(t_1)} - \by^{(t_1)}\|^2 > 
	\frac{12\lambda_2t}{n}\|\by^{(t_1)}\|^2}\le \Prob{}{C_{[1, t]} > 
	2\lambda_2t}.
\end{align*}
On the other hand, we immediately have $\Expec{}{C_{[1, 
t]}}\le\frac{2b}{d}t = \lambda_2t$. Application of Chernoff bound then 
implies:
\[
\Prob{}{C_{[1, t]} > 2\lambda_2t}\le e^{-\frac{\lambda_2t}{2}}\le 
e^{-3\log n},
\]
since $t\ge t_1 = 6\frac{n}{\lambda_3}\log n$ and $\lambda_2\ge 1/n$, 
given the the graph is regular. As a result:
\begin{equation}\label{eq:exists}
	\Prob{}{\exists t\in\{t_1, \ldots , 2t_1\}: \|W_{t} \ldots 
	W_{t_1}\by^{(t_1)} - \by^{(t_1)}\|^2 > 
	\frac{12\lambda_2t}{n}\|\by^{(t_1)}\|^2} < \frac{\log n}{n^2}.
\end{equation}
Recalling \eqref{eq:one} and using \eqref{eq:first} and \eqref{eq:exists} 
we finally obtain:
\begin{equation}
	\Prob{}{\exists t\in\{t_1, \ldots , 2t_1\}: \norm{ \by^{(t)} + 
	\bz^{(t)} - \by^{(0)} }^2 > \frac{48\lambda_2t}{\epsilon n}\|\by^{(0)}\|^2} < 
	\epsilon + \frac{\log n}{n^2}.
\end{equation}
In particular, from \eqref{eq:first}, \eqref{eq:norm_y1} and \eqref{eq:exists}, we know 
that with probability $1 - \epsilon - \frac{\log n}{n^2}$:
\begin{align*} 
	&\norm{ \by^{(t)} + \bz^{(t)} - \by^{(0)} }^2\le 3\|W_{t} \ldots W_{t_1}\by^{(t_1)} - \by^{(t_1)}\| 
	+ 3\norm{ \by^{(t_1)} + \bz^{(t_1)}- \by^{(0)} }^2\\
	&\le\frac{12\lambda_2 t_1}{\epsilon 	n}\|\by^{(0)}\|^2 + 
	\frac{36\lambda_2t}{n}\|\by^{(t_1)}\|^2 \le\frac{12\lambda_2 t}{\epsilon 	n}\|\by^{(0)}\|^2 + 
	\frac{36\lambda_2t}{n}\left(1 + \frac{4\lambda_2 t_1}{\epsilon 
		n}\right)\|\by^{(0)}\|^2\le \frac{48\lambda_2t}{\epsilon 
		n}\|\by^{(0)}\|^2.
\end{align*}
The last inequality in the derivations above follows by noting that 
$\frac{36\lambda_2t}{n}\left(1 + \frac{4\lambda_2 t_1}{\epsilon 
n}\right)\le \frac{36\lambda_2 t}{\epsilon n}$ since $t\le 2t_1\le 
4 n/\lambda_2$, which in turn follows if 
$\frac{\lambda_2}{\lambda_3}\le\frac{\epsilon^4}{96\log n}$ as we 
assume.

\medskip
\noindent{\em Upper bound on $|B_t|$.}
\eqref{eq:numepsgood} and \eqref{eq:always_good} immediately imply that, 
with probability at least $1 - \epsilon - \log n/n^2$
\[
	\forall t\in\{t_1, \ldots , 2t_1\}: 
	|B_t|\le\frac{48\lambda_2t}{\epsilon^3}.
\]
This concludes the proof of Claim \ref{lem:one} .
\end{proof}

\subsubsection*{Wrapping up: Proof of Lemma~\ref{lem:main}}
Let $t_1 = 6\frac{n}{\lambda_3} \log n$ and define $A_{t_1+1}:= V - B_{t_1+1}$ 
the complement of $B_{t_1+1}$ and, for each $t\in [t_1+1,2t_1]$, let $A_t$ 
denote the set of nodes in $A_{t_1+1}$ that have not been averaged along a 
cross edge or with a node in $B_{t}$.
Inductively, if $e_t = (u_t,v_t)$ is the edge chosen at time $t$ then 
\begin{align}
	A_t = \begin{cases} A_{t-1} & \text{ if } e_{t-1} \text{ is not a cross edge and } e_{t-1} \cap B_{t-1} = \emptyset, \\
		A_{t-1} \setminus \{u_t, v_t \} & \text{ otherwise.} 
	\end{cases}
	\label{eq:badedges}
\end{align}
We say that {\em we sampled a good edge at time $t$} in the first case, 
otherwise we say that {\em we sampled a bad edge}.
The proof proceeds in two steps:

\medskip
\noindent{\em Deterministic lower bound on the number of goods nodes.}
We note that $|A_{2t_1}|$ is a lower bound on the number of nodes that are 
$\epsilon$-good at every step $t\in [t_1, 2t_1]$. Indeed, 
every node $v$ in $A_{2 t_1}$ was $\epsilon$-good at all times
between $t_1$ and $2t_1$, because it is a node whose value was good
at time $t_1+1$, and then was averaged only with nodes $u$ on the
same side of the partition (that is, such that $\sgn (\by_u) = \sgn
(\by_v)$) and at times in which $u$ was good as well. 
Moreover, $|A_t|$ decreases by at most 2 compared to 
$|A_{t-1}|$ and this can only happen if we sample a bad edge. As a 
consequence, $|A_{2t_1}|\ge |A_{t_1}| - Z$, with $Z$ the number of bad 
edges sampled in the interval $[t_1, 2t_1]$.

\medskip
\noindent{\em Lower bound on $|A_{2t_1}|$.}
Next, we condition on the event ${\mathcal B} = (\forall t\in\{t_1, \ldots , 2t_1\}: 
|B_t|\le\frac{\lambda_3\epsilon n}{12\log n})$. From Claim 
\ref{lem:one}, this event holds with probability at least $1 - 
\epsilon$, whenever  
$\frac{\lambda_2}{\lambda_3}\le\frac{\lambda_3\epsilon^4}{6192\log^2n}$. Now,
observe that, for a given time step $t$, conditioning on the event ${\mathcal B}$ can 
only increase the probability of sampling a good edge.
As a consequence, recalling that, without conditioning we sample edges 
uniformly at random, for any integer $x > 0$ we have:
\[
\Prob{}{Z > x \,| \, \mathcal B} \le \Prob{}{\hat{Z} > x},
\]
where $\hat{Z}$ is the sum of independent Bernoulli variables with 
parameter $\frac{2b}{d} + \frac{\lambda_3\epsilon}{12\log n}$. Here, the 
first term is the probability of sampling a cross edge, while the 
second is worst-case upper bound on the probability of sampling an edge 
with an endpoint in $A_t$ and the other in $B_t$, provided that $|B_t|\le 
\frac{\lambda_3\epsilon n}{12\log n}$.
As a consequence, recalling that $\lambda_2 = 2b/d$:
\begin{align*}
	&\Expec{}{\hat{Z}} = \lambda_2 t_1 + \frac{\lambda_3\epsilon}{12\log 
	n} t_1 = 6\frac{n}{\lambda_3}\log n\left (\lambda_2 + \frac{\lambda_3\epsilon}{12\log 
	n}\right) = 6\frac{\lambda_2}{\lambda_3}n\log n + \frac{\epsilon n}{2}\le 
	\epsilon n,
\end{align*}
where the last inequality obviously follows given our assumptions on 
$\lambda_2/\lambda_3$.
At this point, a simple application of Chernoff bound allows to 
conclude that $\hat{Z}$ (and thus $Z$, when conditioned to $\mathcal 
B$), is at most $2\epsilon n$ with probability $1 - e^{-\frac{\epsilon 
n}{2}}$. So, conditioned to the event $\mathcal B$, $A_{2t_1}\ge n - 
\frac{\lambda_3 n}{12\log n} - 2\epsilon n$ with probability $1 - e^{-\frac{\epsilon 
n}{2}}$. This concludes the proof.
\qed

\subsection{Tools for the analysis of Algorithm~\ref{algo:sign_labeling}} \label{apx:toolsforcslsign}

\subsubsection*{Proof of Theorem \ref{thm:main_small} }
In the next two subsections, we provide  Lemmas \ref{lm:csh_gen} and \ref{lemma:numberunlucky}, respectively. 
The two lemmas together
easily imply Theorem \ref{thm:main_small}.

\smallskip

\subsubsection*{Lucky nodes and community-sensitive labeling}
We next consider the behavior of Algorithm~$\labelsign(T,1)$ on a graph 
$G = (V, E)$ where $V$ has a sparse cut $(V_1,V_2)$ that we wish to discover.
For every node $u$, let $\puki^{sign}_u(T)$ be the probability that node $u$ 
sets $\bh_u^{sign}$ differently from the sign of the average of the initial
values of the nodes in its own community
\[
\puki^{sign}_u(T) =
    \begin{cases}
        \Prob{}{\bh^{sign}_u \ne \sgn\left( \frac 2n \sum_{v \in V_1} \bx_v \right) } & \text{ if $u \in V_1$},\\
        \Prob{}{\bh^{sign}_u \ne \sgn\left( \frac 2n \sum_{v \in V_2} \bx_v \right) } & \text{ if $u \in V_2$},
    \end{cases}
\]
where the randomness is over the initial choice of $\bx \sim \{ -1,1\}^n$ and 
over the execution of Algorithm~$\labelsign(T,1)$. Notice that, for a given graph
$G$ and partition of the nodes $V_1, V_2$, $\psign(T)$ depends on the node $u$,
on the protocol (\textit{sign}), and on the number of activations $T$ after 
which the node sets $\hsign$ (we will omit parameter $T$ from $\psign$ when 
clear from context). 

To understand the point of this definition, consider the extreme case in which
the cut $(V_1,V_2)$ is empty, while $V_1$ and $V_2$ induce connected graphs.
In this case, the averaging process $\aveproc$ will converge to a global state
in which all nodes in $V_1$ have a local state close to the average
$\frac{2}{n} \sum_{v \in V_1} \bx_v$ and all nodes in $V_2$ have a state close
to $\frac{2}{n} \sum_{v \in V_2} \bx_v$. 
We call $\tau_u$ the (random) step in which the node $u$ sets its $\hsign$ to 
the sign of the state of $u$, i.e., the global round when $u$ achieves $T$ 
activations. If $T$ is chosen so that $\tau_u$ is large enough, we would expect
$\hsign$ to agree with the sign of $\frac{2}{n} \sum_{v \in V_1} \bx_v$ if 
$u \in V_1$ and with the sign of $\frac{2}{n} \sum_{v \in V_2} \bx_v$ if 
$u \in V_2$, with $\psign$ small for all $u$. 
It seems reasonable that a possibly weakened version of the 
considerations above should apply to graphs exhibiting a sparse,
rather than empty, cut, provided the subgraphs induced by $V_1$ and $V_2$
are good expanders. To quantitavely capture this intuition, we 
introduce the notion of a {\em (un)lucky} node.

\begin{definition}[Unlucky nodes]\label{def:unlucky}
We say that node $u$ is $\epsilon$-\emph{unlucky} if $\psign$ is larger than 
$\epsilon$. We thus define the set of $\varepsilon$-unlucky nodes as follows
\[
U^{\epsilon,sign}_{G,(V_1,V_2)} 
= \left\{ u \,|\, \psign\geq \epsilon \right\} \, .
\]
\end{definition}
We write $U^{\epsilon,sign}$ in place of $U^{\epsilon,sign}_{G,(V_1,V_2)}$,
when the underlying graph and partition of the nodes are clear from the 
context.

\begin{lemma}\label{lm:csh_gen}
Let $G = (V, E)$ be a graph, $V_1,V_2$ be a partition of $V$
and fix $\epsilon \in (0,\frac{1}{12}]$.
Then,
\labelsign$ (T, 10 \varepsilon^{-1}\log n)$ performs a community-sensitive labeling 
of $G$ according to Definition~\ref{def:senshash}, with $c_1 = 4\epsilon$, $c_2 = 1/6$ and $\gamma = 
|U^{\epsilon,T}|/n$.
\end{lemma}

\begin{proof}
The proof of the theorem relies on the mutual independence among the components 
of any label $\bh(\cdot)$ and some standard arguments.
We remark that the  independence crucially depends on the fact
that \labelsign(T,m) updates one component per interaction: the evolution of
$\bx(j_1)$ and $\bx(j_2)$ depends solely on the respective initial vector
values (which are independent), and on the sequence of sampled edges which
update component $j_1$ and $j_2$ (which are independent conditional on their
number).

Call $\ell = \frac{10 \log n}{\varepsilon}$ and denote $\bh_{V_i} := (h_{V_i}(1),
\ldots, h_{V_i}(\ell))$ where $h_{V_i}(j) := \sgn\left(\sum_{v\in V_i} x_v\right)$. 
	
We first claim that w.h.p. for every vertex $u \in V_1 \setminus U^{\epsilon}$,
$\hamm(\bh_u,\bh_{V_1}) \leq 2\varepsilon \dimsignature$.
Observe that by definition of $U^{\epsilon}$,
	\begin{equation*}
        \Expec{}{\hamm(\bh_u, \bh_{V_1})} = p(u)\dimsignature  \leqslant \varepsilon \dimsignature.
	\end{equation*}
    Since the $\dimsignature$   components are mutually  independent, the
    Chernoff bound \cite{dubhashi2009concentration} implies that
    $\hamm(\bh_u,\bh_{V_1}) \leq 2\varepsilon \dimsignature$, w.h.p.  
	A union bound over vertices in $V_1 \setminus U^{\epsilon}$ implies the claim.

    Henceforth, let us assume that $\hamm(\bh_u,\bh_{V_1}) \leq 2\varepsilon
    \dimsignature$ for each $u \in V_1 \setminus U^{\epsilon}$ and a similar claim for all
    vertices $v \in V_2 \setminus U^{\epsilon}$.
	
    As for Case (i), w.l.o.g. let us consider $u,v \in V_1 \setminus U$.  By triangle inequality, we get the desired claim.
    \begin{align*}
        \hamm(\bh_u,\bh_{v}) \leq \hamm(\bh_u,\bh_{V_1}) + \hamm(\bh_{V_1}, \bh_v) \leq 4\varepsilon \dimsignature.
    \end{align*}
    As for Case (ii), since the initial values of $x_u(j)$ ($u \in V_1 \cup V_2 $) are chosen independently
    and uniformly at random in $\{-1,1\}$, simple symmetry arguments show that   the probability of the event 
    ``$\sgn(\sum_{u \in V_1} x_u ) = \sgn(\sum_{u\in V_2} x_u)$'' is $1/2$. Hence, $\Expec{}{\hamm(\bh_{V_1},\bh_{V_2})} = \dimsignature/2$
    and from Chernoff bounds we get  that 
    \begin{equation}
        \hamm(\bh_{V_1},\bh_{V_2}) \geqslant \frac{\dimsignature}3,
        \label{eq:avg_hamm_diffcomm}
    \end{equation} with all but a probability exponentially small in $\ell$.  Henceforth, let us condition on the event that $\hamm(\bh_{V_1},\bh_{V_2}) \leq \frac{\dimsignature}{3}$

    Consider $u \in V_1$ and $v \in V_2$.  By triangle inequality, we have that
    \begin{align}
    \hamm(\bh_u,\bh_v) & \geq \hamm(\bh_{V_1},\bh_{V_2}) - \hamm(\bh_u,\bh_{V_1}) - \hamm(\bh_{V_2},\bh_v) \\
    & \geq \frac{\ell}{3} - 2\varepsilon \ell -2 \varepsilon \ell \geq \frac{\ell}{6},
    \end{align}
    concluding the proof.

\end{proof}

\smallskip

\subsubsection{A bound on the number of unlucky nodes}
In Lemma~\ref{lemma:numberunlucky} we give an upper bound on the number of
$\varepsilon$-unlucky nodes. This is the second key step toward proving 
Theorem~\ref{thm:main_small}.

\smallskip\noindent    
We first prove the following technical lemma on the range in which $\tau_v$
falls w.h.p.

\begin{lemma}  
\label{lem:stoppingtimes}
If $T > 72 \log n$ and $t_1 = 3 T n/4$ then
\[
\Prob{}{\{\tau_v \,|\, v \in V \} \subseteq [t_1,2t_1] } \geq 1- \frac{1}{n}.
\]
\end{lemma}
\begin{proof}
For each node $v$, let $X_v^{(i)} = \bone_{[v \text{ is activated at round }
i]}$.  Fix a node $v$.  By applying the Chernoff bound on the i.i.d. random
variables $\{ X_v^{(i)} \}_{i\geq 0}$, we have  
\[ 
\Prob{}{ \sum_{i = 1}^{3Tn/4} X_v^{(i)} \geq T } \leq e^{-\frac{T}{36}}
\quad \text{and}\quad \
\Prob{}{ \sum_{i = 1}^{3Tn/2} X_v^{(i)} \leq T  } \leq e^{ - \frac{T}{12}}\ .
\]
The claim follows by applying a union bound over the nodes.
\end{proof}

\begin{lemma}[Number of unlucky nodes]\label{lemma:numberunlucky}
Let $\varepsilon > 0$ be an arbitrarily small value and let $G$ be an 
$(n,d,b)$-clustered regular graph with 
$\frac{\lambda_2}{\lambda_3}\le\frac{\lambda_3\epsilon^4}{c \log^2 n}$, for a 
large enough costant $c$.
If $T = \frac{8}{\lambda_3} \log n$ then the number of
$\sqrt{\varepsilon}$-unlucky nodes is 
\[
\left|U^{\sqrt{\varepsilon}, sign} \right| \leqslant 6\sqrt \epsilon \, n.
\]
\end{lemma}

\begin{proof}
Let $L^{sign}$ be the set of nodes that freeze their sign $\bh_v^{sign}$ 
according to the sign of $\bx_{\|,v} + \by^{(0)}_v$,
\[
L^{sign} = \left\{v \in V_1 \cup V_2 \,:\, \sgn\left(\bx_v^{\tau_v}\right)
= \sgn\left(\bx_{\parallel,v} + \by_v^{(0)} \right) \right\}.
\]
We first observe that, given any $\varepsilon > 0$, if we have a lower bound
on the expected size of $L^{sign}$, namely $\Expec{}{|L^{sign}|} \geqslant 
n - \varepsilon n$, then we have an upper bound on the number of 
$\sqrt{\epsilon}$-unlucky
nodes, namely $|U^{\sqrt{\varepsilon},sign}| \leqslant 6\sqrt{\varepsilon} n$. 
Indeed,
\begin{align}
\label{eq:connectionLUeps}
\Expec{}{|L^{sign}|} 
& = \sum_{u \in U^{\sqrt{\varepsilon},sign}} \Prob{}{u \in L^{sign}} +
    \sum_{u \notin U^{\sqrt{\varepsilon},sign}} \Prob{}{u \in
    L^{sign}}\nonumber\\
& \leqslant (1-\sqrt{\varepsilon}) \left| U^{\sqrt{\varepsilon},sign} \right| 
    + n - \left|U^{\sqrt{\varepsilon},sign}\right| \nonumber\\
& = n - \sqrt{\varepsilon} \left|U^{\sqrt{\varepsilon},sign}\right|.
\end{align}
We now give a lower bound on the expected size of $L^{sign}$.
Let $\Gamma$ be the event 
\[
\Gamma = \text{``$  |\by_v^{(0)}  + \bx_{\|, v}|
\geqslant \frac {\varepsilon}{\sqrt{n}}$''}.
\]

Recall that $\tau_v$ denotes the time at
which node $v$ freezes its value of $\bh^{sign}_v$. Notice that the 
value $| \by_v^{(0)} + \bx_{\|, v}|$ does not depend on the
node $v$, only on the initial assignment.
Hence, for any node $u \in V_1 \cup V_2$ we have that
\begin{align}
\Prob{}{u \in L^{sign}} & \geqslant 
\Prob{}{\Gamma \wedge \{u \mbox{ is } \varepsilon\mbox{-good at round }
\tau_u\}} \nonumber\\
& \geqslant \Prob{}{\Gamma \wedge \{u \mbox{ is } \varepsilon\mbox{-good 
at all rounds } t \in [t_1, 2 t_1]\} \wedge \{\tau_u \in [t_1, 2 t_1] \}}\nonumber\\
& \geqslant 1 - \Prob{}{\overline{\Gamma}} 
- \Prob{}{\{u \mbox{ is \textit{not} } \varepsilon\mbox{-good at \textit{some}
round } t \in [t_1, 2 t_1]\}} - \Prob{}{\tau_u \notin [t_1, 2 t_1]} \nonumber\\
& = \Prob{}{u \mbox{ is } \varepsilon\mbox{-good at all rounds } t \in 
[t_1, 2 t_1]} 
- \Prob{}{\overline{\Gamma}} - \Prob{}{\tau_u \notin [t_1, 2 t_1]}.
\label{eq:estimateL}
\end{align}

From Lemmas~\ref{lem:toptwoprojections} and~\ref{lem:stoppingtimes} it follows that
$\Prob{}{\overline{\Gamma}} \leqslant \bigO(\varepsilon)$ and that 
$\Prob{}{\tau_u \notin [t_1, 2 t_1]} \leqslant 1/n$. 
Hence, from \eqref{eq:estimateL} the expected size of $L^{sign}$ is
\begin{align}
\Expec{}{|L^{sign}|} 
&= \sum_{u} \Prob{}{u \in L^{sign}} \nonumber\\
&\geqslant \Expec{}{|\{u \,:\, u \mbox{ is } \varepsilon\mbox{-good at all 
rounds } t \in [t_1, 2 t_1] \}|}	- \frac{4}{\sqrt{2\pi}}\varepsilon n - 1.
\label{eq:boundL}
\end{align}
Finally, from Lemma~\ref{lem:main} we have
\begin{align*}
&\Expec{}{|\{u \,:\, u \mbox{ is } \varepsilon\mbox{-good at all
rounds } t \in [t_1, 2 t_1] \}|}\ge (1 - 3\epsilon)(1 - 
\epsilon) n.
\end{align*}
Thus from \eqref{eq:boundL} and the previous inequality we get that
\begin{align*}
&\Expec{}{|L^{sign}|} \geqslant (1 - 3\epsilon)(1 - 
\epsilon) n - \frac{4}{\sqrt{2\pi}}\varepsilon n - 1 > (1 - 
6\epsilon)n
\end{align*}
and the thesis follows from~\eqref{eq:connectionLUeps}.
\end{proof}

\section{Omitted Proofs from Section \ref{ssec:pl-analysis}}
\label{app:pl-proof}

The main goal of this subsection is to prove Theorem~\ref{thm:main-concen}. 
For the sake of convenience, we rewrite the state of the averaging process as
\begin{align*}
    \bx^{(t)} = a_{||} \cdot \left(\bone / \sqrt{n}\right) + a_y(t) \cdot \left(\chi / \sqrt{n}\right) + \bz^{(t)}
\end{align*}
where $a_{||}, a_y(t) \in \R$ and $\bz^{(t)}$ is orthogonal to both $\chi$ and $\bone$. Recall also that 
 $\by^{(t)} = a_y(t) \cdot \left(\chi / \sqrt{n}\right)$ and that $a_{||}$ remains unchanged throughout the algorithm.
 Suppose now   we fix a starting vector $\bx^{(0)}$,   then    we can exactly compute the expectation of $a_y(t)$ as stated more formally below.

\begin{obs} \label{obs:onemoment}
For all $t \in \Z_{\geqs 0}$, we have $\E_{\cE} [ a_y(t) ] = \left(1 - \frac{2\delta \lambda_2}{n}\right)^t a_{y}(0).$
\end{obs}

\begin{proof}[Proof of Observation~\ref{obs:onemoment}]
To prove the above statement, it is enough to show that $\E_{(u_t, v_t)} [ a_y(t) ] = \left(1 - \frac{2\delta \lambda_2}{n}\right) a_y(t - 1)$ 
for every $t \in \N$. Indeed,  $\E_{(u_t, v_t)} [ a_y(t) ]$ can be rewritten as follows.
\begin{align*}
\E_{W \sim \cM} \left[ \frac{\chi^T W \bx^{(t - 1)}}{\sqrt{n}} \right]
    = \frac{\chi^T \bar{W} \bx^{(t-1)}}{\sqrt{n}}
    = \frac{\left(1 - \frac{2\delta \lambda_2}{n}\right)\chi^T\bx^{(t-1)}}{\sqrt{n}}
    = \left(1 - \frac{2\delta \lambda_2}{n}\right)a_y(t - 1),
\end{align*}
concluding the proof.
\end{proof}

Let $\mu(t) \triangleq \E [ a_y(t) ] = \left(1 - \frac{4\delta b}{dn}\right)^t a_{y}(0)$ be the expectation of $a_y(t)$. We will show that, if we start with $\bx^{(0)}$ such that $\|\bz^{(0)}\|^2$ is not too much larger than $n\|\by^{(0)}\|^2$, then $a_y(t)$ concentrates around $\mu(t)$, as long as $t \leqs O_{b,d,\delta}(n^2)$. Moreover, we will also show that, for $t \geqs \Omega_{b, d, \delta}(n \log n)$, $\|\bz^{(t)}\|$ becomes small compared to $\mu(t)$. This is stated more precisely below.

\begin{theorem} \label{thm:main}
Let $\beta$ be any real number such that $1 \leqs \beta \leqs \frac{d}{\varepsilon b}$. For any initial vector $\bx^{(0)}$ that satisfies satisfies $\|\bz^{(0)}\|^2 \leqs n\beta\|\by{(0)}\|^2$ and for any $t \in \left[\frac{8n}{\delta(\lambda_3 - \lambda_2)} \log\left(\frac{nd\beta}{\varepsilon b}\right), \frac{n^2 \beta}{128 \delta(\lambda_3 - \lambda_2)}\right]$, we have
\begin{align} \label{eq:main-concen-y}
\Pr_{\cE}\left[0.5\mu(t) \leqs a_y(t) \leqs 1.5\mu(t)\right] \geqs 1 - O(\varepsilon \beta b / d)
\end{align}
and
\begin{align} \label{eq:main-concen-z}
\Pr_{\cE}\left[\|\bz^{(t)}\| \leqs \left(0.5\sqrt[4]{\varepsilon b / d}\right) \mu(t)\right] \geqs 1 - O(\sqrt{\varepsilon b / d}).
\end{align}
\end{theorem}


We defer the proof of Theorem~\ref{thm:main} to Subsection~\ref{app:proofthm:main}. For now, let us turn our attention back to the proof of Theorem~\ref{thm:main-concen}.
To go from here to Theorem~\ref{thm:main-concen}, we will also need to upper bound the probability that $\|\bz^{(0)}\|^2 > n\beta\|\by^{(0)}\|^2$. More specifically, when $\bx^{(0)}$ is a random $\pm 1$ vector, we have the following bound.

\begin{prop} \label{prop:starting-condition}
For any $\beta > 0$, we have $\Pr_{\bx^{(0)} \sim \{\pm 1\}^n}\left[\|\bz^{(0)}\|^2 > n \beta \|\by^{(0)}\|^2\right] \leqs O(1/\sqrt{\beta} + 1/\sqrt{n}).$
\end{prop}

This proposition was also essentially proved in~\cite{BCNPT17}; we repeat the proof from~\cite{BCNPT17} below for completeness.

\begin{proof}[Proof of Proposition~\ref{prop:starting-condition}]
First, note that $\|\bz(0)\|^2 \leqs \|\bx(0)\|^2 = n$. Hence, it suffices to upper bound the probability that $\|\by(0)\|^2$ is less than $1/\beta$. Since $\|\by(0)\| = \left|\frac{\chi^T \bx(0)}{\sqrt{n}}\right|$, the probability that $\|\by(0)\|^2 < 1/\beta$ is exactly equal to the probability that a sum of $n$ i.i.d. Rademacher random variables lie in $[-\sqrt{n/\beta}, \sqrt{n/\beta}]$. The latter probability is exactly equal to
\begin{align*}
    \frac{1}{2^n}\sum_{i = \left(n - \lfloor \sqrt{n/\beta} \rfloor\right) /
    2}^{\left(n + \lfloor \sqrt{n/\beta} \rfloor\right)/2} \binom{n}{i} \leqs
    \left(\sqrt{n/\beta} + 1\right)\frac{\binom{n}{n/2}}{2^n} \leqs
    O(1/\sqrt{\beta} + 1/\sqrt{n}),
\end{align*}
where the second inequality comes from a well-known fact that $\binom{n}{n/2} = O(2^n/\sqrt{n})$.
\end{proof}

 By combining Theorem~\ref{thm:main} and Proposition~\ref{prop:starting-condition}, we immediately get Theorem \ref{thm:main-concen}.  

 \smallskip

\begin{proof}[Proof of Theorem~\ref{thm:main-concen}]
choosing  $\beta = \left(\frac{d}{\varepsilon b}\right)^{2/3}$, we can upper bound $\Pr_{\bx^{(0)}, \cE}\left[\|\bz^{(t)}\|^2 \leqs \sqrt{\varepsilon b / d} \|\by^{(t)}\|^2\right]$ by
\begin{align*}
\Pr_{\cE}\left[\|\bz^{(t)}\|^2 \leqs \sqrt{\varepsilon b / d} \|\by^{(t)}\|^2 \midv \|\bz^{(0)}\|^2 \leqs n \beta \|\by^{(0)}\|^2\right] + \Pr_{\bx^{(0)} \sim \{\pm 1\}^n}\left[\|\bz^{(0)}\|^2 > n \beta \|\by^{(0)}\|^2\right].
\end{align*}
Then, from Theorem~\ref{thm:main}, the first term is at most $O(\varepsilon \beta b / d) + O(\sqrt{\varepsilon b/d}) = O(\sqrt[3]{\varepsilon b / d})$. Moreover, from Proposition~\ref{prop:starting-condition}, the second term is also at most $O(1/\sqrt{\beta} + 1/\sqrt{n}) = O(\sqrt[3]{\varepsilon b / d} + 1/\sqrt{n})$.
 \qed
\end{proof}


\subsection{Proof of Theorem~\ref{thm:main}} \label{app:proofthm:main}
\subsubsection{Evolution of State in One Time Step}

The first step in proving Theorem~\ref{thm:main} is to understand what 
happens in a single update. Specifically, we would like to understand 
how $\|\by^{(t)}\|$ and $\|\bz^{(t)}\|$ behave, given $\|\by^{(t - 
1)}\|$ and $\|\bz^{(t - 1)}\|$. To this end, we prove the following 
lemma, which gives bounds on expectations of $\|\by^{(t)}\|^2$ and 
$\|\bz^{(t)}\|^2$ based on $\|\by^{(t - 1)}\|^2$ and $\|\bz^{(t - 
1)}\|^2$.

\begin{lemma} \label{lem:one-step}
Let $G$ be as above. Let $\by$ be a vector parallel to $\chi$ and $\bz$ 
be a vector orthogonal to $\by$ and to $\bone$. Let $P_\chi$ be the 
projection matrix on $\chi$, that is, $P_\chi = \frac 1n \chi \chi^T$ 
and let $P_\perp$ be the the projection matrix on the space orthogonal 
$\chi$, that is, $P_\perp = I - \frac 1n \chi\chi^T$. Moreover, let 
$\by' = P_\chi W(y + z)$ and $\bz' = P_\perp W (y + z)$ where $W$ is 
randomly selected according to $\cM$. Then,
\begin{align*}
\E_W [ \|\by'\|^2 ] \leqs \left(1 - \frac{8\delta b}{dn} + 
\frac{16\delta^2b}{dn^2}\right)\| \by \|^2 + \left(\frac 
{16\delta^2b}{d n^2}\right)\|\bz\|^2
\end{align*}
and
\begin{align*}
\E_W [ \|\bz'\|^2 ]  \leqs \left(\frac{8\delta^2b}{dn}\right)\|\by\|^2 
+ \left( 1 - \frac{4\delta(1 - \delta)\lambda_3}{n}\right) \|\bz\|^2.
\end{align*}
\end{lemma}

We note here that Lemma~\ref{lem:one-step} is simply a restatement of Lemma~\ref{lem:one-step-1}.

\begin{proof}
Note that $\by'$ and $\bz'$ are orthogonal. We will estimate the 
expectation of $\|\by' + \bz'\|^2$ and  of $\|\by'\|^2$, and the we 
will use Pythagoras's theorem to deduce a bound on $\|\bz'\|^2$.

To estimate the expected norm squared of $\by' + \bz'$ we see that
\begin{align*} 
\E_W [ \| \by' + \bz'\|^2 ]
    & = \E_W [ \| W (\by + \bz) \|^2 ] \\
    & = \E_W [ (\by + \bz)^T W^T W (\by + \bz) ] \\
    & = \E_W [ (\by + \bz)^T W^2 (\by + \bz) ] \\
    &= \E_W [ (\by + \bz)^T ((2\delta - 1)I + 2(1 - \delta)W) (\by + \bz) ] \\
    & = (2\delta - 1)\left(\|\by\|^2 + \|\bz\|^2\right) + 2(1 - \delta) (\by + \bz)^T\bar{W}(\by + \bz) \\
    & = (2\delta - 1)\left(\|\by\|^2 + \|\bz\|^2\right) + 2(1 - \delta)\left(\by^T \bar{W}\by + 2 \bz^T \bar{W}\by + \bz^T \bar{W}\bz\right) \\
    & = (2\delta - 1)\left(\|\by\|^2 + \|\bz\|^2\right) + 2(1 -
    \delta)\left(\left(1 - \frac{4\delta b}{dn}\right)\|\by\|^2 + \bz^T \bar{W}
    \bz\right),
\end{align*}
where $\bar W := \Expec{}{W}$, and in $(a)$ we used that $\bar{W}\by = \left(1 -
\frac{4\delta b}{dn}\right)\by$.

Moreover, note that $\bz$ can be written as a linear combination of 
eigenvectors of $\bar{W}$ whose eigenvalues are at most $1 - 
\frac{2\delta\lambda_3}{n}$; This implies that $\bz^T\bar{W}\bz \leqs 
\left(1 - \frac{2\delta\lambda_3}{n}\right)\|\bz\|^2$. Plugging this 
inequality into the above equality, we have
\begin{align} \label{eq:yplusz}
\E_W [ \| \by' + \bz'\|^2 ] &\leqs \left(1 - \frac{8\delta(1 - 
\delta)b}{dn}\right)\|\by\|^2 + \left(1 - \frac{4\delta(1 - 
\delta)\lambda_3}{n}\right)\|\bz\|^2.
\end{align}

Now we estimate the expected squared norm of $\by'$. Observe that
\begin{align*}
    \E_W [ \| \by'\|^2 ] 
    & = \E_W [ \| P_\chi W (\by + \bz)\|^2 ]\\
    & = \E_W \ [ (\by+ \bz) ^T W^T P_\chi^T P_\chi W (\by + \bz) ]\\
    & = \E_W \ [ (\by+ \bz) ^T W^T  P_\chi W (\by + \bz) ].
\end{align*}

Let us consider two cases: whether the edge $(u,v)$ defining $W$ is an 
{\em internal} edge, that is, an edge whose endpoints are on the same 
side of the partition, or it is a {\em cross} edge having endpoints on 
different sides of the partition. 
\begin{enumerate}
\item
If $(u,v)$ is an internal edge, which happens with probability $ 1- 
\frac{b}{d}$, then $W \chi = \chi$, and so $W P_\chi W = \frac{1}{n} W
\chi \chi^T W = \frac{1}{n} \chi \chi^T = P_\chi$. This implies that 
$(\by+ \bz) ^T W^T  P_\chi W (\by + \bz) = \|\by\|^2.$
\item
If $(u,v)$ is a cross edge such that $u \in V_1$ and $v\in V_2$, recall 
that $W = I - \delta \be_{u, v} \be_{u, v}^T$. To bound $\|P_\chi W
(\by + \bz)\|^2$, first observe that
\begin{align*}
P_\chi W \chi 
= P_\chi(\chi - \delta \be_{u, v}(\be_{u, v}^T \chi))
= \frac{\chi}{n}(\chi^T\chi - \delta \|\be_{u, v}^T\chi\|^2)
= \chi - \frac{4 \delta}{n}\chi.
\end{align*}

Hence, we have
\begin{align} \label{eq:pwy}
P_\chi W \by = \left(1 - \frac{4 \delta}{n}\right) \by.
\end{align}

Now, let us consider $P_\chi W \bz$. Observe that
\begin{align} \label{eq:pwz}
P_\chi(W \bz) = P_\chi(\bz - \delta \be_{u, v}(\be_{u, v}^T \bz)) = 
P_\chi(\bz - \delta (z_u - z_v) \be_{u, v}) = -\frac{2\delta(z_u - 
z_v)}{n} \chi.
\end{align}

By combining (\ref{eq:pwy}) and (\ref{eq:pwz}), we have
\begin{align*}
\|P_\chi W (\by + \bz)\|^2 = \left(1 - \frac{4 \delta}{n}\right)^2 
\|\by\|^2 - 2 \left(\frac{2\delta(z_u - z_v)}{\sqrt{n}}\right)\left(1 - 
\frac{4 \delta}{n}\right) \|\by\| + \frac{4\delta^2(z_u - z_v)^2}{n}.
\end{align*}

Now, if we take the expectation over cross edges $(u,v)$, the second 
term becomes zero, because both $z_u$ and $z_v$ average to zero for a 
random cross edge (the margninal of $u$ is uniform over $V_1$ and the 
marginal of $v$ is uniform over $V_2$). Moreover, we have
\begin{align*}
\E_{(u,v) \rm \ cross edge} [ (z_u - z_v)^2 ] \leqs \E_{(u, v) \rm 
cross edge} [ 2(z_u^2 + z_v^2) ] = \E_{u \in V} [ 4z_u^2 ]  = 
\frac{4}{n} \|\bz\|^2.
\end{align*}

where the first equality follows from the fact that each vertex has $b$ cross edges.

Thus, in this case, we have
\begin{align*}
\left(1 - \frac{4 \delta}{n}\right)^2 \|\by\|^2 \leqs \|P_\chi W (\by + 
\bz)\|^2 \leqs \left(1 - \frac{4 \delta}{n}\right)^2 \|\by\|^2 + 
\frac{16\delta^2}{n^2} \|\bz\|^2.
\end{align*}
\end{enumerate}

Putting the two cases together, we arrive at the following inequality.
\begin{align} \label{eq:yprime}
	 \left(1 - \frac{8\delta b}{dn} + \frac{16\delta^2b}{dn^2}\right) \| \by \|^2 
	 \leqs \E_W [ \|\by'\|^2 ] 
	 \leqs \left(1 - \frac{8\delta b}{dn} + 
	 \frac{16\delta^2b}{dn^2}\right)\| \by \|^2 + \left(\frac 
	 {16\delta^2b}{d n^2}\right)\|\bz\|^2.
\end{align}

Finally, note that the upper bound in (\ref{eq:yprime}) is already the desired
upper bound for $\E_W [ \|\by'\|^2 ]$ and that the lower bound in (\ref{eq:yprime})
together with (\ref{eq:yplusz}) implies the desired upper bound on $\E_W [ \|\bz'\|^2 ]$.
\end{proof}

\subsubsection{From Evolution of State to Bounds on $\E_\cE [ \|\by^{(t)}\|^2 ]$ and $\E_\cE [ \|\bz^{(t)}\|^2 ]$}

We next turn the bounds from Lemma~\ref{lem:one-step} to bounds on $\E_\cE
[ \|\by^{(t)}\|^2 ]$ and $\E_\cE [ \|\bz^{(t)}\|^2 ]$ based only on $\|\by(0)\|^2,
\|\bz(0)\|^2, \delta$ and the parameters of our graph. This bound will indeed
be enough for us to prove certain concentrations of $\|\by^{(t)}\|$ and
$\|\bz^{(t)}\|$, which are at the heart of the analysis. Before we state the
bounds on $\E_\cE [ \|\by^{(t)}\|^2 ]$ and $\E_\cE [ \|\bz^{(t)}\|^2 ]$, let us define
the following shorthands for some expressions that will appear regularly
throughout the rest of the section.
\begin{itemize}
\item Let $\xi \triangleq \left(1 - \frac{4\delta b}{dn}\right)^2, \xi_1 \triangleq \left(1 - \frac{8\delta b}{dn} + \frac{336\delta^2b}{dn^2}\right), \xi_2 \triangleq \left(1 - \frac{4\delta(1 - \delta)\lambda_3}{n}\right)$. Note that $\mu(t) = \xi^{t/2}a_y(0)$.
\item Let $\kappa \triangleq 1 + (40\varepsilon b/d)\beta$. Recall that $\beta$ is a parameter in Theorem~\ref{thm:main} which satisfies $\beta \geqs \frac{\|\bz(0)\|^2}{n\|\by(0)\|^2}$.
\end{itemize}
We can now state our bounds on $\E_\cE [ \|\by^{(t)}\|^2 ]$ and $\E_\cE [ \|\bz^{(t)}\|^2 ]$:

\begin{lemma} \label{lem:solved-formula-yz}
For any  $t \in \Z_{\geqs 0}$, we have
\begin{align*}
\E_{\cE} [ \|\by^{(t)}\|^2 ] \leqs (\kappa \xi_1^t) \|\by(0)\|^2
& &\text{ and }
& &\E_{\cE} [ \|\bz^{(t)}\|^2 ] \leqs \left((20\varepsilon b/d) \kappa \xi_1^t + \beta n \xi_2^t\right) \|\by(0)\|^2.
\end{align*}
\end{lemma}

We defer the proof of Lemma~\ref{lem:solved-formula-yz}, which is essentially solving the recurrence relation from Lemma~\ref{lem:one-step}, to Subsection~\ref{app:recurrence}. Let us now proceed to use this lemma to derive concentrations of $\|\by^{(t)}\|, \|\bz^{(t)}\|$.

\subsubsection{Concentrations of $\|\by^{(t)}\|$ and $\|\bz^{(t)}\|$}

A direction application of Markov's inequality to the bound on $\E_{\cE} [
\|\bz^{(t)}\|^2 ]$ from Lemma~\ref{lem:solved-formula-yz} gives us the desired
concentration for $\|\bz^{(t)}\|$:

\begin{lemma} \label{lem:concen-z}
For every $t \in \Z_{\geqs 0}$,
$
\Pr_{\cE}\left[\|\bz^{(t)}\|^2 \geqs 0.25\sqrt{\varepsilon b/d} (\mu(t))^2\right] \leqs 80\sqrt{\varepsilon b/d}\kappa(\xi_1 / \xi)^t + \frac{4\beta n}{\sqrt{\varepsilon b/d}} (\xi_2 / \xi)^t.
$
\end{lemma}

For $\|\by^{(t)}\|$, since we know that ${\|\by^{(t)}\|}^2$ is simply
$a_y(t)^2$ and we also know $\mu(t) = \E_{\cE} [ a_y(t) ]$, we can apply
Cherbychev's inequality on $a_y(t)$, which results in the following lemma.

\begin{lemma} \label{lem:concen-y}
For every $t \in \Z_{\geqs 0}$, 
$
\Pr_{\cE}[a_y(t) \notin (0.5 \mu(t), 1.5 \mu(t))] \leqs 4\left(\kappa(\xi_1/\xi)^t - 1\right).
$
\end{lemma}

\begin{proof}
Recall from Observation~\ref{obs:onemoment} that $\E_{\cE} [ a_y(t) ] = \mu(t) = \xi^{t/2} a_y(0)$. Moreover, from
 Lemma~\ref{lem:solved-formula-yz}, we have $\E_{\cE} [ a_y(t)^2 ] = \E_{\cE} [ \|\by^{(t)}\|^2 ] \leqs (\kappa \xi_1^t)(a_{y^{(0)}})^2$.
  Hence, from Chebyshev's inequality, we have
\begin{align*}
\Pr_{\cE}[a_y(t) \notin (0.5 \mu(t), 1.5 \mu(t))] 
\leqs \frac{\E_{\cE} [ a_y(t)^2 ] - \left(\mu(t)\right)^2}{(0.5\mu(t))^2} 
= 4\left(\kappa(\xi_1/\xi)^t - 1\right)
\end{align*}
as desired.
\end{proof}

\subsubsection{Putting things together}

Finally, we will now prove Theorem~\ref{thm:main} by plugging in the appropriate value for variables in Lemma~\ref{lem:concen-z} and Lemma~\ref{lem:concen-y}. To this end, let us first state a couple of inequalities that will be useful.

\begin{lemma} \label{lem:aux1}
If $\beta \leqs \frac{d}{\varepsilon b}$, then, for any $t \leqs \frac{n^2 \beta}{1344 \delta(\lambda_3 - \lambda_2)}$, we have $\kappa (\xi_1/\xi)^t \leqs 1 + 81(\varepsilon b \beta / d)$.
\end{lemma}

\begin{lemma} \label{lem:aux2}
If $\beta \leqs \frac{d}{\varepsilon b}$, then, for any $t \geqs \frac{8n}{\delta(\lambda_3 - \lambda_2)} \log\left(\frac{nd\beta}{\varepsilon b}\right)$, we have $\frac{4\beta n}{\sqrt{\varepsilon b/d}} (\xi_2 / \xi)^t \leqs 4\sqrt{\varepsilon b / d}$.
\end{lemma}

We defer the proofs of both lemmas, which are basically calculations, to Appendix~\ref{app:proofaux}. Let us now proceed to prove Theorem~\ref{thm:main}.

\begin{proof}[Proof of Theorem~\ref{thm:main}]
Let $t$ be any positive integer such that $\frac{8n}{\delta(\lambda_3 - \lambda_2)} \log\left(\frac{nd}{\varepsilon b}\right) \leqs t \leqs \frac{n^2 \beta}{1344 \delta(\lambda_3 - \lambda_2)}$. From Lemma~\ref{lem:concen-y} and Lemma~\ref{lem:aux1}, we have $\Pr_{\cE}[a_y(t) \notin (0.5 \mu(t), 1.5 \mu(t))] \leqs O(\varepsilon b \beta / d)$.

Moreover, from Lemma~\ref{lem:concen-z}, Lemma~\ref{lem:aux1} and Lemma~\ref{lem:aux2}, we have
\begin{align*}
\Pr_{\cE}\left[\|\bz^{(t)}\| \geqs  0.5 \sqrt[4]{\varepsilon b/d} \mu(t)\right]
= \Pr_{\cE}\left[\|\bz^{(t)}\|^2 \geqs  0.25 \sqrt{\varepsilon b/d} (\mu(t))^2\right]
&\leqs 80\sqrt{\varepsilon b/d}\kappa(\xi_1 / \xi)^t + \frac{4\beta n}{\sqrt{\varepsilon b/d}} (\xi_2 / \xi)^t \\
\text{(From Lemma~\ref{lem:aux1} and Lemma~\ref{lem:aux2})} &\leqs O(\sqrt{\varepsilon b / d}),
\end{align*}
which concludes the proof of Theorem~\ref{thm:main}.

\end{proof}

\subsubsection{Proof of Lemma~\ref{lem:solved-formula-yz}}\label{app:recurrence}

The goal of this subsection is to prove 
Lemma~\ref{lem:solved-formula-yz}, which is essentially just solving 
the recurrence relation from Lemma~\ref{lem:one-step}. Before we 
proceed to the proof, we state a fact and an observation regarding 
eigenvalues and eigenvectors of certain $2 \times 2$ natruces, which will be 
useful in our proof.

\begin{fact} \label{fact:2x2eigenformula}
Let $A = \begin{bmatrix} a_{11} & a_{12} \\ a_{21} & a_{22} 
\end{bmatrix} \in \mathbb{R}^{2 \times 2}$ be a $2 \times 2$ 
real-valued matrix such that $a_{11} \ne 0$ and $(a_{11} - a_{22})^2 
\ne 4 a_{12}a_{21}$. Then, its eigenvalues are
\begin{align*}
\alpha_1(A) &\triangleq \frac{1}{2} \left(a_{11} + a_{22} + \sqrt{(a_{11} - a_{22})^2 + 4 a_{12}a_{21}}\right) \text{ and } \\
\alpha_2(A) &\triangleq \frac{1}{2} \left(a_{11} + a_{22} - \sqrt{(a_{11} - a_{22})^2 + 4 a_{12}a_{21}}\right),
\end{align*}
and its eigenvectors are
\begin{align*}
\begin{bmatrix}
1 \\ \frac{\alpha_1(A) - a_{11}}{a_{12}}
\end{bmatrix}
\text{ and }
\begin{bmatrix}
1 \\ \frac{\alpha_2(A) - a_{11}}{a_{12}}
\end{bmatrix}.
\end{align*}
\end{fact}

\begin{obs} \label{obs:2x2eigenvaluebound}
Let $A, a_{11}, a_{12}, a_{21}, a_{22}, \alpha_1(A), \alpha_2(A)$ be as in Fact~\ref{fact:2x2eigenformula}. Suppose further that $a_{11}, a_{12}, a_{21}, a_{22} \geqs 0, a_{11} > a_{22}$ and that $(a_{11} - a_{22})^2 > 4a_{12}a_{21}$. Then, $$a_{11} + \frac{a_{12}a_{21}}{a_{11} - a_{22}}  \geqs \alpha_1(A) \geqs a_{11} \geqs a_{22} \geqs \alpha_2(A) \geqs a_{22} - \frac{a_{12}a_{21}}{a_{11} - a_{22}}.$$
\end{obs}

\begin{proof}
The inequalities come from an observation that $a_{11} - a_{22} \leqs \sqrt{(a_{11} - a_{22})^2 + 4 a_{12}a_{21}} \leqs a_{11} - a_{22} + \frac{2 a_{12}a_{21}}{a_{11} - a_{22}}$.
\end{proof}

We are now ready to prove Lemma~\ref{lem:solved-formula-yz}.

\begin{proof}[Proof of Lemma~\ref{lem:solved-formula-yz}]
    Let $\{y^{(t)}\}_{t \in \Z_{\geqs 0}}$ and $\{z^{(t)}\}_{t \in \Z_{\geqs 0}}$ be sequence of non-negative real numbers defined by $y^{(0)} = \|\by(0)\|^2, z^{(0)} = \|\bz(0)\|^2$ and
    \begin{align*}
    \begin{bmatrix} y^{(t)} \\ z^{(t)} \end{bmatrix}
    =
    A
    \begin{bmatrix} y^{(t - 1)} \\ z^{(t - 1)} \end{bmatrix}
    \text{ where } A \triangleq \begin{bmatrix}
    1 - \frac{8\delta b}{dn} + \frac{16\delta^2b}{dn^2} & \frac{16\delta^2b}{d n^2} \\
    \frac{8\delta^2b}{dn} & 1 - \frac{4\delta(1 - \delta)\lambda_3}{n}
    \end{bmatrix}
    \end{align*}
    for every $t \in \N$. Note that, from Lemma~\ref{lem:one-step} and from the
    initial values $y^{(0)}, z^{(0)}$, we have $y^{(t)} \geqs \E_\cE [
    \|\by(t)\|^2 ]$ and $z^{(t)} \geqs \E_\cE [ \|\bz(t)\|^2 ]$ for every $t
    \in \Z_{\geqs 0}$. Hence, to prove the lemma, it suffices to prove that
    \begin{align*}
    y^{(t)} &\leqs \left(y^{(0)} + \left(\frac{40\varepsilon b}{dn}\right)z^{(0)}\right)\left(1 - \frac{8\delta b}{dn} + \frac{336\delta^2b}{dn^2}\right)^t
    \end{align*}
    and
    \begin{align*}
    z^{(t)} &\leqs \left((y^{(0)} + \left(\frac{40\varepsilon b}{dn}\right)z^{(0)}\right)\left(1 - \frac{8\delta b}{dn} + \frac{336\delta^2b}{dn^2}\right)^t + z^{(0)}\left(1 - \frac{4\delta(1 - \delta)\lambda_3}{n}\right)^t.
    \end{align*}

    Let $a_{11}, a_{12}, a_{21}, a_{22}$ be the entries of $A$. Note that
    \begin{align*} 
    a_{11} - a_{12} &= \left(1 - \frac{8\delta b}{dn} + \frac{16\delta^2b}{dn^2}\right) - \left(1 - \frac{4\delta(1 - \delta)\lambda_3}{n}\right) \\
    &\geqs \frac{4\delta(1 - \delta)\lambda_3}{n} - \frac{8\delta b}{dn} \\
    &= \frac{4\delta(1 - \delta)\lambda_3}{n} - \frac{4\delta\lambda_2}{n} \\
    &= \frac{4\delta}{n}\left(\lambda_3 - \lambda_2 - \delta\lambda_3\right).
    \end{align*}

    Observe here that, since the sum of eigenvalues of $L$ is equal to $\tr(L) = n$ and $\lambda_1, \lambda_2 \geqs 0$, we have $\lambda_3 \leqs \frac{n}{n - 2}$. Thus, from $\delta \leqs 0.8(\lambda_3 - \lambda_2)$, we have that $\delta\lambda_3 \leqs 0.9(\lambda_3 - \lambda_2)$ for any sufficiently large $n$. Plugging this into the above inequality gives
    \begin{align} \label{eq:a11minusa22}
    a_{11} - a_{12} \geqs \frac{0.4\delta(\lambda_3 - \lambda_2)}{n} = \frac{0.4\delta^2}{\varepsilon n}.
    \end{align}

    The above inequality implies that $a_{11} - a_{22} > 0$. Moreover, for $n \geqs 1000$, we have
    \begin{align*}
    (a_{11} - a_{22})^2 > \frac{0.16\delta^4}{\varepsilon^2 n^2} > \frac{128\delta^4b^2}{d^2n^3} = a_{12}a_{21}.
    \end{align*}

    In other words, the conditions in Fact~\ref{fact:2x2eigenformula} and Observation~\ref{obs:2x2eigenvaluebound} are satisfied. From Fact~\ref{fact:2x2eigenformula}, the eigenvalues of $A$ are
    \begin{align*}
    \alpha_1 \triangleq \frac{1}{2} \left(a_{11} + a_{22} + \sqrt{(a_{11} - a_{22})^2 + 4 a_{12}a_{21}}\right)
    \end{align*}
    and
    \begin{align*}
    \alpha_2 \triangleq \frac{1}{2} \left(a_{11} + a_{22} - \sqrt{(a_{11} - a_{22})^2 - 4 a_{12}a_{21}}\right).
    \end{align*}
    Furthermore, the eigenvectors of $A$ are
    \begin{align*}
    \bv_1 \triangleq 
    \begin{bmatrix}
    1 \\ \frac{\alpha_1 - a_{11}}{a_{12}}
    \end{bmatrix}
    \text{ and }
    \bv_2 \triangleq
    \begin{bmatrix}
    1 \\ \frac{\alpha_2 - a_{11}}{a_{12}}
    \end{bmatrix}.
    \end{align*}
    Let $\gamma_1 \triangleq \left(\frac{a_{11} - \alpha_2}{\alpha_1 - \alpha_2}\right)y^{(0)} + \left(\frac{a_{12}}{\alpha_1 - \alpha_2}\right)z^{(0)}$ and $\gamma_2 \triangleq \left(\frac{\alpha_1 - a_{11}}{\alpha_1 - \alpha_2}\right)y^{(0)} - \left(\frac{a_{12}}{\alpha_1 - \alpha_2}\right)z^{(0)}$. It is easy to see that
    \begin{align*}
    \begin{bmatrix} y^{(0)} \\ z^{(0)} \end{bmatrix} = \gamma_1 \bv_1 + \gamma_2 \bv_2.
    \end{align*}
    Since $\bv_1$ and $\bv_2$ are eigenvectors of $A$ with eigenvalues $\alpha_1$ and $\alpha_2$ respectively, we have
    \begin{align*}
    \begin{bmatrix} y^{(t)} \\ z^{(t)} \end{bmatrix} = \gamma_1 \alpha_1^t \bv_1 + \gamma_2 \alpha_2^t \bv_2
    \end{align*}
    for every $t \in \Z$. In other words, we have
    \begin{align*}
    y^{(t)} = \gamma_1 \alpha_1^t + \gamma_2 \alpha_2^t, 
    & &z^{(t)} = \left(\frac{\alpha_1 - a_{11}}{a_{12}}\right) \gamma_1 \alpha_1^t + \left(\frac{\alpha_2 - a_{11}}{a_{12}}\right) \gamma_2 \alpha_2^t.
    \end{align*}
    Having derived the above formula, we will now bound $y^{(t)}, z^{(t)}$ by appropriately bounding the eigenvalues and coefficients. Before we do so, let us list a few inequalities that will be useful.
    \begin{itemize}
    \item From (\ref{eq:a11minusa22}), we have the following three ineqalities.
    \begin{align} \label{eq:a12divdiff}
    \frac{a_{12}}{a_{11} - a_{22}} \leqs \frac{40\varepsilon b}{dn},
    \end{align}
    \begin{align} \label{eq:a21divdiff}
    \frac{a_{21}}{a_{11} - a_{22}} \leqs \frac{20\varepsilon b}{d},
    \end{align}
    and
    \begin{align} \label{eq:proddivdiff}
    \frac{a_{12}a_{21}}{a_{11} - a_{22}} \leqs \frac{320\varepsilon\delta^2b^2}{d^2n^2} \leqs \frac{320\delta^2b}{dn^2}
    \end{align}
    where the second inequality comes from $\varepsilon < 1$ and $b \leqs d$.
    \item From (\ref{eq:proddivdiff}) and from Observation~\ref{obs:2x2eigenvaluebound}, we have
    \begin{align} \label{eq:eigenvaluesinq}
    a_{11} + \frac{320\delta^2b}{dn^2} \geqs \alpha_1 \geqs a_{11} \geqs a_{22} \geqs \alpha_2 \geqs a_{22} - \frac{320\delta^2b}{dn^2}.
    \end{align}
    Note also that the right-most term above is non-negative for sufficiently large $n$.
    \end{itemize}

    \paragraph{Bounding $y^{(t)}$.} With the above inequalities in place, it is now easy to bound $y^{(t)}$ as follows.
    \begin{align*}
    y^{(t)} &= \gamma_1 \alpha_1^t + \gamma_2 \alpha_2^t \\
    (\text{Since } \gamma_2 \leqs \left(\frac{\alpha_1 - a_{11}}{\alpha_1 - \alpha_2}\right)y^{(0)} \text{ and } \alpha_2 \geqs 0) &\leqs \gamma_1 \alpha_1^t + \left(\frac{\alpha_1 - a_{11}}{\alpha_1 - \alpha_2}\right)y^{(0)} \alpha_2^t \\
    (\text{Since } \alpha_2 \leqs \alpha_1) &\leqs \gamma_1 \alpha_1^t + \left(\frac{\alpha_1 - a_{11}}{\alpha_1 - \alpha_2}\right)y^{(0)} \alpha_1^t \\
    &= \left(y^{(0)} + \left(\frac{a_{12}}{\alpha_1 - \alpha_2}\right)z^{(0)}\right)\alpha_1^t \\
    (\text{From }(\ref{eq:eigenvaluesinq})) &\leqs \left(y^{(0)} + \left(\frac{a_{12}}{a_{11} - a_{22}}\right)z^{(0)}\right)\alpha_1^t \\
    (\text{From }(\ref{eq:a12divdiff})) &\leqs \left(y^{(0)} + \left(\frac{8\varepsilon b}{dn}\right)z^{(0)}\right)\alpha_1^t \\
    (\text{From }(\ref{eq:eigenvaluesinq})) &\leqs \left(y^{(0)} + \left(\frac{40\varepsilon b}{dn}\right)z^{(0)}\right)\left(1 - \frac{8\delta b}{dn} + \frac{336\delta^2b}{dn^2}\right)^t
    \end{align*}
    as desired.

    \paragraph{Bounding $z^{(t)}$.} Recall that $z^{(t)} = \left(\frac{\alpha_1 - a_{11}}{a_{12}}\right) \gamma_1 \alpha_1^t + \left(\frac{\alpha_2 - a_{11}}{a_{12}}\right) \gamma_2 \alpha_2^t$. Let us bound the two terms separately, starting with the first term $\left(\frac{\alpha_1 - a_{11}}{a_{12}}\right) \gamma_1 \alpha_1^t$. To this end, we can bound $\left(\frac{\alpha_1 - a_{11}}{a_{12}}\right) \gamma_1$ as follows.
    \begin{align*}
    \left(\frac{\alpha_1 - a_{11}}{a_{12}}\right) \gamma_1 &\leqs \left(\frac{a_{21}}{a_{11} - a_{22}}\right) \gamma_1 \\
    (\text{From } (\ref{eq:a21divdiff})) &\leqs \left(\frac{20\varepsilon b}{d}\right) \gamma_1 \\
    &= \left(\frac{20\varepsilon b}{d}\right)\left(\left(\frac{a_{11} - \alpha_2}{\alpha_1 - \alpha_2}\right)y^{(0)} + \left(\frac{a_{12}}{\alpha_1 - \alpha_2}\right)z^{(0)}\right) \\
    (\text{From }(\ref{eq:eigenvaluesinq})) &\leqs \left(\frac{20\varepsilon b}{d}\right)\left(y^{(0)} + \left(\frac{a_{12}}{a_{11} - a_{22}}\right)z^{(0)}\right) \\
    (\text{From } (\ref{eq:a12divdiff})) &\leqs \left(\frac{20\varepsilon b}{d}\right)\left(y^{(0)} + \left(\frac{40\varepsilon b}{dn}\right)z^{(0)}\right).
    \end{align*}
    Note that the first inequality comes from Observation~\ref{obs:2x2eigenvaluebound} and from $\gamma_1 \geqs 0$. From the above bound on $\left(\frac{\alpha_1 - a_{11}}{a_{12}}\right) \gamma_1$ and our bound on $\alpha_1$ from (\ref{eq:eigenvaluesinq}), we have
    \begin{align} \label{eq:ztbound-firstterm}
    \left(\frac{\alpha_1 - a_{11}}{a_{12}}\right) \gamma_1 \alpha_1^t \leqs \left(\frac{20\varepsilon b}{d}\right)\left(y^{(0)} + \left(\frac{40\varepsilon b}{dn}\right)z^{(0)}\right)\left(1 - \frac{8\delta b}{dn} + \frac{336\delta^2b}{dn^2}\right)^t.
    \end{align}

    Let us next bound $\left(\frac{\alpha_2 - a_{11}}{a_{12}}\right) \gamma_2 \alpha_2^t$. Again, we first rearrange the coefficient $\left(\frac{\alpha_2 - a_{11}}{a_{12}}\right) \gamma_2$ as
    \begin{align*}
    \left(\frac{\alpha_2 - a_{11}}{a_{12}}\right) \gamma_2 &= \left(\frac{\alpha_2 - a_{11}}{a_{12}}\right) \left(\left(\frac{\alpha_1 - a_{11}}{\alpha_1 - \alpha_2}\right)y^{(0)} - \left(\frac{a_{12}}{\alpha_1 - \alpha_2}\right)z^{(0)}\right) \\
    (\text{Since } \frac{\alpha_2 - a_{11}}{a_{12}} \leqs 0 \text{ and } \frac{\alpha_1 - a_{11}}{\alpha_1 - \alpha_2} \geqs 0) &\leqs \left(\frac{a_{11} - \alpha_2}{\alpha_1 - \alpha_2}\right) z^{(0)} \\
    (\text{From }(\ref{eq:eigenvaluesinq})) &\leqs z^{(0)}.
    \end{align*}
    Hence, from the above inequality and $(\ref{eq:eigenvaluesinq})$, we have
    \begin{align} \label{eq:ztbound-secondterm}
    \left(\frac{\alpha_2 - a_{11}}{a_{12}}\right) \gamma_2 \alpha_2^t \leqs z^{(0)} a_{22}^t = z^{(0)}\left(1 - \frac{4\delta(1 - \delta)\lambda_3}{n}\right)^t.
    \end{align}
    Combining (\ref{eq:ztbound-firstterm}) and (\ref{eq:ztbound-secondterm}) indeed yields the desired bound on $z^{(t)}$.
\end{proof}

\subsubsection{Proofs of Lemma~\ref{lem:aux1} and Lemma~\ref{lem:aux2}} \label{app:proofaux}

\begin{proof}[Proof of Lemma~\ref{lem:aux1}]
Since $\kappa = 1 + 40 \varepsilon b \beta / d$ and since $\varepsilon b \beta / d \leqs 1$, it suffices to show that $(\xi_1/\xi)^t \leqs 1 + \varepsilon b \beta / d$. To show this, let us rearrange $(\xi/\xi_1)^t$ as follows.
\begin{align*}
    \left(\frac{\xi}{\xi_1}\right)^t
    &= \left(1 - \frac{\xi_1 - \xi}{\xi_1}\right)^t \\
    &\geqs \left(1 - \frac{\frac{336\delta^2 b}{dn^2}}{\xi_1}\right)^t \\
    (\text{Since } \xi \geqs 1/2 \text{ when } n \geqs 8) &\geqs \left(1 - \frac{672\delta^2 b}{dn^2}\right)^t \\
    (\text{From Bernoulli's inequality}) &\geqs 1 - \frac{672\delta^2 b t}{dn^2} \\
    (\text{Since } t \leqs \frac{n^2 \beta}{1344 \delta(\lambda_3 - \lambda_2)}) &\geqs 1 - \frac{\varepsilon b \beta}{2d}.
    \end{align*}
    Note that we can apply Bernoulli's inequality since $\frac{672\delta^2
    b}{dn^2} \leqs 1$ for any sufficiently large $n$ (i.e. $n \geqs 30$).
    Finally, note that the above inequality implies that $(\xi_1/\xi)^t \leqs 1
    + \varepsilon b \beta / d$ since $\frac{1}{1 - \frac{\varepsilon b
    \beta}{2d}} \leqs 1 + \varepsilon b \beta/d$ because $\varepsilon b \beta/d
    \leqs 1$.
\end{proof}

\begin{proof}[Proof of Lemma~\ref{lem:aux2}]
    Observe that, in (\ref{eq:a11minusa22}), we have already proved that $\xi - \xi_2 \geqs \frac{2\delta^2}{\varepsilon n} = \frac{2\delta(\lambda_3 - \lambda_2)}{n}$. Moreover, observe that $\xi_2 \leqs 1$. Hence, we have
    \begin{align*}
    \left(\frac{\xi}{\xi_2}\right)^t 
    &= \left(1 + \frac{\xi - \xi_2}{\xi_2}\right)^t \\
    &\geqs \left(1 + \frac{2\delta(\lambda_3 - \lambda_2)}{n}\right)^t. \\
    (\text{From Bernoulli's inequality}) &\geqs 2^{\frac{2\delta(\lambda_3 - \lambda_2)t}{n}} \\
    &\geqs 2^{4\log\left(\frac{nd\beta}{\varepsilon b}\right)} \\
    &= \left(\frac{nd\beta}{\varepsilon b}\right)^4,
    \end{align*}
    which implies the inequality stated in the lemma.
\end{proof}

\section{Omitted Proofs from Section~\ref{sec:reconstruct-analysis}}
\label{app:reconstruct}

The main goal of this section is to prove Theorem~\ref{thm:main-reconstruct}. The actual proof deviates in a couple of subtle ways from the outline in Section~\ref{sec:reconstruct-overview}:
\begin{itemize}
\item Firstly, we use a slightly different notion of ``good at time $t$''. In the outline, we say that a
node is good at time $t$ if $\chi_u(\bx_u^{(t)} - a_{||}) \approx \mu(t) / n$.
However, since $\chi_u \cdot \bx_u^{(T_u(\tau^{\stored}_u)) } >
\chi_u \cdot \bx_u^{(T_u(\tau^{\tend}_u))}$ suffices to conclude that
$\bh^{jump}_u = \chi_u$, it is enough for us to pick $\eta \in \R$ as a
cutoff threshold and says that $u$ such that $\chi_u(\bx_u^{(T_u(t))} - a_{||})
\geqs \eta$ is \emph{good for stored time $t$} and $u$ such that
$\chi_u(\bx_u^{(T_u(t))} - a_{||}) < \eta$ is \emph{good for end time $t$}.
More formally, for each $t \in \N$, let $R^\eta_t \triangleq \{u \in V \mid
 \chi_u(\bx_u^{(t)} - a_{||}) \geqs \eta\}$ be the set of good nodes for stored time $t$ and $\oR^\eta_t \triangleq V \setminus R^\eta_t$ be the set of good nodes for end time $t$.
\item Secondly, instead of arguing that $[T_u(\tau^{\stored}), T_u(\ttau^{\stored})] \cap [0.5 n \tau^{\stored}, 0.5 n \ttau^{\stored}]$ is large for most $u$ (and similarly for the ending time), we will argue that $[T_u(\tau^{\stored}), T_u(\ttau^{\stored})] \subseteq [0.4 n \tau^{\stored}, 0.6 n \ttau^{\stored}]$ for most $u$, which suffices for our purpose.
\end{itemize}

More precisely, the main steps of the proof are as follows. After selecting appropriate values of $\tau^{\stored}, \ttau^{\stored}, \tau^{\tend}, \ttau^{\tend}, \eta$, our proof consists of three main steps as follows. For brevity, let us focus on the stored time here as the statements for the end time are analogous.
\begin{enumerate}
\item We start by using the concentration result from the previous section to
    argue that, for each $t \in [0.4 n \tau^{\stored}, 0.6 n \ttau^{\stored}]$,
    most nodes are good for stored time $t$, i.e., $\E_{\cE} |R^{\eta}_t|$ is large. In other words, we will show that $\E_{\cE} |\oR^{\eta}_t|$ is small for such $t$'s.
\item We next argue that, for most nodes $u$, $T_u(\tau^{\stored}), \dots, T_u(\ttau^{\stored})$ are ``sufficiently uniform'' in the following sense: $[T_u(\tau^{\stored}), T_u(\ttau^{\stored})] \subseteq [0.4 n \tau^{\stored}, 0.6 n \ttau^{\stored}]$ and, for most $\tau \in [\tau^{\stored}, \ttau^{\stored}]$, $T_u(\tau + 1) - T_u(\tau)$ is not too much smaller than its expected value, $n/2$.
\item Finally, we show that, if most nodes are uniform, then using local time
    is not much worse than using global time. In other words, we show that, if
    the average size of the sets of bad nodes $\oR^{\eta}_t$ is small over $t \in
    [0.4 n \tau^{\stored}, 0.6 n \ttau^{\stored}]$, then the average size of
    $\oR^{\eta}_{T_u(\tau)}$ is also small over all $\tau \in [\tau^{\stored},
    \ttau^{\stored}]$. The latter indeed implies that most $u$ is unlikely to
    be bad for random stored time $\tau^{\stored}_u \in [\tau^{\stored},
    \ttau^{\stored}]$.
\end{enumerate}

The values of the parameters that we will be using throughout this section are as follows.
\begin{itemize}
    \item Pick $\tau^{\stored} \triangleq
        \frac{100\log\left(\frac{nd}{\varepsilon b}\right)}{\delta(\lambda_3 -
        \lambda_2)}$, $\ttau^{\stored} \triangleq 2 \tau^{\stored}$,
        $\tau^{\tend} \triangleq 3\ttau^{\stored} + \frac{10d}{\delta b}$ and
        $\ttau^{\tend} \triangleq 2 \tau^{\tend}$.
    \item Let $\eta \triangleq 0.25 \mu(0.6 n\ttau^{\stored}) / n$. 
\end{itemize}


\subsection{The proof}

Let us now proceed to the proof. The two main lemmas of the first step can be stated as follows. Since both lemmas follow easily from our concentration result, we defer their proofs to Appendix~\ref{subsec:bound-mu-correlated}.

\begin{lemma} \label{lem:tstored}
For any $\bx^{(0)}$ such that $\|\bz^{(0)}\|^2 \leqs n \sqrt{d/(\varepsilon b)} \|\by^{(0)}\|^2$ and every $t \in [0.4n\tau^{\stored}, 0.6n\ttau^{\stored}]$, we have $\E_{\cE} |\oR^{\eta}_t| \leqs O\left(n \sqrt{\varepsilon b / d}\right).$
\end{lemma}

\begin{lemma} \label{lem:tend}
For any $\bx^{(0)}$ such that $\|\bz^{(0)}\|^2 \leqs n \sqrt{d/(\varepsilon b)} \|\by^{(0)}\|^2$ and every $t \in [0.4n\tau^{\tend}, 0.6n\ttau^{\tend}]$, we have $\E_{\cE} |R^{\eta}_t| \leqs O\left(n \sqrt{\varepsilon b / d}\right).$
\end{lemma}

As stated in the proof overview, the second step is to argue that a random
sequence of edges $\cE$ is ``sufficiently uniform'' for most node $u$ with high probability. The notion of uniformity needed here is formalized below. Note that the parameters $a, b$ below will later be set to $\frac{\tau^{\stored}}{\log n}, \frac{\ttau^{\stored}}{\log n}$ to achieve uniformity for stored time and $\frac{\tau^{\tend}}{\log n}, \frac{\ttau^{\tend}}{\log n}$ to achieve uniformity for end time.

\begin{definition}
Let $a, b, \zeta$ be any positive real number such that $b \geqs 2a$. We say
that a node $u \in V$ is \emph{$(a, b, \zeta)$-uniform} (with respect to a sequence of edges $\cE$) if
\begin{itemize}
\item $T_u(a \log n) > 0.4an\log n$ and $T_u(b \log n + 1) \leqs 0.6bn\log n$, and, 
\item $\Pr_{\tau \in [a \log n, b \log n]} \left[T_u(\tau + 1) < T_u(\tau) + \sqrt{\zeta} n\right] \leqs 4\sqrt{\zeta}$. 
\end{itemize}
\end{definition}

We can argue that, for a random $\cE$, with high probability, most nodes are uniform, as stated below. Since this follows from standard Chernoff bound, we defer the proof to Appendix~\ref{app:standard}.

\begin{lemma} \label{lem:manystandardvertices}
With probability $1 - n^{-\Omega(\sqrt{\zeta}a)}$, at least $n - n^{1 -
    \Omega(\sqrt{\zeta}a)}$ nodes are $(a, b, \zeta)$-uniform.
\end{lemma}

To state the main lemma of the final step of the proof, let us defined an additional notation: we call a sequence $\{S_t\}_{t \in \N}$ of subsets $S_t \subseteq V$ \emph{compatible with $\cE$} if, for every $t \in \N$, $S_t \triangle S_{t + 1} \subseteq \{u_{t + 1}, v_{t + 1}\}$, i.e., $S_{t + 1}$ can only differ from $S_t$ on the endpoints of the edge in step $t$. Observe that the sequences $\{R^\eta_t\}_{t \in N}$and $\{\oR^\eta_t\}_{t \in N}$ are compatible with $\cE$ since $\bx_u^{(t + 1)}$ can change only when $u \in \{u_{t + 1}, v_{t + 1}\}$. The main lemma of this part is stated below.

\begin{lemma} \label{lem:globaltolocal}
For any sequence of edges $\cE$ such that at least $(1 - \sqrt{\zeta})n$ nodes are $(a, b, \zeta)$-uniform and for any sequence of subsets $\{S_t\}_{t \in \N}$ that is compatible with $\cE$ and that $\E_{t \in [0.4an\log n, 0.6b n \log n]} \left[|S_t|\right] \leqs \zeta n$, we have
\begin{align*}
\Pr_{u \in V, \tau \in [a\log n, b\log n]}\left[u \in S_{T_u(\tau)}\right] \leqs O(\sqrt{\zeta}).
\end{align*}
\end{lemma}

$S_t$ should be thought of as the set of bad nodes for $t$; for stored time, we
should think of $S_t$ as $\oR^{\eta}_t$ whereas, for end time, we should think
of $S_t$ as $R^{\eta}_t$. The above lemma asserts that, by randomly choosing a
local time from $[a\log n, b \log n]$, most nodesill likely end up in a global time step where it is good. The proof of Lemma~\ref{lem:globaltolocal} is deferred to Subsection~\ref{subsec:global-to-local}.

Let us now show how to use these lemmas to prove Theorem~\ref{thm:main-reconstruct}.

\begin{proof}[Proof of Theorem~\ref{thm:main-reconstruct}]
    First, note that the fact that every node is labeled at time $O(\frac{n \log
    n}{\delta(\lambda_3 - \lambda_2)} + \frac{nd}{\delta b})$ w.h.p. follows easily
    from applying a union bound on top of a Chernoff bound on $\Pr_{\cE}
    [T_u(\ttau^{\tend}) \leqs 100 n \ttau^{\tend}]$ for each node $u \in V$;
    this latter probability is simply the same as the probability that sum of
    $100 n\ttau^{\tend}$ i.i.d. Bernoulli random variables each with mean $2/n$
    is less than $\ttau^{\tend}$.


    To we prove the reconstruction guarantee, let us define the following notations for brevity:
    \begin{itemize}
    \item $\Theta^{\initial}$ denotes the event that $\|\bz^{(0)}\|^2 \leqs n\sqrt{d/(\varepsilon b)}\|\by^{(0)}\|^2$.
    \item $\Theta^{\stored}$ denotes the event that $\E_{t \in [0.4n\tau^{\stored}, 0.6n\tau^{\stored}]} |\oR^\eta_t| \leqs n\sqrt[4]{\varepsilon b / d}$.
    \item $\Theta^{\tend}$ denotes the event that $\E_{t \in [0.4n\tau^{\tend}, 0.6n\tau^{\tend}]} |R^\eta_t| \leqs n\sqrt[4]{\varepsilon b / d}$.
    \item Let $a^{\stored} \triangleq \frac{\tau^{\stored}}{\log n},
        b^{\stored} \triangleq \frac{\ttau^{\stored}}{\log n}, a^{\tend}
        \triangleq \frac{\tau^{\tend}}{\log n}$ and $b^{\tend} \triangleq
        \frac{\ttau^{\tend}}{\log n}$.
    \item Let $\zeta \triangleq \max\left\{\sqrt{\frac{\varepsilon b}{d}}, \frac{1}{\log n}\right\}$.
    \item $\Theta^{\uni, \stored}$ denotes the event that at least $(1 -
        \sqrt{\zeta})n$ nodes are $(a^{\stored}, b^{\stored}, \zeta)$-uniform.
    \item $\Theta^{\uni, \tend}$ denotes the event that at least $(1 -
        \sqrt{\zeta})n$ nodes are $(a^{\tend}, b^{\tend}, \zeta)$-uniform.
    \item $\Theta$ denotes the event that $\Theta^{\initial}, \Theta^{\stored}, \Theta^{\tend}, \Theta^{\uni, \stored}$ and $\Theta^{\uni, \tend}$ all occur. 
    \end{itemize}
    Note that $\Theta$ here is the ``nice'' event, where the conditions required in Lemma~\ref{lem:tstored}, Lemma~\ref{lem:tend} and Lemma~\ref{lem:globaltolocal} are satisfied and we can invoke them. Our proof will proceed in two steps: we will first show that the probability that $\Theta$ occurs is large and, then, we will use our auxiliary lemmas to show that, conditioned on $\Theta$ happening, we achieve the desired reconstruction most of the time.

    To bound $\Pr_{\bx^{(0)}, \cE}\left[\neg \Theta\right]$, first note that,
    from Proposition~\ref{prop:starting-condition}, we have
    $\Pr_{\bx^{(0)}}[\neg \Theta^{\initial}] \leqs O(\sqrt[4]{\varepsilon b /
    d})$. Moreover, from Lemma~\ref{lem:tstored}, we have $$\E_{\bx^{(0)},
    \cE}\left[\E_{t \in [0.4n\tau^{\stored}, 0.6n\tau^{\stored}]} |\oR^\eta_t|
    \midv \Theta^{\initial}\right] \leqs O\left(n \sqrt{\varepsilon b /
    d}\right).$$ From Markov's inequality, this implies that $\Pr_{\bx^{(0)},
    \cE}[\neg \Theta^{\stored} \mid \Theta^{\initial}] \leqs
    O\left(\sqrt[4]{\varepsilon b / d}\right)$. Similarly, Lemma~\ref{lem:tend}
    implies that $\Pr_{\bx^{(0)}, \cE}[\neg \Theta^{\tend} \mid \Theta^{\initial}]
    \leqs O\left(\sqrt[4]{\varepsilon b / d}\right)$. Now, note that, since $\zeta
    \geqs 1 / \log n$ and $a^{\stored} \geqs 1$, we have that $\sqrt{\zeta} n \geqs
    n^{1 - \Omega(\sqrt{\zeta} a^{\stored})}$ for sufficiently large $n$. Hence, we
    can apply Lemma~\ref{lem:manystandardvertices}, which implies that
    $\Pr_{\cE}[\neg \Theta^{\uni, \stored}] \leqs n^{-\Omega(\sqrt{\zeta} a)} \leqs
    O(\sqrt{\zeta})$. Similarly, we also have $\Pr_{\cE}[\neg \Theta^{\uni, \tend}]
    \leqs O(\sqrt{\zeta})$. Finally, by combining these four bounds, we have
    \begin{align}
    \Pr[\neg \Theta] 
    &\leqs \Pr[\neg \Theta^{\initial}] + \Pr[\neg \Theta^{\stored} \mid
        \Theta^{\initial}] \nonumber\\
    &\qquad+ \Pr[\neg \Theta^{\tend} \mid \Theta^{\initial}] + \Pr[\neg
        \Theta^{\uni, \stored}]  + \Pr[\neg \Theta^{\uni, \tend}] \nonumber \\
    &\leqs O(\sqrt{\zeta}). \label{eq:nicebound}
    \end{align}

    We now proceed to the second part of the proof. Let us denote the set of
    nodes incorrectly labeled by $V^{\incorrect}$, i.e., $V^{\incorrect} =
    \left\{u \in V \mid \chi_u\left(\bx_{T_u(\tau_u^{\stored})} -
    \bx_{T_u(\tau_u^{\tend})}\right) < 0\right\}$. We will show that 
    \begin{align} 
        \label{eq:expected-correct}
        \E_{\bx^{(0)}, \cE, \{\tau_u^{\stored}\}_{u \in V},
        \{\tau_u^{\tend}\}_{u \in V}} \left[|V^{\incorrect}| \midv
        \Theta\right] \leqs O\left(n\sqrt{\zeta}\right).
    \end{align}
    Before we prove the above inequality, let us first show how this implies
    the desired reconstruction property. By applying Markov's inequality to
    (\ref{eq:expected-correct}), we can conclude that $\Pr
    \left[|V^{\incorrect}| \geqs n\sqrt[4]{\zeta} / 2 \midv \Theta\right] \leqs
    O(\sqrt[4]{\zeta})$. From this and from (\ref{eq:nicebound}), we can
    conclude that
    $\Pr\left[|V^{\incorrect}| \geqs n\sqrt[4]{\zeta} / 2\right] \leqs
    O(\sqrt[4]{\zeta})$. Note that, when $|V^{\incorrect}| \leqs
    n\sqrt[4]{\zeta} / 2$, the protocol achieves a $\sqrt[4]{\zeta}$-weak
    reconstruction. Hence, with probability $1 - O(\sqrt[4]{\zeta})$, our
    protocol achieves a $\sqrt[4]{\zeta}$-weak reconstruction of the graph.
    Since $\zeta = \max\{\sqrt{\varepsilon b/d}, 1/\log n\}$, this indeed
    implies the reconstruction property as stated in
    Theorem~\ref{thm:main-reconstruct}.

    Finally, let us next prove (\ref{eq:expected-correct}) and complete our
    proof of Theorem~\ref{thm:main-reconstruct}. Since $\{\tau_u^{\stored}\}_{u
    \in V}, \{\tau_u^{\tend}\}_{u \in V}$ are independent of $\Theta$, we can
    write the left-hand side of (\ref{eq:expected-correct}) as
    \begin{align*}
        & \E_{\bx^{(0)}, \cE, \{\tau_u^{\stored}\}_{u \in V}, \{\tau_u^{\tend}\}_{u
            \in V}} \left[|V^{\incorrect}| \midv \Theta\right] \\
        & = \E_{\bx^{(0)}, \cE, \{\tau_u^{\stored}\}_{u \in V},
            \{\tau_u^{\tend}\}_{u \in V}} \left[n \cdot \Pr_{u \in V} [u \in
            V^{\incorrect}] \midv \Theta\right] \\
        & = \E_{\bx^{(0)}, \cE} \left[n \cdot \Pr_{u \in V, \tau_u^{\stored},
            \tau_u^{\tend}} [u \in V^{\incorrect}] \midv \Theta\right].
    \end{align*}

    Hence, it suffices for us to show that, assuming that $\Theta$ happens, $\Pr_{u \in V, \tau_u^{\stored}, \tau_u^{\tend}} [u \in V^{\incorrect}] \leqs O(\sqrt{\zeta})$. To prove this, first observe that if $u \in V^{\incorrect}$, then either $u \in \oR^\eta_{T_u(\tau^{\stored}_u)}$ or $u \in R^\eta_{T_u(\tau^{\tend}_u)}$ or both. Thus, we have
    \begin{align*}
    \Pr_{u \in V, \tau_u^{\stored}, \tau_u^{\tend}} [u \in V^{\incorrect}] \leqs \Pr_{u \in V, \tau_u^{\stored}} [u \in \oR^\eta_{T_u(\tau^{\stored}_u)}] + \Pr_{u \in V, \tau_u^{\tend}} [u \in R^\eta_{T_u(\tau^{\tend}_u)}].
    \end{align*}
    Since we assume that $\Theta$ occurs, $\Theta^{\uni, \stored}$ also occurs, which means that we can apply Lemma~\ref{lem:globaltolocal} for the sequence $S_t = \oR^{\eta}_t, a = a^{\stored}$ and $b = b^{\stored}$. This implies that
    \begin{align*}
    \Pr_{u \in V, \tau_u^{\stored}} [u \in \oR^\eta_{T_u(\tau^{\stored}_u)}] \leqs O(\sqrt{\zeta}).
    \end{align*}
    Similarly, applying Lemma~\ref{lem:globaltolocal} with $S_t = R^{\eta}_t, a = a^{\tend}$ and $b = b^{\tend}$, we also have
    \begin{align*}
    \Pr_{u \in V, \tau_u^{\stored}} [u \in R^\eta_{T_u(\tau^{\tend}_u)}] \leqs O(\sqrt{\zeta}).
    \end{align*}
    By combining the above three inequalities, we indeed arrive at the desired bound.
\end{proof}

\subsection{Bounding $|R^\eta_t|$ and $|\oR^{\eta}_t|$: Proofs of Lemma~\ref{lem:tstored} and Lemma~\ref{lem:tend}} \label{subsec:bound-mu-correlated}

\begin{proof}[Proof of Lemma~\ref{lem:tstored}]
Consider any $t \in [0.4n\tau^{\stored}, 0.6n\ttau^{\stored}]$. For brevity, let $\Theta$ denote an event that $a_y(t) \in [0.5\mu(t), 1.5\mu(t)]$ and $\|\bz(t)\| \leqs \left(0.5\sqrt[4]{\varepsilon b / d}\right) \mu(t)$. Applying Theorem~\ref{thm:main} with $\beta = \sqrt{\frac{d}{\varepsilon b}}$, we have $\Pr_\cE[\Theta] \geqs 1 - O(\sqrt{\varepsilon b / d})$.

Now, let us bound the size of $\oR^{\eta}_t$ conditioned on $\Theta$ happening. To do so, first recall that, since $\bx(t) = a_{||} \cdot (\bone / \sqrt{n}) + a_y(t) \cdot (\chi / \sqrt{n}) + \bz(t)$, we have $\bx_u(t) = a_{||} + \frac{a_y(t)\chi_u}{\sqrt{n}} + \bz_u(t)$. Hence, we have
\begin{align*}
(\bx_u(t) - a_{||})\chi_u - \eta = \frac{a_y(t)}{n} + \chi_u\bz_u(t) - \eta \geqs \frac{0.25 \mu(t)}{n} + \chi_u\bz_u(t).
\end{align*}
where the inequality comes from $a_y(t) \geqs 0.5\mu(t)$ and from $\mu(t) \geqs \mu(0.6n\ttau^{\stored})$. This inequality implies that, if $u \in \oR^{\eta}_t$, then $|\bz_u(t)| > \frac{0.25 \mu(t)}{\sqrt{n}}$. As a result, we can conclude that
\begin{align*}
|\oR^\eta_t| < \frac{\|\bz(t)\|^2}{(0.25\mu(t) / \sqrt{n})^2} \leqs O\left(n\sqrt{\varepsilon b / d}\right)
\end{align*}
where the latter comes from $\|\bz(t)\| \leqs \left(0.5\sqrt[4]{\varepsilon b / d}\right)\mu(t)$.

Thus, we have
\begin{align*}
\E_\cE |\oR^\eta_t|
= \E_\cE [|\oR^\eta_t| \mid \Theta]\Pr[\Theta] + \E_\cE [|\oR^\eta_t| \mid \neg \Theta]\Pr[\neg \Theta]
\leqs O(n\sqrt{\varepsilon b / d}) \cdot 1 + n \cdot O(\sqrt{\varepsilon b / d})
\leqs O(n\sqrt{\varepsilon b / d})
\end{align*}
as desired.
\end{proof}

The proof of Lemma~\ref{lem:tstored} is analogous to the above proof and is presented below.

\begin{proof}[Proof of Lemma~\ref{lem:tend}]
    Consider any $t \in [0.4n\tau^{\tend}, 0.6n\ttau^{\tend}]$. First of all, let us note that
    \begin{align} \label{eq:mutoeta}
    \frac{\mu(t)}{\eta} 
    = 4\left(1 - \frac{4\delta b}{d n}\right)^{t - 0.6n\ttau^{\stored}}
    \leqs 4\left(1 - \frac{4\delta b}{d n}\right)^{\frac{4dn}{\delta b}}
    \leqs 4 \cdot e^{-16}
    \leqs 0.25
    \end{align}
    where the second inequality comes from the fact that $1 + x \leqs e^x$ for all $x \in \R$.

    The rest of the proof proceeds similar to the proof of Lemma~\ref{lem:tstored}. Again, let $\Theta$ denote an event that $a_y(t) \in [0.5\mu(t), 1.5\mu(t)]$ and $\|\bz(t)\| \leqs \left(0.5\sqrt[4]{\varepsilon b / d}\right) \mu(t)$. From Theorem~\ref{thm:main} with $\beta = \sqrt{\frac{d}{\varepsilon b}}$, we have $\Pr_\cE[\Theta] \geqs 1 - O(\sqrt{\varepsilon b / d})$.

    Conditioned on $\Theta$, observe that
    \begin{align*}
    (\bx_u(t) - a_{||})\chi_u - \eta = \frac{a_y(t)}{n} + \chi_u\bz_u(t) - \eta \leqs - \frac{0.5 \mu(t)}{n} + \chi_u\bz_u(t).
    \end{align*}
    where the inequality comes from $a_y(t) \leqs 1.5\mu(t)$ and (\ref{eq:mutoeta}). Hence, if $u \in R^\eta_t$, then $|\bz_u(t)| > \frac{0.5\mu(t)}{\sqrt{n}}$. Since $\|\bz(t)\| \leqs \left(0.5\sqrt[4]{\varepsilon b / d}\right)\mu(t)$, this implies that $|R^\eta_t| < \frac{\|\bz(t)\|^2}{(0.5\mu(t) / \sqrt{n})^2} \leqs n\sqrt{\varepsilon b / d}$.

    Thus, we have
    \begin{align*}
    \E_\cE |R^\eta_t|
    = \E_\cE [|R^\eta_t| \mid \Theta]\Pr[\Theta] + \E_\cE [|R^\eta_t| \mid \neg \Theta]\Pr[\neg \Theta]
    \leqs n\sqrt{\varepsilon b / d} + n \cdot O(\sqrt{\varepsilon b / d})
    \leqs O(n\sqrt{\varepsilon b / d})
    \end{align*}
    as desired.
\end{proof}

\subsection{Most Vertices are Uniform: Proof of Lemma~\ref{lem:manystandardvertices}} \label{app:standard}

\begin{proof}[Proof of Lemma~\ref{lem:manystandardvertices}]
    Let us fix a vertex $u \in V$. We will compute the probability that $u$ is
    $(a, b, \zeta)$-uniform. First, we will bound the probability that the
    first condition is not satisfied. To do so, let us introduce an additional
    notation; we use $X_t$ to denote an indicator variable of the event $u \in
    \{u_t, v_t\}$. Note each $X_t$ is an i.i.d. Bernoulli random variable which
    is one with probability $2/n$. The probability that $T_u(a \log n) <
    0.4an\log n$ can now be written in terms of $X_t$'s as follows.
    \begin{align} 
    \Pr_{\cE}[T_u(a \log n) \leqs 0.4an\log n]
    &= \Pr_{\cE}[X_1 + \cdots + X_{0.4an\log n} \geqs a \log n] 
    \leqs 2^{-\Omega(a \log n)}
        \label{eq:standard1}
    \end{align}
    where the inequality comes from an application of Chernoff bound. Similarly, we get the following bound for $\Pr_{\cE}[T_u(b \log n + 1) > 0.6b \log n]$:
    \begin{align} \label{eq:standard2}
    \Pr_{\cE}[T_u(b \log n + 1) < 0.6b \log n] = \Pr_{\cE}[X_1 + \cdots + X_{0.6bn\log n} \leqs b \log n] \leqs 2^{-\Omega(b \log n)}.
    \end{align}

    Next, we proceed to bound the probability that the second condition fails. Let $Y_\tau$ denote an indicator variable of the event $T_u(\tau + 1) < T_u(\tau) + \sqrt{\zeta} n$. Observe that $Y_\tau$'s are i.i.d. Moreover, the probability that $Y_\tau = 1$ can be bounded as follows.
    \begin{align*}
    \Pr[Y_\tau = 1] = \Pr\left[\bigwedge_{i = T_u(\tau)}^{T_u(\tau) + \sqrt{\zeta} n - 1} u \in \{u_t, v_t\}\right]
    \leqs \sum_{i = T_u(\tau)}^{T_u(\tau) + \sqrt{\zeta} n - 1} \Pr\left[u \in \{u_t, v_t\}\right]
    = 2 \sqrt{\zeta}.
    \end{align*}

    Hence, by Chernoff bound, we have
    \begin{align} \label{eq:standard3}
    \Pr\left[\sum_{\tau = a \log n}^{b \log n} Y_\tau \leqs 4\sqrt{\zeta}(b - a)\log n\right] &\leqs 2^{-\Omega(\sqrt{\zeta}(b - a)\log n)}.
    \end{align}

    Observe that $\sum_{\tau = a \log n}^{b \log n} Y_\tau \leqs
    4\sqrt{\zeta}(b - a)\log n$ is equivalent to 
    \[
        \Pr_{\tau \in [a \log n, b \log n]} \left[T_u(\tau + 1) < T_u(\tau) +
        \sqrt{\zeta} n\right] \leqs 4\sqrt{\zeta}.
    \]
    Thus, (\ref{eq:standard1}), (\ref{eq:standard2}) and
    (\ref{eq:standard3}) together with the fact that $b = 2a$ imply that the
    probability that $u$ is $(a, b, \zeta)$-standard is at least $1 -
    2^{-\Omega(\sqrt{\zeta} a \log n)}$, which is at least $1 -
    n^{-C\sqrt{\zeta} a}$ for some global constant $C$.

    As a result, the expected number of vertices that are not $(a, b,
    \zeta)$-uniform is at most $n^{1 - C\sqrt{\zeta} a}$. Hence, an application
    of Markov's inequality implies that, with probability at most
    $n^{C\sqrt{\zeta} a / 2}$, the number of non-uniform vertices is at most
    $n^{1 - C\sqrt{\zeta}a/2}$, which concludes the proof of this lemma.
\end{proof}

\subsection{From Global to Local Time: Proof of Lemma~\ref{lem:globaltolocal}} \label{subsec:global-to-local}

In this subsection, we present the proof of Lemma~\ref{lem:globaltolocal}.

\begin{proof}[Proof of Lemma~\ref{lem:globaltolocal}]
Let $St \subseteq V$ denote the set of $(a, b, \zeta)$-uniform vertices. We can first write the left-handside term so that we separate out the uniform $u$'s from the non-uniform ones as follows.
\begin{align*}
\Pr_{u \in V, \tau \in [a\log n, b\log n]}\left[u \in S_{T_u(\tau)}\right] 
&\leqs \Pr_{u \in V, \tau \in [a\log n, b\log n]}\left[u \in St \wedge u \in S_{T_u(\tau)}\right] + \Pr_{u \in V}[u \notin St] \\
&\leqs \Pr_{u \in V, \tau \in [a\log n, b\log n]}\left[u \in St \wedge u \in S_{T_u(\tau)}\right] + \sqrt{\zeta}
\end{align*}
where the last inequality comes from our assumption that there are only $\sqrt{\zeta}n$ non-uniform vertices. For each vertex $u \in V$, denote the set of $\tau \in [a \log n, b \log n]$ such that $T_u(\tau + 1) - T_u(\tau) < \sqrt{\zeta} n$ by $R_u$. We can further bound the term $\Pr_{u \in V, \tau \in [a\log n, b\log n]}\left[u \in St \wedge u \in S_{T_u(\tau)}\right]$ by
\begin{align*}
\Pr_{u \in V, \tau \in [a\log n, b\log n]}\left[u \in St \wedge u \in S_{T_u(\tau)}\right]
&\leqs \Pr_{u \in V, \tau \in [a\log n, b\log n]}\left[u \in St \wedge \tau \notin R_u \wedge u \in S_{T_u(\tau)}\right] \\
&\text{ } + \Pr_{u \in V, \tau \in [a\log n, b\log n]}\left[u \in St \wedge \tau \in R_u\right] \\
&\leqs \Pr_{u \in V, \tau \in [a\log n, b\log n]}\left[u \in St \wedge \tau \notin R_u \wedge u \in S_{T_u(\tau)}\right] + O(\sqrt{\zeta})
\end{align*}
where the second inequality comes from the fact that, if $u$ is $(a, b,
\zeta)$-uniform, then $\Pr_{\tau \in [a\log n, b\log n]}\left[\tau \in
R_u\right] \leqs 4\sqrt{\zeta}$. Hence, to prove the intended inequality, it
suffices to show that 
\[
    \Pr_{u \in V, \tau \in [a\log n, b\log n]}\left[u \in St
        \wedge \tau \notin R_u \wedge u \in S_{T_u(\tau)}\right] \leqs \sqrt{\zeta}.
\]
Observe that this probability can be further rearranged as
\begin{align} \label{eq:boundlastterm}
    &\Pr_{u \in V, \tau \in [a\log n, b\log n]}\left[u \in St \wedge \tau
        \notin R_u \wedge u \in S_{T_u(\tau)}\right] \nonumber \\
    &= \frac{1}{n(b \log n - a \log n + 1)}\left(\sum_{u \in V} \sum_{\tau \in
        [a\log n, b\log n]} \ind\left[u \in St \wedge \tau \notin R_u \wedge u \in
        S_{T_u(\tau)}\right] \right) \nonumber \\
    &\leqs \frac{1}{na\log n}\left(\sum_{u \in St} \sum_{\tau \in [a\log n,
        b\log n] \setminus R_u} \ind\left[u \in S_{T_u(\tau)}\right]\right).
\end{align}
Let us fix $u \in St$ and $\tau \in [a\log n, b\log n] \setminus R_u$. Recall that, since $\{S_t\}_{t \in \N}$ is compatible with $\cE$, we have $\ind[u \in S_{t}] = \ind[u \in S_{T_u(\tau)}]$ for every $t \in [T_u(\tau), T_u(\tau + 1))$. Moreover, because $\tau \notin R_u$, we have $T_u(\tau + 1) - T_u(\tau) \geqs \sqrt{\zeta} n$. Thus, we have
\begin{align*}
 \ind\left[\tau \notin R_u \wedge u \in S_{T(\tau)}\right] &= \ind\left[\tau \notin R_u\right]\ind\left[u \in S_{T_u(\tau)}\right] \\
 &= \ind\left[\tau \notin R_u\right]\left(\frac{\sum_{t = T_u(\tau)}^{T_u(\tau + 1) - 1} \ind\left[u \in S_t\right]}{T_u(\tau + 1) - T_u(\tau)}\right) \\
 &\leqs \ind\left[\tau \notin R_u\right]\left(\frac{\sum_{t = T_u(\tau)}^{T_u(\tau + 1) - 1} \ind\left[u \in S_t\right]}{\sqrt{\zeta} n}\right) \\
 &\leqs \frac{\sum_{t = T_u(\tau)}^{T_u(\tau + 1) - 1} \ind\left[u \in S_t\right]}{\sqrt{\zeta} n}.
\end{align*}

Plugging the above inequality back into (\ref{eq:boundlastterm}), the
probability
\[
    \Pr_{u \in V, \tau \in [a\log n, b\log n]}\left[u \in St \wedge \tau \notin R_u \wedge u \in S_{T_u(\tau)}\right]
\]
can be upper bounded by
\begin{align*}
&\frac{1}{\sqrt{\zeta}n^2 a \log n} \left(\sum_{u \in St} \sum_{\tau \in [a\log n, b\log n] \setminus R_u} \sum_{t = T_u(\tau)}^{T_u(\tau + 1) - 1} \ind\left[u \in S_t\right]\right) \\
&\leqs \frac{1}{\sqrt{\zeta}n^2 a \log n} \left(\sum_{u \in St} \sum_{\tau \in [a\log n, b\log n]} \sum_{t = T_u(\tau)}^{T_u(\tau + 1) - 1} \ind\left[u \in S_t\right]\right) \\
&= \frac{1}{\sqrt{\zeta}n^2 a \log n} \left(\sum_{u \in St} \sum_{t = T_u(a \log n)}^{T_u(b \log n + 1) - 1} \ind\left[u \in S_t\right]\right).
\end{align*}
Finally, recall from definition of uniform vertices that, if $u$ is $(a, b, \zeta)$-uniform, then $T_u(a \log n) > 0.4an \log n$ and $T_u(b \log n + 1) \leqs 0.6 b \log n$. Combining this with the above inequality, we have
\begin{align*}
\Pr_{u \in V, \tau \in [a\log n, b\log n]}\left[u \in St \wedge \tau \notin R_u \wedge u \in S_{T_u(\tau)}\right] &\leqs \frac{1}{\sqrt{\zeta}n^2 a \log n} \left(\sum_{u \in St} \sum_{t = 0.4 an\log n}^{0.6 b \log n - 1} \ind\left[u \in S_t\right]\right) \\
&\leqs \frac{1}{\sqrt{\zeta}n^2 a \log n} \left(\sum_{u \in V} \sum_{t = 0.4 an\log n}^{0.6 b \log n - 1} \ind\left[u \in S_t\right]\right) \\
&= \frac{1}{\sqrt{\zeta}n^2 a \log n} \left(\sum_{t = 0.4 an\log n}^{0.6 b \log n - 1} \sum_{u \in V} \ind\left[u \in S_t\right]\right) \\
&= \frac{1}{\sqrt{\zeta}n^2 a \log n} \left(\sum_{t = 0.4 an\log n}^{0.6 b \log n - 1} |S_t|\right) \\
\left(\text{Since } \E_{t \in [0.4an\log n, 0.6 b \log n]} |S_t| \leqs \zeta n\right) &\leqs \frac{1}{\sqrt{\zeta}n^2 a \log n} \left((0.6 bn \log n - 0.4 an \log n + 1)\zeta n\right) \\
&= O(\sqrt{\zeta}),
\end{align*}
which concludes our proof.
\end{proof}

\end{document}